\documentclass[11pt]{article}
\usepackage{geometry}        
\usepackage[parfill]{parskip}  
\usepackage{graphicx}
\usepackage{amssymb}
\usepackage{epstopdf}
\usepackage{fancyhdr}
\usepackage{graphicx}
\usepackage{textcomp} 
\usepackage{graphicx}  
\usepackage{flafter}  
\usepackage{amsmath,amssymb,amsthm}  
\usepackage{bm}  
\usepackage{memhfixc}  
\usepackage[usenames]{color}
\usepackage{listings}
\usepackage{appendix}
 \usepackage[usenames,dvipsnames]{pstricks}
 \usepackage{epsfig}
 \usepackage{pst-grad} 
 \usepackage{pst-plot} 
 \usepackage{morefloats}
 \numberwithin{equation}{section}

\usepackage[colorlinks=true, pdfstartview=FitV, linkcolor=blue, 
            citecolor=blue, urlcolor=blue]{hyperref}

\title{A Bianchi Type IV Viscous Fluid Model of The Early Universe}
\author{Ikjyot Singh Kohli}
\begin{document}

\abstract{We are interested in formulating a viscous model of the universe based on The Bianchi Type IV algebra. We first begin by considering a congruence of fluid lines in spacetime, upon which, analyzing their propagation behaviour, we derive the famous Raychaudhuri equation, but, in the context of viscous fluids. We will then go through in great detail the topological and algebraic structure of a Bianchi Type IV algebra, by which we will derive the corresponding structure and constraint equations. From this, we will look at The Einstein field equations in the context of orthonormal frames, and derive the resulting dynamical equations: The \emph{Raychaudhuri Equation}, \emph{generalized Friedmann equation}, \emph{shear propagation equations}, and a set of non-trivial constraint equations. We show that for cases in which the bulk viscous pressure is significantly larger than the shear viscosity, this cosmological model isotropizes asymptotically to the present-day universe. We finally conclude by discussing The Penrose-Hawking singularity theorem, and show that the viscous universe under consideration necessarily emerged from a past singularity point.}
\newpage
\maketitle
\newpage
\tableofcontents
\newpage
\section{Introduction}
\indent For centuries past and into the present time, much of physics has been taught and done in the flat, Euclidean space, with time being considered as a dependent, constant variable.  Indeed, two such fields of physics have emerged. There is a \emph{local} physics, which involves formulating theories of local phenomena such as fluid dynamics, quantum mechanics, and electromagnetism. On the other hand, there is an \emph{universal} physics which attempts to describe the nature of the universe on a large-structure scale. While the former have been described extremely successfully by The Navier-Stokes equations, Schrodinger's Equation, and Maxwell's equations, only Einstein's Theory of General Relativity has been able to successfully, and completely describe the latter. In recent times, many attempts have been made with respect to combining the two into one grand perspective of describing the whole universe, from the very small to the very large. Einstein's theory conveniently allows for this by the notion that matter induces curvature in the spacetime manifold, and the resulting spacetime geometry prescribes the motion of of this matter. In the context of this description for how our universe really works, we attempt to formulate a theory of viscous fluids in a spatially homogeneous spacetime. This is not as easy as simply applying the Navier-Stokes equations. Despite their overwhelming success in the classical theory of fluid dynamics, the Navier-Stokes equations only apply in Euclidean space. The main question is that of the classification of partial differential equations. In theory, since the Navier-Stokes equations are a coupled set of \emph{hyperbolic} partial differential equations, and therefore, constitute \emph{causal} solutions, one can map these over to the pseudo-Riemannian manifold of General Relativity by the usual partial to covariant derivative rule. However, in the case of viscosity, the Navier-Stokes equations become \emph{elliptical} partial differential equations, and therefore lead to \emph{acausal} solutions. This is not allowed in any relativistic theory, since this would technically imply that fluid propagation speeds are able to exceed the speed of light. Therefore, other methods must be used to describe the viscous fluid flow in a General Relativistic context. We are particularly interested in viscous/imperfect fluids as compared to inviscid/perfect fluids for the sole reason that inviscid flow is an idealization. In the context of cosmology, many astrophysical fluids such as gases exhibit high compressibility, where viscosity must be included in any fluid analysis for more realistic results. In addition, to the best of the author's knowledge at the time of the writing of this thesis, viscous flows in the Bianchi Type-IV universe have not been described, so this thesis also is an opportunity to present some new and valuable information. Perhaps the most important application is to \emph{early-universe} cosmology. The early universe consisted of high-temperature fluids, which necessarily means that viscosity must be included to appropriately describe any fluid matter. In addition, the early universe is taken to be spatially homogeneous and anisotropic, which somehow isotropized to the current FLRW models of the present-day universe. \\
\indent We will proceed in the following manner: We will first derive the viscous analogue of the all-important Raychaudhuri equation which describes the expansion/contraction behaviour of a congruence of fluid lines, and necessarily the universe. After discussing some qualitative solutions to this equation, we solve The Einstein Field equations for a Bianchi Type-IV algebra with a viscous fluid matter source using the approach of orthonormal frames.The point of this approach is that the field equations which typically constitute ten coupled, non-linear, hyperbolic partial differential equations reduce to a system of autonomous, first-order, coupled ordinary differential equations which are easier to solve using numerical methods. We will conclude with a discussion of The Penrose-Hawking singularity theorem applied to a viscous universe. We assume that the reader has adequate knowledge of general relativity and differential geometry throughout. In addition, we use geometrized units, such that $c = G = 1$. In addition, as typical in cosmology settings, we will use comoving coordinates, such that $u^{a} = (1,0,0,0)$, the metric signature $(-1,1,1,1)$, which implies that $u^{a}u_{a} = -1$.


\section{The Notion of Spatial Homogeneity}
Standard cosmological models have for the most part taken the universe to be at a minimum, spatially homogeneous. The historical reasons for this have been well- described in the literature. The interested reader should see for example \cite{ellisobserve} and \cite{mtwch30}. In terms of the timeline of our universe, cosmological models are essentially divided up into three categories. The first category consists of cosmological models that describe the immediate post-Big Bang universe, that is the universe on Planck scales. Such models are the typical quantum gravity models based on the concept of Misner's minisuperspace \cite{misnerb} \cite{misnerc} and The Hamiltonian formulation of General Relativity given by Arnowitt, Deser, and Misner \cite{ADM} , and Dirac \cite{dirac}. There are of course others that have contributed to this subject, however, the main work seems to be accredited to ADM and Dirac. The second category consists of models of the early universe, that is, the post-Big Bang universe considered on scales much larger than Planck scales such that classical General Relativity can be applied. The key point of these models is that they are taken to be spatially homogeneous, but since it is widely assumed that this stage of the universe still had much anisotropy in terms of cosmic radiation, these models are taken to be anisotropic, and hence conveniently characterized by the Bianchi classifications. Much work has been done in this area as well. For example, Singh and Baghel have recently studied Bianchi Type V bulk viscous cosmological models with a time dependent cosmological constant. \cite{singhbaghel} Misner famously studied the Bianchi Type IX/Mixmaster models which describe a chaotic decay of the anisotropy of the universe. \cite{misner} In addition, Belinskii, Khalatnikov, and Lifshitz have studied in detail the oscillatory approach to a singular point in the relativistic cosmology. \cite{lifshitz} The third category consists of the present-day models of the universe which are entirely described by The Friedmann-LeMatire-Robertson-Walker (FLRW) models with a perfect fluid matter source. Not only are these models spatially homogeneous, but they are also isotropic about every point making them a special class of the more general Bianchi models. One of the major areas of research in cosmology is to understand how the anisotropies of the early universe decayed to the present-day state of isotropy. This problem has been looked at in quite some detail as well. The interested reader is encouraged to read \cite{gutzwiller}, \cite{berger}, \cite{barrow}, especially \cite{hobbill}, \cite{chitre}, and \cite{misnerc2}. In addition, Salucci and Fabbri \cite{salucci} have in great detail analyzed the cosmological evolution of general Bianchi models in the adiabatic regime. They showed that a constant or slowly varying energy-momentum tensor isotropized the Hubble expansion. As alluded to earlier, our area of interest for this thesis is to develop a model of the early universe with a viscous fluid matter source. This model is unique in that the Bianchi Type-IV model we will consider has not been described in the literature. We have so far said a great deal about spatial homogeneity and (an)isotropy. Therefore, it is appropriate to describe what these terms mean in a more precise, mathematical context. 

\subsection{The Mathematical Foundations of Spatial Homogeneity}
In this section, we will build the mathematical foundations for spatial homogeneity. The arguments we present here are largely based on the descriptions given by Hervik and Gr{\o}n \cite{hervik}, Ellis and MacCallum \cite{ellismac}, Eisenhart \cite{eisenhart} \cite{eisenhart2}, and Stephani et.al \cite{stephani2}. The interested reader is encouraged to refer to these sources for further details. However, we attempt to provide as comprehensive of a description as possible. Basically, a space is homogeneous if it admits a group of motions. That is, a space is topologically homogeneous if you can carry one point to any other point via an isometry. If we consider a space $M$ on the manifold described by the metric tensor $\mathbf{g}$, we can then define an isometry group as:

\begin{equation}
Isom(M) \equiv \{\phi: M \rightarrow M|\phi \mbox{ isometry} \}
\end{equation}

An isometry operation on the metric is one that leaves the metric unchanged, that is, $\phi*\mathbf{g} = \mathbf{g}$. From classical Riemannian geometry, we know that Killing vector fields generate isometries. We can therefore say that a space is homogeneous if for each pair of points $(p,q) \in M, \exists \phi \in Isom(M) \mbox{ such that } \phi(p) = q$. In some literature \cite{fields}, homogeneous spaces are referred to as transitive. 

We can extend this definition a bit further. We could define an isotropy subgroup of a point $ p \in M$ by:

\begin{equation}
S_{p}(M) = \{ \phi \in Isom(M) | \phi(p) = p\}
\end{equation}

Intuitively, this definition is providing for a mechanism for spatially homogeneous spacetimes that are also isotropic. By leaving one point fixed as the above definition implies, we are allowing for the possibility of rotations. We shall see the importance of this subgroup in due course.

If we further let $dim\{Isom(M)\} = n$, and $dim\{S_{p}(M)\} = m$, a space is homogeneous if $n \geq dim\{M\}$. It is said that $M$ is simply transitive if it is homogeneous and $n = dim\{M\}$. $M$ is said to be multiply transitive if $M$ is homogeneous, and $n \geq dim\{M\}$.  

In addition, if we consider for a point $p \in M$ in a subspace of $M$ such that:

\begin{equation}
H_{p} = \{q \in M | q = \phi(p) \mbox{ for a } \phi \in Isom(M)\}
\end{equation}

This subspace is known as the orbit of $p$ under the isometry group. Going back to the idea of moving from one point to another in a space, all points that we can reach via the action of the isometry group on $p$, is necessarily in the orbit of $p$. In addition, if the orbit of $p$ is actually the whole space, that is, $H_{p} = M$, then space is homogeneous/transitive.

In a simply transitive space, there always exists a basis of Killing vectors $\mathbf{\xi}_{i}$, such that:

\begin{equation}
\left[\mathbf{\xi}_{i}, \mathbf{\xi}_{j}\right] = F^{k}_{ij}\mathbf{\xi}_{k}
\end{equation}

This indeed is the main method described by Stephani \cite{stephani2} and Taub \cite{taub}. However, this seems slightly inefficient as to get an expression for the metric tensor involves expressing these Killing vector basis vectors in canonical form, that is, in a coordinate basis. Except for trivial cases, this leads to complicated expressions for the Ricci tensor that are undesirable to work with when solving The Einstein field equations. We therefore take another approach by noting that we can define another set of basis vectors as follows: At a point $p \in M$, we choose a basis $\mathbf{e}_{i}$ and define what is known as a \emph{left-invariant frame} by Lie transporting it around the space in the manifold. That is, by the definition of the Lie derivative:

\begin{equation}
\mathcal{L}_{\mathbf{\xi}_{j}} \mathbf{e}_{i} = \left[\mathbf{\xi}_{j}, \mathbf{e}_{i}\right] = 0
\end{equation}

In general, a frame with basis vectors $\mathbf{e}_{j}$ defines a left-invariant frame, while the frame with Killing vectors $\mathbf{\xi}_{i}$ as the basis vectors defines a right-invariant frame. In addition, applying the general property:

\begin{equation}
\mathcal{L}_{\mathbf{\xi}_j} \left[\mathbf{e}_{i}, \mathbf{e}_{k}\right] = 0
\end{equation}

We see that the basis vectors $\mathbf{e}_{i}$ span a Lie algebra. Therefore, for structure constants $C^{k}_{ij}$, they obey the general property:

\begin{equation}
\left[\mathbf{e}_{i}, \mathbf{e}_{j}\right]  = C^{k}_{ij} \mathbf{e}_{k}
\end{equation}

We will therefore construct a homogeneous space by first defining a left-invariant frame as:

\begin{equation*}
\left[\mathbf{e}_{i}, \mathbf{e}_{j}\right]  = C^{k}_{ij} \mathbf{e}_{k}
\end{equation*}

By definition, if $\omega^{k}$ is the dual basis to $\mathbf{e}_{k}$, then, necessarily:

\begin{equation}
d\mathbf{\omega}^{k} = -\frac{1}{2} C^{k}_{ij} \mathbf{\omega}^{i} \wedge \mathbf{\omega}^{j}
\end{equation}

These invariant one-forms can be used in turn to define a metric (which is also left-invariant) as:

\begin{equation}
ds^2 = g_{ij} \mathbf{\omega}^{i} \otimes \mathbf{\omega}^{j}
\end{equation}

This is a spatially homogeneous metric on the space. It is interesting to note that the Killing vectors of the metric are given by $\mathbf{\xi}_{i}$, while, the basis vectors $\mathbf{e}_{j}$ constitute a left-invariant frame.

\subsection{The Bianchi Classifications}
So far, we have considered the notion of spatial homogeneity in a general sense. We obviously would like to apply the aforementioned concepts to general relativity, and specifically, cosmology. We begin by assuming that four-dimensional spacetime can be foliated with three-dimensional spatial sections. That is, given a spacetime manifold $M$ , we have:

\begin{equation}
M = \mathbb{R} \times \Sigma_{t}
\end{equation}

Here, $\mathbb{R}$ represents the space of time vectors that are moving through the spatial slices, and $\Sigma_{t}$  is a homogeneous space of dimension three that is unique at each time $t_{i}$. We have attempted to illustrate this concept in Fig. (2.1).

\begin{figure}[h]
\begin{center}
\caption{3+1 Decomposition into Spatial Slices Orthogonal to the Time Vector}
\fbox{
\scalebox{1} 
{
\begin{pspicture}(0,-3.7365909)(15.202162,3.7365909)
\psbezier[linewidth=0.04,fillstyle=solid](2.0202675,3.1829271)(2.8659642,3.716591)(6.5509152,3.3628833)(6.5602674,2.362927)(6.56962,1.3629707)(4.1266665,2.2022796)(3.2602675,1.702927)(2.3938687,1.2035744)(1.760535,-1.4715507)(0.8802675,-0.997073)(0.0,-0.5225953)(1.1745708,2.6492631)(2.0202675,3.1829271)
\usefont{T1}{ptm}{m}{n}
\rput(3.8817227,2.327927){$M \subseteq \mathbb{R}^{4} $}
\psline[linewidth=0.04cm,arrowsize=0.05291667cm 2.0,arrowlength=1.4,arrowinset=0.4]{->}(6.2002673,1.842927)(9.060267,1.082927)
\psline[linewidth=0.04cm,arrowsize=0.05291667cm 2.0,arrowlength=1.4,arrowinset=0.4]{<-}(11.460267,3.142927)(11.380267,-3.297073)
\usefont{T1}{ptm}{m}{n}
\rput(10.811723,3.287927){$\mathbb{R}$}
\psbezier[linewidth=0.04,fillstyle=solid](10.040268,-2.617073)(10.080691,-1.6178904)(12.891315,-1.717555)(12.860268,-2.717073)(12.829219,-3.716591)(12.359618,-2.7131069)(11.360268,-2.677073)(10.360917,-2.641039)(9.999844,-3.6162555)(10.040268,-2.617073)
\psbezier[linewidth=0.04,fillstyle=solid](10.040268,-1.1370729)(10.080691,-0.13789035)(12.891315,-0.23755509)(12.860268,-1.237073)(12.829219,-2.2365909)(12.359618,-1.233107)(11.360268,-1.197073)(10.360917,-1.161039)(9.999844,-2.1362555)(10.040268,-1.1370729)
\psbezier[linewidth=0.04,fillstyle=solid](10.0202675,0.32292703)(10.060691,1.3221097)(12.871316,1.2224449)(12.840267,0.22292702)(12.809219,-0.7765909)(12.339618,0.22689301)(11.340267,0.26292703)(10.340917,0.29896104)(9.979844,-0.6762556)(10.0202675,0.32292703)
\psbezier[linewidth=0.04,fillstyle=solid](10.080268,1.682927)(10.120691,2.6821096)(12.931315,2.582445)(12.900268,1.582927)(12.869219,0.58340913)(12.399618,1.586893)(11.400268,1.6229271)(10.400917,1.658961)(10.039844,0.6837444)(10.080268,1.682927)
\usefont{T1}{ptm}{m}{n}
\rput(14.111723,-2.832073){$\Sigma_{1}$}
\usefont{T1}{ptm}{m}{n}
\rput(14.111723,-1.2520729){$\Sigma_{2}$}
\usefont{T1}{ptm}{m}{n}
\rput(14.111723,0.107927024){$\Sigma_{3}$}
\usefont{T1}{ptm}{m}{n}
\rput(14.0917225,1.787927){$\Sigma_{4}$}
\end{pspicture} 
}
}
\label{manifold3+1}
\end{center}
\end{figure}
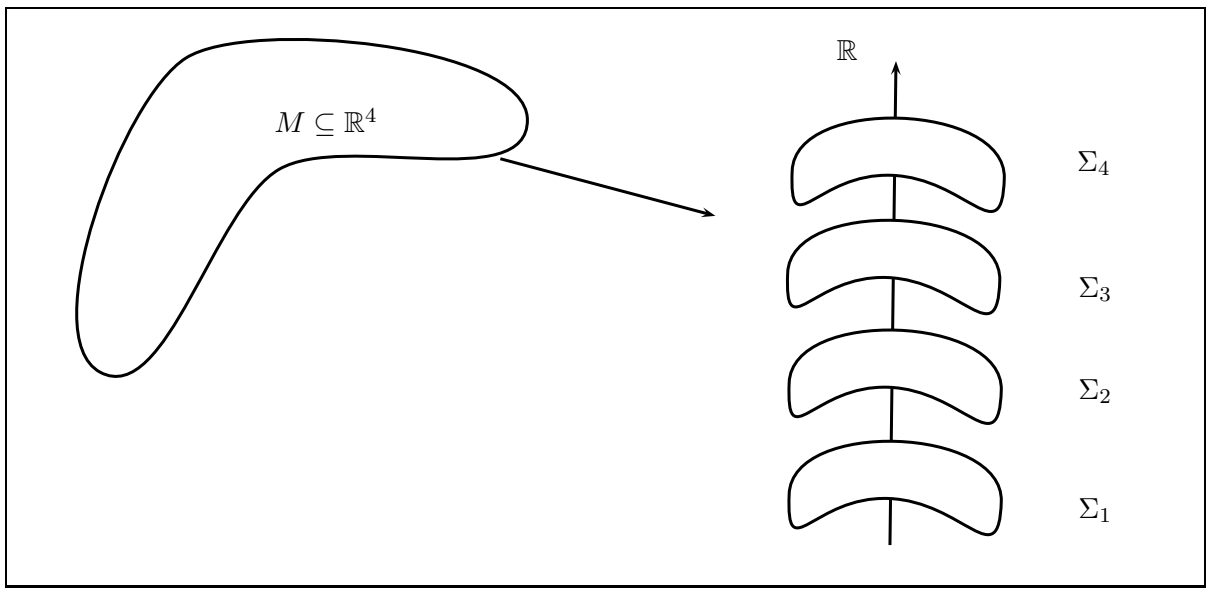

For a homogeneous space, there are two interesting cases with respect to cosmology. The first occurs for $dim\{Isom(M)\} = 6$, which means that $\Sigma_{t}$ is multiply transitive and is maximally transitive. That is, it represents a spatially homogeneous space that is translationally and rotationally invariant, or, it is spatially homogeneous and isotropic. The models of the present-day universe, the Friedmann-LeMa\^{\i}tre-Robertson-Walker models belong to this classification. With respect to early-universe cosmology, we are interested in cases where $dim\{Isom(M)\} = 3$, which implies that $\Sigma_{t}$  is simply transitive. That is, we are studying spacetimes that are only spatially homogeneous, neglecting the additional assumption of isotropy. The question seems to be then how many different classifications of Lie algebras can we have in three dimensions? Such classifications are called The Bianchi models, first developed by Luigi Bianchi \cite{Bianchi}. One typically begins by listing the Bianchi models in terms of their structure constants. In general, we can employ the Behr decomposition which enables us to write the structure constants as a decomposition of the symmetric and asymmetric parts:

\begin{equation}
C^{k}_{ij} = \epsilon_{ijl}n^{lk} + a_{l} \left(\delta^{l}_{i} \delta^{k}_{j} - \delta^{l}_{j} \delta^{k}_{i}\right)
\end{equation}

In this equation, $a_{i}$ is a vector belonging to the Lie algebra, for which we have the freedom to choose a basis such that $a_{i} = a \delta^{3}_{i}$. 
\begin{equation}
d\mathbf{\omega}^{k} = -\frac{1}{2} C^{k}_{ij} \mathbf{\omega}^{i} \wedge \mathbf{\omega}^{j}
\end{equation}

One then writes the metric tensor as:

\begin{equation}
ds^2 = -dt^2 + g_{ij}(t) \mathbf{\omega}^{i} \otimes \mathbf{\omega}^{j}
\end{equation}

This approach is useful when considering the Hamiltonian/Lagrangian formulations of general relativity with respect to Quantum gravity, or gravitational chaos. Such approaches have been discussed in detail by Ryan and Shepley \cite{shepley}, Hobbill et.al \cite{hobbill}, Misner \cite{misnerc}, Burd \cite{burd}, and others. Now that we have cemented the geometric structure of a spatially homogeneous spacetime, we now apply these concepts to deriving a form of The Einstein field equations that will result in a first-order system of ordinary differential equations. First, however, we will need a way to describe the presence of viscous fluid matter in spacetime. We do this in the next chapter.

\section{The Mathematical Description of a Viscous Fluid Matter Source in Spacetime}
So far, we have discussed the topological formalism of a spatially homogeneous spacetime by introducing the concept of The Bianchi types.  This is obviously only one half of the story. To develop a general relativistic theory of the early universe, one also needs to introduce a matter source that couples to the curvature of spacetime. We are particularly concerned with viscous fluids in this thesis. It is believed that the early universe contained a matter source that was a high-temperature viscous fluid. \cite{brevik} \cite{brevik2} It is quite plausible that such a fluid had strong properties of turbulent and compressible flow.  Therefore, in order to account for these effects, it is necessary to include the effects of viscosity in any matter source description via the energy-momentum tensor.


\subsection{Energy-Momentum Tensor for a Viscous Fluid}
As per the aforementioned requirements, we need an expression for the energy momentum tensor $T_{ab}$ for the viscous fluid flow. One is already quite familiar with the energy-momentum tensor of a perfect,  non-viscous fluid:

\begin{equation}
T^{ab} = u^{a}u^{b}(\mu + p) - u^{c}u_{c}g^{ab}p 
\end{equation}
(Where $\mu$ represents the total mass-energy density of the fluid)

This comes from the idea of considering a pressure-less fluid:

\begin{equation}
T_{00} = \mu u^0 u^0 \Rightarrow T_{ab} = \mu u^{a}u^{b}
\end{equation}

Assuming there is a fluid pressure, from elementary fluid mechanics it is well known that such a pressure gradient only is defined for the spatial coordinates, so any such contribution to the energy-momentum tensor should be proportional to $p h^{ab}$, where $h^{ab}$ is the projection tensor, and is defined as: $h^{ab} \equiv g^{ab} - \frac{u^{a}u^{b}}{u^{c}u_{c}}$.  We can therefore write that:

\begin{equation}
T_{ab} = \mu u^{a} u^{b} + \alpha ph^{ab} \mbox{  $(\alpha = \pm 1)$}
\end{equation}

This in turn implies that:

\begin{equation}
T_{ab} = u^{a}u^{b} \left(\mu - \alpha p/u^cu_c\right) + \alpha p g^{ab}
\end{equation}

Letting $\alpha = -u^c u_c$, we find that the appropriate modified form of the energy-momentum tensor for a perfect fluid is evidently:

\begin{equation}
T^{ab} = (\mu + p)u^{a}u^{b} - u^{c}u_{c} g^{ab} p
\end{equation}

Letting $\mu + p = w$:

\begin{equation}
T^{ab} = w u^{a} u^{b} - u^{c}u_{c}g^{ab} p
\end{equation}

This is precisely the equation for the energy-momentum tensor for a perfect fluid as given in \cite{landau_fluid}. We proceed with our derivation by first denoting the additional viscous contributions by $\mathcal{V}_{ab}$, so that:

\begin{equation}
T_{ab} = wu_{a}u_{b} - u_{c}u^{c}g_{ab}p + \mathcal{V}_{ab}
\end{equation}

To obtain the form of this additional tensor term, we note that from classical fluid mechanics, the Euler equation is given as:

\begin{equation}
\left(\rho u_i \right)_{,t} = -\Pi_{ik,k}
\end{equation}

(Where, $\Pi_{ik,k}$ is the momentum flux tensor). Also, recall that for a non-viscous fluid, one has:

\begin{equation}
\Pi_{ik} = p \delta_{ik} + \rho u_i u_k
\end{equation}

We simply add a term that represents the viscous momentum flux, $\Sigma'_{ik}$:

\begin{equation}
\Pi_{ik} = p \delta_{ik} + \rho u_i u_k - \Sigma'_{ik} = -\Sigma_{ik} + \rho u_i u_k
\end{equation}

Where

\begin{equation}
\Sigma_{ik} = -p \delta_{ik} + \Sigma'_{ik}
\end{equation}

is the stress tensor, while, $\Sigma'_{ik}$ is the \emph{viscous} stress tensor. 

The general form of the viscous stress tensor can be formed by recalling that viscosity is formed when the fluid particles move with respect to each other at different velocities, so this stress tensor can only depend on spatial components of the fluid velocity. We assume that these gradients in the velocity are small, so that the momentum tensor only depends on the first derivatives of the velocity in some Taylor series expansion. Therefore, $ \Sigma'_{ik}$ is some function of the $u_{i,k}$.  In addition, when the fluid is in rotation, no internal motions of particles can be occurring, so we consider linear combinations of $u_{i,k} + u_{k,i}$, which clearly vanish for a fluid in rotation with some angular velocity, $\Omega_i$. The most general viscous tensor that can be formed is given by:

\begin{equation}
\Sigma'_{ik} = \eta \left(u_{i,k} + u_{k,i} - \frac{2}{3}\delta_{ik} u_{l,l}\right) + \xi \delta_{ik} u_{l,l}
\end{equation}

Here, $\eta$ and $\xi$ are the coefficients of shear and bulk/second viscosity respectively \cite{landau_fluid} \cite{kundu}. In the above equation, we note that $\delta_{ik} u_{l,l}$  is an expansion rate tensor, and $\left(u_{i,k} + u_{k,i} - \frac{2}{3}\delta_{ik} u_{l,l}\right)$ represents the shear rate tensor. Since we would like to generalize this expression in a general relativistic sense, we replace the partial derivatives above with covariant derivatives, and the Kroenecker tensor with a more general metric tensor, that is, $\delta_{ik} \rightarrow g_{ik}$. Therefore, we have that:

\begin{equation}
\Sigma'_{ik} = \eta \left(u_{i;k} + u_{k;i} - \frac{2}{3}g_{ik} u_{l;l}\right) + \xi g_{ik} u_{l;l}
\end{equation}

As we shall show later when discussing the Raychaudhuri equation, a fluid moving through spacetime can be described as a congruence of fluid lines in motion that exhibit properties of expansion/collapse, rotation, and shear. Denoting the shear rate tensor as $\sigma_{ab}$, and the expansion rate tensor as $\theta \equiv u^{a}_{;a}$, the above expression generalizes to:

\begin{equation}
\mathcal{V}_{ab} = -2\eta \sigma_{ab} - \xi \theta h_{ab}
\end{equation}

 Substituting these relationships into Eq. (3.7), we finally obtain:

\begin{equation}
T_{ab} = \left(\mu + p\right)u_{a}u_{b} - u_{c}u^{c}g_{ab}p - 2\eta \sigma_{ab} - \xi \theta h_{ab}
\end{equation}

\subsection{Equation of Motion for The Fluid}
Now that we have the energy-momentum tensor in hand, we are in a position to derive an equation of motion for the fluid. It is quite a remarkable fact that the classical equations of motion of fluid mechanics and electromagnetism can be derived quite easily from the appropriate energy-momentum tensor. That is, equations of motion are typically found from conservation laws. The generalization of conservation of energy-momentum in general relativity appears via the divergence-free property of the energy-momentum tensor, which itself, is coupled to the divergence-free property of the Einstein tensor. 

For notational convenience (which will be clearer in what follows), we will let $\pi_{ab} = 2\eta \sigma_{ab}$, such that:
\begin{equation}
T_{ab} = (\mu + p)u_{a}u_{b} + g_{ab} p - \pi_{ab} - \xi \theta h_{ab} 
\end{equation}

We write this in the mixed form:
\begin{eqnarray}
T^{b}_{a} &=& (\mu + p) u_{a} u^{b} + \delta_{a}^{b} p - \pi^{b}_{a} - \xi \theta h^{b}_{a} \nonumber \\
&=& (\mu + p) u_{a} u^{b} + \delta_{a}^{b} p - \pi^{b}_{a} - \xi \theta h^{b}_{a}  \nonumber \\
&=& (\mu + p) u_{a} u^{b} + \delta_{a}^{b} p - \pi^{b}_{a} - \xi \theta \left(\delta^{b}_{a} + u_{a}u^{b}\right) 
\end{eqnarray}

Energy conservation requires that this quantity should have a vanishing divergence, $T^{ab}_{;b} = 0 = T^{b}_{a;b}$. Upon a contraction with $u^{a}$, we get that:
\begin{eqnarray}
u^{a} \left((\mu + p) u_{a} u^{b} + \delta_{a}^{b} p - \pi^{b}_{a} - \xi \theta \left(\delta^{b}_{a} + u_{a}u^{b}\right) \right)_{;b} = 0
\end{eqnarray}

We will now evaluate this expression term-by-term. For the first term, we get:

\begin{eqnarray}
u^{a} \left( (\mu + p)_{;b} u_{a} u^{b} + (\mu + p) u_{a;b} u^{b} + (\mu + p) u_{a} u^{b}_{;b}\right) \nonumber \\
=  u^{a} \left( (\mu + p)_{,b} u_{a} u^{b} + (\mu + p)u_{a;b} u^{b} + (\mu + p) u_{a} \theta \right) \nonumber \\
= \left( -(\mu + p)_{,b} u^{b}  - (\mu + p) \theta \right) \nonumber \\
= \left(-\dot{\mu} - \dot{p} - (\mu + p)\theta \right)
\end{eqnarray}

For the second term, we have:

\begin{eqnarray}
u^{a} \left( \delta^{b}_{a} p \right)_{;b} \nonumber \\
= u^{a} \left(\delta^{b}_{a;b} p + \delta^{b}_{a} p_{;b}\right) \nonumber \\
= u^{a} \left(\delta^{b}_{a} p_{,b} \right) \nonumber \\
= \dot{p}
\end{eqnarray}

For the third term, we have:
\begin{eqnarray}
u^{a} \left(\pi^{b}_{a}\right)_{;b} \nonumber \\
= -\sigma_{ab} \pi^{ab}
\end{eqnarray}

For the fourth term, we have:

\begin{eqnarray}
u^{a} \left(\xi \theta \left(\delta^{b}_{a} + u_{a} u^{b} \right)\right)_{;b} \nonumber \\
= u^{a} \left(\xi_{;b} \theta  \left(\delta^{b}_{a} + u_{a} u^{b} \right) + \xi \theta_{;b}  \left(\delta^{b}_{a} + u_{a} u^{b} \right) + \xi \theta \left(\delta^{b}_{a} + u_{a} u^{b} \right)_{;b} \right) \nonumber \\
= -\xi \theta^{2}
\end{eqnarray}

Combining these individual expressions, we obtain:

\begin{equation}
\left(-\dot{\mu} - \dot{p} - (\mu + p)\theta \right) + \dot{p}  + \sigma_{ab} \pi^{ab} + \xi \theta^{2} = 0
\end{equation}

Cleaning this up slightly, we have:

\begin{equation}
\dot{\mu} + (\mu + p)\theta - \sigma_{ab}\pi^{ab} - \xi \theta^2   = 0
\end{equation}

Recall, we had initially defined $\pi^{ab} = 2\eta \sigma^{ab}$, so this equation becomes:

\begin{equation}
\dot{\mu} + (\mu + p)\theta - 2\eta \sigma_{ab} \sigma^{ab} - \xi \theta^2  = 0
\end{equation}

Using the definition $\sigma^2 = \frac{1}{2} \sigma^{ab} \sigma_{ab} \Rightarrow \sigma^{ab} \sigma_{ab} = 2\sigma^2$, we have finally that:

\begin{equation}
\dot{\mu} + (\mu + p)\theta - 4\eta \sigma^2 - \xi \theta^2 = 0
\end{equation}

We assume that in the early universe, any fluid would behave ultra-relativistically, that is, it would obey an equation of state:

\begin{equation}
p = \frac{1}{3} \mu
\end{equation}

Substituting this relationship in the previous equation, we get:

\begin{equation}
\dot{\mu} + \left(\frac{4}{3} \mu\right) \theta - 4 \eta \sigma^2 - \xi \theta^2 = 0
\end{equation}

This seemingly simple equation is quite a remarkable result in actuality. It describes the evolution of the fluid in our universe. It is important to note, however, that this is one equation with three unknown variables. As we shall see later, this evolution equation is actually the equation that closes the Einstein field equations constituting a complete dynamical system description of the universe and its matter content.  

\subsection{A Discussion on Energy Conditions}
Later on, we will attempt to prove the existence of an initial singularity for this viscous fluid universe. The arguments presented there will depend strongly on what are known as the strong and weak energy conditions. Having just introduced the form of the energy momentum tensor, this would be an appropriate opportunity to discuss these concepts. Assuming that the universe on a large-scale structure does not have \emph{exotic} matter, that is, matter that does not have a negative energy density, any observer moving with four-velocity $u^{a}$ will necessarily measure the energy density (which is a scalar quantity) as $T_{ab}u^{a}u^{b}$. This general condition is typically broken down into two specific conditions: \emph{The Weak Energy Condition}, an \emph{The Strong Energy Condition} \cite{ellis3} \cite{hervik} .  The weak energy condition says that for all $u^{a}$ in a spacetime,
\begin{equation}
T^{ab}u_{a}u_{b} \geq 0
\end{equation}

On the other hand, the strong energy condition says that for all $u^{a}$ in a spacetime,

\begin{equation}
\left(T^{ab} - \frac{1}{2}T g^{ab}\right)u_{a}u_{b} \geq 0
\end{equation}

Using The Einstein field equations, we can write this equation as simply:

\begin{equation}
R^{ab} u_{a} u_{b} \geq 0
\end{equation}

From a physical point of view, the weak energy condition is a statement about all timelike observers measuring a positive energy density. The strong energy condition is more restrictive than that in the following sense. To diagonalize the energy-momentum tensor, let us choose a frame whose basis consists of the eigenvectors of the energy-momentum tensor, and denote the eigenvalues and eigenvectors by $\mu$ and $p_{i} \mbox{ } i = 1,...,3$. Note that it is common to label the $p_{i}$ as the principal pressures. We then see that the weak energy condition (WEC) satisfies the condition:

\begin{equation}
\mu \geq 0 \wedge \mu + p_{i} \geq 0 
\end{equation}

Whereas, the strong energy condition (SEC) satisfies the condition:

\begin{equation}
\mu + \sum_{i} p_{i} \geq 0 \wedge \mu + p_{i} \geq 0
\end{equation}

Furthermore, Eq. (3.31) implies that any spacetime whose corresponding energy-momentum tensor satisfies the SEC necessarily has positive curvature. This, in turn, says a great deal about the geodesics of this spacetime. In that, for this spacetime two neighbouring geodesics initially parallel and separated will eventually converge. These concepts will be used later on in the discussion of the necessary criteria for the existence of a past singularity.

\section{The Einstein Field Equations}
 At the heart of general relativity is the idea that spacetime itself is dynamical. That is, in The Einstein field equations is the notion that a matter source induces spacetime curvature, which in turn, induces motion on the matter source. Since we are particularly interested in early-universe cosmology, we will solve these equations for a spatially homogeneous, anisotropic, viscous fluid universe. As we alluded to earlier, we will specifically work with the Bianchi Type-IV model. The reason for this is largely that Bianchi Type-IV models have not been studied in the context of early-universe, viscous cosmology. We therefore feel that solutions to this model will provide valuable insights to the fields of early-universe cosmology, and more broadly, cosmological dynamical systems. As General Relativity is a \emph{coordinate-invariant} theory, there are often several practical approaches to solving The Einstein field equations. That is, there are several appropriate choices for coordinate systems that can simplify the process of obtaining a solution to the field equations. The great detail in describing the Bianchi structure above was with great reason. Given the structure of the Bianchi models and the general covariance of The Einstein field equations, we will employ the powerful method of \emph{orthonormal frames}. This approached was first used by Ellis and MacCallum \cite{ellismac} \cite{ellis}, and in the author's opinion, is one of the most remarkable results to come out of general relativistic theory in recent history.  In order to use this approach, we will first give a covariant description of fluids in spacetime.
 
 \subsection{The Orthonormal Frame Approach}
 If one considers a fluid described by a family of time-like curves, with 4-velocity, $u^{a}$, the properties of the fluid flow are most appropriately described by the decomposition \cite{dyer}:

\begin{equation}
u_{a;b} = -a_{a}u_{b} + \sigma_{ab} + \omega_{ab} + \frac{1}{3}h_{ab}\theta
\end{equation}

In this equation, $\theta \equiv u^{a}_{;a}$ is a measure of the divergence of the family of time-like curves, $a_{a} \equiv u_{a;b}u^{b}$ is a measure of how much the curves represent non-geodesics, thus, can be taken to be the fluid acceleration, and $h_{ab} \equiv g_{ab} - u_{a}u_{b} / u^{c}u_{c}$ is a projection tensor. We also have terms that represent the shear and vorticity of the fluid:

\begin{equation}
\sigma_{ab} = \left[\frac{1}{2}\left(u_{m;n} + u_{n;m}\right) - \frac{1}{3}u^{c}_{;c}h_{mn}\right]h^{m}_{a}h^{n}_{b}
\end{equation}

\begin{equation}
\omega_{ab} = \frac{1}{2}\left(u_{m;n} - u_{n;m}\right)h^{m}_{a} h^{n}_{b}
\end{equation}

In addition, we note that since for any velocity field $u^a u_{a} = \pm 1$, the following addition relations hold:

\begin{equation}
a_c u^c = \omega_{nc}u^c = \sigma_{nc}u^c = h_{nc}u^c = 0
\end{equation}
 
When using orthonormal frames, the metric tensor, $g_{u v}$ takes a very simple form for the following reason. Any cosmological model can be described by a coordinate system in which we can define a basis $ \{\mathbf{e}_{u} \}$, and its dual basis of differential one-forms $ \{\mathbf{\omega}^{u} \}$. As usual, the metric tensor is then defined as $g_{uv} = \mathbf{g(\mathbf{e}_{u}, \mathbf{e}_{v})}$, from which the line element is $ds^2 = g_{uv} \omega^{u} \omega^{v}$. For a given coordinate chart on the pseudo-Riemannian manifold, we can take the basis formed by $ \{\mathbf{e}_{u} \}$ as the \emph{coordinate basis} $ \{ \partial / \partial x^{i} \}$, with the dual basis then playing the role of the coordinate differential one-forms, $\{ dx^i \}$.  If we now assume that the basis constitutes the \emph{orthonormal frame} of this section, then all four basis vectors $ \{\mathbf{e}_{u} \}$ are mutually orthonormal, and thus satisfy the relationship: $ \mathbf{g(\mathbf{e}_{u}, \mathbf{e}_{v})} = \eta_{uv} = \delta_{uv} = diag(-1,1,1,1)$. Hence, $g_{uv}$  is really just the Kronecker tensor, and all tensors in this section in the orthonormal frame approach have their indices lowered/raised with it. Furthermore, the structure constants are now functions, which we denote by a lowercase $c$ to distinguish between the constants described previously by an uppercase $C$. These functions satisfy the commutation relation:
 
 \begin{equation}
 \left[\mathbf{e}_{u}, \mathbf{e}_{v}\right] = c^{a}_{uv} \mathbf{e}_{a}
 \end{equation}
 
 Recall that for an arbitrary basis, for a general (torsion-free) spacetime, we have:
 
 \begin{equation}
 c^{a}_{uv} = \Gamma^{a}_{vu} - \Gamma^{a}_{uv}
 \end{equation}

We can then write the connection coefficients as:

\begin{equation}
\Gamma_{auv} = \frac{1}{2} \left(g_{ab} c^{b}_{vu} + g_{uv}c^{b}_{av} - g_{vb}c^{b}_{au}\right)
\end{equation} 

Building on the discussion of the previous section, we shall assume that the fluid is \emph{non-tilted}. That is, the fluid four-velocity is orthogonal to the spatial hypersurfaces. This immediately implies that the fluid is irrotational and geodesic:

\begin{equation}
\omega_{uv} = u_{u;v}u^{v} = 0
\end{equation}

We can therefore write that:
\begin{equation}
\theta_{uv} = u_{u;v} = \frac{1}{3}\theta h_{uv} + \sigma_{uv}
\end{equation}

As we mentioned the four-velocity is orthogonal to the spatial slices, which additionally implies:
\begin{equation}
\theta_{uv} = \Gamma^{t}_{uv} \Rightarrow c^{t}_{ta} = c^{t}_{ab} = 0
\end{equation}

As we did with the Behr decomposition, we can write the structure coefficients as:
\begin{equation}
c^{a}_{tb} = -\theta^{a}_{b} + \epsilon^{a}_{bc} \Omega^{c}
\end{equation}

The vector $\Omega^{c}$ has the special significance of being interpreted as a local angular velocity in the rest frame of an observer with spatial frame $\{\mathbf{e}_{a}\}$. It is conventionally defined as:

\begin{equation}
\Omega^{a} = \frac{1}{2} \epsilon^{abcd} u_{b} \mathbf{e}_{c} \cdot \dot{\mathbf{e}}_{d}
\end{equation}

It can be shown that that the structure coefficients are spatial, and therefore, must correspond to The Bianchi models. We can therefore write that:

\begin{equation}
c^{k}_{ij} = \epsilon_{ijl}n^{lk} + a_{l} \left(\delta^{l}_{i} \delta^{k}_{j} - \delta^{l}_{j} \delta^{k}_{i}\right)
\end{equation}

Since each structure coefficient is constant along each orbit of transitivity, $n^{ab}$ and $a_{i}$ are functions of time only. Even though the spatial frame described by the basis  $ \{\mathbf{e}_{u} \}$ is a left-invariant set on the spatial slices, the difference in the orthonormal frame approach is that the structure coefficients are time-dependent functions. Before, the structure coefficients were always taken to be constant in the true sense of the word. We can find evolution equations for  $n^{ab}$ and $a_{i}$ by first realizing that the Jacobi identity holds for all vectors. The Jacobi identity as applied to the set of vectors $\left(\mathbf{U}, \mathbf{e}_{a}, \mathbf{e}_{b}\right)$ gives:

\begin{eqnarray}
[\mathbf{U},[ \mathbf{e}_{a},  \mathbf{e}_{b}]] + [ \mathbf{e}_{a},[ \mathbf{e}_{b},\mathbf{U}]]+ [ \mathbf{e}_{b},[\mathbf{U}, \mathbf{e}_{a}]] = 0 \nonumber \\
\Rightarrow (\mathbf{U}(c^{v}_{ab}) + c^{v}_{tu}c^{u}_{ab} + c^{v}_{au}c^{u}_{bt} + c^{v}_{bu}c^{u}_{ta}) \mathbf{e}_{v} = 0
\end{eqnarray}

In this derivation,  $\mathbf{U}$ is a gauge vector that will typically be set to $\partial_{t}$. Note that since we are assuming that the fluid is non-tilted, necessarily $\theta_{uv} = \Gamma^{t}_{uv} \Rightarrow c^{t}_{ta} = c^{t}_{ab} = 0$. We therefore obtain the identity:

\begin{equation}
\partial_{t} (c^{k}_{ab}) + c^{k}_{td}c^{d}_{ab} + c^{k}_{ad}c^{d}_{bt} + c^{k}_{bd}c^{d}_{ta} = 0
\end{equation}

One can show that upon applying Jacobi's identity to the three spatial vectors, we get the eigenvalue equation \cite{hervik}:
\begin{equation}
n^{ij}a_{i} = 0
\end{equation}

Combining the previous equations and taking the trace of Eq. (4.15) gives the evolution equation for $a_{i}$ as:

\begin{equation}
\boxed{
\dot{a_{i}} + \frac{1}{3}\theta a_{i} + \sigma_{ij}a^{j} + \epsilon_{ijk}a^{j}\Omega^{k} = 0}
\end{equation}

The \emph{trace-free} part of Eq. (4.15) gives the evolution equation for $n_{ab}$ as:

\begin{equation}
\boxed{
\dot{n_{ab}} + \frac{1}{3}\theta n_{ab} + 2n^{k}_{_{(a}\epsilon_{b})kl}\Omega^{l}-2n_{k(_{a}}\sigma_{b})^{k} = 0}
\end{equation}

The important thing to keep in mind is that we would like the structure constants as defined in Eq. (4.13) to correspond to a Lie algebra. In order for this to be the case, Eq. (4.16) must hold.  Bianchi models are typically classified into two categories: Class A models and Class B models. Class A models are all Bianchi models such that $a_{i} = 0$ for which the latter is satisfied. For class B models, $a_{i}$ must be an eigenvector of $n^{ij}$ with zero eigenvalue. It is important to note that we always take $n^{ij}$ to be a symmetric matrix and as such we can diagonalize it using a orientation of our choice for the spatial frame. It has been conventional to assume:

\begin{equation}
n_{ij} = diag(n_{1}, n_{2}, n_{3}), \mbox{   }, a^{i} = (0, 0, a)
\end{equation}

The Jacobi identity immediately implies that $n_{33}a = n_{3}a = 0$. It is a fact of linear algebra that eigenvalues of a matrix are invariant under a conjugation operation with respect to rotations. One can then classify the different Bianchi models by the signs of the eigenvalues $n_{11}, n_{22}, n_{33} = n_{1}, n_{2}, n_{3}$ and $a$. Since, we are dealing with the Bianchi Type IV, we note that:

\begin{equation}
a^{i} = a \delta^{i}_{3} > 0, \mbox{ and } n_{1} > 0, n_{2} = n_{3} = 0.
\end{equation}

(The signs of the eigenvalues $n_{i}$ and $a$ for the other Bianchi types can be found in \cite{ellis} \cite{hervik}.

\subsection{The Dynamical Equations as Implied by The Einstein Field Equations}
In this chapter, we have so far explained the orthonormal frame approach, derived the evolution equations for $n_{ij}$ and $a^{i}$, and described the classification of the Bianchi Type-IV model. We will now derive the dynamical equations which are implied by The Einstein Field equations in the orthonormal frame approach. To give a sense of how remarkable this approach is to solving the field equations, one should first look at them written in terms of the metric tensor alone \cite{thorne}.

\begin{eqnarray}
 \tilde g^{uv} \frac{{\partial ^2 \tilde g^{ab} }}{{\partial x^u \partial x^v }} = \tilde g_{,u} ^{av} \tilde g^{bu} _{,v}  + \frac{1}{2}\left( {\tilde g^{ab} \tilde g_{lu} \tilde g_{,p} ^{lv} \tilde g_{,v} ^{pu} } \right) + \tilde g_{lu} \tilde g^{vp} \tilde g_{,v} ^{al} \tilde g_{,p} ^{bu}  -  \nonumber \\ 
 \tilde g^{av} \tilde g_{uv} \tilde g_{,p} ^{bv} \tilde g_{,l} ^{up}  - \tilde g^{bl} \tilde g_{uv} \tilde g_{,p} ^{av} \tilde g_{,l} ^{up}  +  \nonumber \\ 
 \frac{1}{8}\left( {2\tilde g^{av} \tilde g^{bu}  - \tilde g^{ab} \tilde g^{lu} } \right)\left( {2\tilde g_{vp} \tilde g_{st}  - \tilde g_{ps} \tilde g_{vt} } \right)\tilde g_{,l} ^{vt} \tilde g_{,u} ^{ps} \nonumber \\ 
 \end{eqnarray}

This system of partial differential equations represents the vacuum field equations in the deDonder gauge ($\tilde{g}^{ab}_{,b} = 0$), where $\tilde{g}^{uv} \equiv \sqrt{-g} g^{uv}$. This is a very complicated set of ten coupled highly nonlinear partial differential equations. As we shall show below, the field equations lead to a system of first-order ordinary differential equations in the orthonormal frame approach.

\subsection{The Raychaudhuri Equation}
We begin by recalling that one can consider a fluid evolution as described by a family of time-like curves, with 4-velocity, $u^{a}$. We had from before:
\begin{equation*}
u_{a;b} = -a_{a}u_{b} + \sigma_{ab} + \omega_{ab} + \frac{1}{3}h_{ab}\theta
\end{equation*}

Analyzing the equation of the divergence of the fluid flow, we are interested in how this divergence evolves in time:

\begin{equation}
\frac{d\theta}{dt} = \left(u^{a}_{;a}\right)_{;b}u^b = u^{a}_{;a;b}u^b
\end{equation}

Recalling the definition of the Riemann curvature tensor as determine whether the second covariant derivatives of some vector field in the manifold do not commute:

\begin{eqnarray}
u_{a;d;b} - u_{a;b;d} &=& R^{c}_{adb}u_c \nonumber \\
\Rightarrow u^{a}_{;d;b} - u^{a}_{;b;d} &=& R^{ca}_{db} u_{c} \nonumber \\
\Rightarrow u^{a}_{;d;b} - u^{a}_{;b;d} &=& -R^{ac}_{db} u_{c} \nonumber \\
\Rightarrow u^{a}_{;d;b} - u^{a}_{;b;d} &=& -R^{a}_{cdb} u^{c} 
\end{eqnarray}

Performing a contraction on the indices $a$ and $d$, we finally get:

\begin{equation}
u^{a}_{;a;b} - u^{a}_{;b;a} = - R_{bc}u^{c}
\end{equation}

Substituting this equation into Eq. (4.23), we obtain:

\begin{equation}
\frac{d\theta}{dt} =   u^a_{;b;a}u^{b} - R_{bc}u^c u^b
\end{equation}

Using the fact that: $\left(u^a_{;b}u^b\right)_{;a} = u^{a}_{;b;a}u^b + u^{a}_{;b}u^{b}_{;a}$, we have:

\begin{equation}
\frac{d\theta}{dt} = a^{a}_{;a} - R_{bc}u^{c}u^{b} - u_{a;b}u^{b;a}
\end{equation}

Since  $h_{ab}h^{ab} = 3, \sigma^2 = \frac{1}{2}\sigma^{ab}\sigma_{ab}, \omega^2 = \frac{1}{2}\omega^{ab}\omega_{ab}$, we now have:

\begin{equation}
u_{a;b}u^{b;a} = 2\sigma^2 - 2\omega^2 + \frac{1}{3}\theta^2 + \frac{2}{3}\theta h_{ab}\sigma^{ab}
\end{equation}

In addition, we also have that:

\begin{equation}
h^{ab}u_{a;b} = -a_{a}u_{b}h^{ab} + h^{ab}\sigma_{ab} + h^{ab}\omega_{ab} + \frac{1}{3}\theta h_{ab}h^{ab}
\end{equation}

Since: $h^{ab} \sigma_{ab} = 0$, and that:

\begin{equation}
u_{a;b} u^{b;a} = 2\sigma^2 - 2\omega^2 + \frac{1}{3}\theta^2
\end{equation}

We now finally get the \emph{Raychaudhuri equation} for $\theta$:

\begin{eqnarray}
\frac{d\theta}{dt} &=& -R_{bc}u^{b}u^{c} + a^{a}_{;a} + 2\omega^2 - 2\sigma^2 - \frac{1}{3}\theta^2 \nonumber \\
&=& -R_{bc}u^{b}u^{c} + \left(a^{a}_{,a} - \Gamma^{a}_{an}a_{n}\right) + 2\omega^2 - 2\sigma^2 - \frac{1}{3}\theta^2 
\end{eqnarray}

We can make a further simplification using the Einstein field equations as follows for the Ricci tensor term in the above equation:

\begin{eqnarray}
R_{ab} - \frac{1}{2}Rg_{ab} &=& \kappa T_{ab} \nonumber \\
\Rightarrow R - 2R &=& \kappa T \nonumber \\
\Rightarrow R &=& -\kappa T \nonumber \\
\Rightarrow R_{ab} + \frac{1}{2}\kappa T g_{ab} &=& \kappa T_{ab} \nonumber \\
\Rightarrow R_{ab} &=& \kappa \left(T_{ab} - \frac{1}{2}Tg_{ab} \right)
\end{eqnarray}

We compute the trace of the energy-momentum tensor as follows:
\begin{eqnarray}
T &=& T^{a}_{a} = g^{ab}T_{ab} = g^{ab}\left(wu_{a}u_{b} - u_{c}u^{c}g_{ab}p - 2\eta \sigma_{ab} - \xi \theta h_{ab} \right) \nonumber \\
&=& g^{ab} \left(wu_{a}u_{b}\right) - g^{ab}\left(u_{c}u^{c}g_{ab}p\right) - 2\eta g^{ab}\sigma_{ab} - \xi g^{ab} \left(\theta h_{ab} \right)  \nonumber \\
&=& u^{a} u_{a} (w- p \delta^{a}_{a}) -  2\eta g^{ab}\sigma_{ab} - \xi g^{ab} \left(\theta h_{ab} \right)  \nonumber \\ 
&=& u^{a} u_{a} (w - 4p) -  2\eta g^{ab}\sigma_{ab} - \xi g^{ab} \left(\theta h_{ab} \right) \nonumber \\
&=& u^{a}u_{a} (\mu - 3p)-  2\eta g^{ab}\sigma_{ab} - \xi g^{ab} \left(\theta h_{ab} \right)
\end{eqnarray}

We can simplify things slightly by writing:
\begin{equation}
\sigma_{ab} = \left[\frac{1}{2}\left(u_{m;n} + u_{n;m}\right) - \frac{1}{3}u^{c}_{;c}h_{mn}\right]h^{m}_{a}h^{n}_{b} = u_{(a;b)} + a_{(_{a}u_b)} - \frac{1}{3}\theta h_{ab}
\end{equation}

We then have:

\begin{eqnarray}
g^{ab}\sigma_{ab} &=& g^{ab}\left(u_{(a;b)} + a_{(_{a}u_b)} - \frac{1}{3}\theta h_{ab}\right) \nonumber \\
&=& g^{ab}u_{(a;b)}  + g^{ab} a_{(_{a}u_b)} - \frac{1}{3}\theta g^{ab}h_{ab} \nonumber \\
&=& u^{a}_{;a} + 0 -\theta \nonumber \\
&=& 0
\end{eqnarray}

Here we have used the fact that:

\begin{eqnarray}
g^{ab}h_{ab} &=& g^{ab} \left(g_{ab} - \frac{u_a u_b}{u^c u_c}\right) \nonumber \\
&=& \delta^{a}_{a} - \frac{u^b u_b}{u^c u_c} \nonumber \\
&=& 3
\end{eqnarray}

We then have:

\begin{equation}
g^{ab}\left(\theta h_{ab}\right) = \theta g^{ab} h_{ab} = 3 \theta
\end{equation}

Therefore, the trace of the energy-momentum tensor $T$ takes the form:

\begin{eqnarray}
T &=& u^{a}u_{a} (\mu - 3p)-  2\eta g^{ab}\sigma_{ab} - \xi g^{ab} \left(\theta h_{ab} \right)  \nonumber \\
&=& u^{a}u_{a}(\mu - 3p) - 3\xi \theta \nonumber \\
\end{eqnarray}

From our original definition of $T_{ab}$, it follows that:

\begin{eqnarray}
T_{ab}u^{a}u^{b} &=& \left(wu_{a}u_{b} - u_{c}u^{c}g_{ab}p - 2\eta \sigma_{ab} - \xi \theta h_{ab} \right)u^{a}u^{b} \nonumber \\
&=& wu_{a}u_{b}u^{a}u^{b} - u_cu^{c}g_{ab}u^{a}u^{b}p - 2\eta \sigma_{ab}u^{a}u^{b} - \xi \theta h_{ab}u^{a}u^{b}  \nonumber \\
&=& w\left(u_{a}u^{a}\right)^2 - p - 2\eta \sigma_{ab} u^{a} u^{b} - \xi \theta h_{ab}u^{a}u^{b} \nonumber \\
&=& w - p - 2\eta \sigma_{ab} u^{a} u^{b} - \xi \theta h_{ab}u^{a}u^{b}   \nonumber \\
&=& w - p \nonumber \\
&=& \mu + p - p  \nonumber \\
&=& \mu 
\end{eqnarray}

Combining the results we obtained above, we obtain:

\begin{eqnarray}
R_{ab}u^{a}u^{b} &=& \kappa u^{a} u^{b} \left(T_{ab} - \frac{1}{2}Tg_{ab}\right) \nonumber \\
&=& \kappa u^{a}u^{b} T_{ab} - \frac{1}{2}\kappa T g_{ab} u^{a}u^{b} \nonumber \\
&=& \kappa (\mu) - \frac{1}{2} \kappa \left[u^{a}u_{a}(\mu - 3p) - 3\xi \theta \right] u_{a}u^{a} \nonumber \\
&=& \kappa (\mu )- \frac{1}{2} \kappa \left[ (\mu - 3p) - 3 \xi \theta u^{a} u_{a}  \right] \nonumber \\
&=& \kappa \mu - \frac{1}{2} \kappa \mu + \frac{3}{2} \kappa p + \frac{3}{2} \kappa \xi \theta u^{a}u_{a} \nonumber \\
&=& \frac{1}{2} \kappa \left(\mu + 3p \right) + \frac{3}{2} \kappa \xi \theta u^{a} u_{a} \nonumber \\
\end{eqnarray}

Since we committed to the signature $(-1,+1,+1,+1)$, $u^{a} u_{a} = -1$, and:

\begin{equation}
R_{ab} u^{a} u^{b} = \frac{1}{2} \kappa \left(\mu + 3p \right) - \frac{3}{2} \kappa \xi \theta
\end{equation}

Therefore, we obtain \emph{The Raychaudhuri equation} for our viscous fluid:

\begin{eqnarray}
\frac{d\theta}{dt} &=& -R_{bc}u^{b}u^{c} + a^{a}_{;a} + 2\omega^2 - 2\sigma^2 - \frac{1}{3}\theta^2 \nonumber \\
&=& -\frac{1}{2} \kappa \left(\mu + 3p \right) + \frac{3}{2} \kappa \xi \theta - \frac{1}{2} \kappa + a^{a}_{;a} + 2\omega^2 - 2\sigma^2 - \frac{1}{3}\theta^2\nonumber \\
\end{eqnarray}

Since we are assuming the fluid is geodesic and irrotational, we obtain:

\begin{equation}
\frac{d\theta}{dt} = -\frac{1}{2} \kappa \left(\mu + 3p \right) + \frac{3}{2} \kappa \xi \theta- 2\sigma^2 - \frac{1}{3}\theta^2
\end{equation}

Since $p = \frac{1}{3}\mu$, we finally obtain:

\begin{equation}
\boxed{
\frac{d\theta}{dt}  = - \kappa \mu+ \frac{3}{2} \kappa \xi \theta- 2\sigma^2 - \frac{1}{3}\theta^2
}
\end{equation}

Note that we have taken some care to derive The Raychaudhuri equation by going through each step carefully since such a derivation for a bulk viscous fluid is not readily available in the literature.  

\subsection{The Generalized Friedmann Equation}
Exploiting the assumption of a non-tilted cosmology again, we state the additional fact that in this case:

\begin{equation}
u_{a;b} = \theta_{ab} = K_{ab}
\end{equation}

Here, we have denoted the \emph{extrinsic curvature} of the spatial slices by the \emph{extrinsic curvature tensor}, $K_{ab}$. Recall from standard differential geometry, we have that for any three-dimensional spatial slice:

\begin{eqnarray}
^{(4)}R &=& ^{(3)}R + K^2 - K^{\alpha \beta}K_{\alpha \beta} - 2 ^{(4)}R_{\alpha \beta}u^{\alpha} u^{\beta} \nonumber \\
\Rightarrow \kappa T &=& -^{(3)}R - K^2 + K^{\alpha \beta}K_{\alpha \beta} +  2 ^{(4)}R_{\alpha \beta}u^{\alpha} u^{\beta} \nonumber \\
\Rightarrow \kappa T &=&  -^{(3)}R - K^2 + K^{\alpha \beta}K_{\alpha \beta} + 2\kappa T_{\alpha \beta}u^{\alpha} u^{\beta} + \kappa T \nonumber \\
\Rightarrow 0 &=&-^{(3)}R - K^2 + K^{\alpha \beta}K_{\alpha \beta} + 2\kappa T_{\alpha \beta}u^{\alpha} u^{\beta} \nonumber \\
\Rightarrow \kappa T_{\alpha \beta} u^{\alpha \beta} &=& \frac{1}{2} \left(^{(3)}R - K^{\alpha \beta}K_{\alpha \beta} + K^2\right)
\end{eqnarray}

If we now take once again the decomposition equation and set $a_{a} = \omega_{ab} = 0$, we obtain:

\begin{equation}
u_{\alpha;\beta} = \theta_{\alpha \beta} = K_{\alpha \beta} = \sigma_{\alpha \beta} + \frac{1}{3}h_{\alpha \beta} \theta
\end{equation}

In addition we note that:

\begin{equation}
K^2 \equiv \frac{1}{2}K_{\alpha \beta}K^{\alpha \beta} 
\end{equation}

Substituting Eq. (4.38) into Eq. (4.45), we obtain:

\begin{eqnarray}
\kappa \mu &=& \frac{1}{2} \left(^{(3)}R - K^{\alpha \beta}K_{\alpha \beta} + K^2\right)
\end{eqnarray}

Using the definition in Eq. (4.47), one obtains the \emph{generalized Friedmann equation}:

\begin{equation}
\boxed{
\frac{1}{3}\theta^2 = \frac{1}{2}\sigma_{ab}\sigma^{ab} - \frac{1}{2}^{(3)}R + \kappa \mu
}
\end{equation}

\subsection{The Shear Propagation Equations}
The last set of dynamical equations implied by The Einstein Field equations are the \emph{shear propagation equations}. They essentially describe the evolution of the anisotropy in a cosmological model as a function of time. The lengthy derivation of these equations has been done in much of the literature on cosmological dynamical systems \cite{ellis} \cite{hervik} \cite{plebanski}. The idea behind the derivation is similar to that of the derivation of The Raychaudhuri equation. The \emph{shear propagation equations} are:

\begin{equation}
\boxed{
\dot{\sigma_{ab}} + \theta \sigma_{ab} -\sigma^{d}_{a}\epsilon_{bcd}\Omega^{c} - \sigma^{d}_{b}\epsilon_{acd}\Omega^{c} + ^{(3)}R_{ab} - \frac{1}{3}\delta_{ab}^{(3)}R = 2 \eta \kappa \sigma_{ab}}
\end{equation}

\subsection{The Ricci Curvature and Constraint Equations}
The greatest achievement of the orthonormal frame approach is that one no longer needs to express the components of the Ricci tensor in terms of the coordinate basis functions. The Ricci tensor is expressed entirely in terms of $a^{i}$ and $n_{ab}$. As discussed by Ellis and MacCallum \cite{ellismac} and Gr{\o}n and Hervik \cite{hervik}, the Ricci tensor takes the form:
\begin{equation}
^{(3)}R_{ab} = -\epsilon^{cd}_{a} n_{bc} a_d - \epsilon ^{cd}_{b} n_{ac} a_{d} + 2 n_{ad}n^{d}_{b} - n n_{ab} - \delta_{ab}\left(2a^2 + n_{cd}n^{cd} - \frac{1}{2}n^2\right)
\end{equation}

Contracting this equation, we obtain the Ricci scalar:

\begin{equation}
^{(3)}R = ^{(3)}R_{a}^{a} = -\left(6a^2 + n_{cd}n^{cd} - \frac{1}{2}n^2\right)
\end{equation}

One can also show that the off-diagonal components of the four-dimensional Ricci tensor yield a constraint equation:

\begin{equation}
3a^{b}\sigma_{ba} - \epsilon_{abc}n^{cd}\sigma^{b}_{d} = 0
\end{equation}

It will however be more convenient to write this equation out in component form:
\begin{eqnarray}
3a\sigma_{33} + \left(n^{11} - n^{22}\right)\sigma_{21} &=& 0 \nonumber \\
3a\sigma_{31} + n^{22}\sigma_{32} &=& 0 \nonumber \\
3a\sigma_{32} - n^{11}\sigma_{31} &=& 0 
\end{eqnarray}

We have now discussed and derived all the necessary formalism of a spatially homogeneous spacetime that is additionally anisotropic. We will now finally apply these results to The Bianchi Type IV metric in the next chapter.

\section{The Bianchi Type IV Dynamical Equations}
In this chapter, we will derive the dynamical equations associated with the Bianchi Type IV algebra, and study their solution. For the Bianchi Type IV geometry, we take $n_{11} = N$, and $n_{22} = n_{33} = 0$, where $N>0$. The constraint relations (Eq. 4.54) then lead to:

\begin{eqnarray}
3a \sigma_{33} +N \sigma_{21} &=& 0 \\
3a \sigma_{31} &=& 0 \\
3a \sigma_{32} - \sigma_{31} &=& 0
\end{eqnarray}

We can see that  $\sigma_{32} = 0$. 

Since the shear tensor is taken to be symmetric and traceless, in three dimensions, it only has two independent components. It is also clear that the only non-zero off-diagonal component is $\sigma_{21} = \sigma_{12}$. 

Therefore, we have that, for $a = b$:
\begin{equation}
\sigma_{ab} =  \left(\sigma_{+} + \sqrt{3}\sigma_{-}, \sigma_{+} - \sqrt{3}\sigma_{-}, -2\sigma_{+}\right)
\end{equation}

We have also one non-zero, independent, off-diagonal, component for the shear tensor:

\begin{equation}
\sigma_{21} = \sigma_{12} = \frac{-3a \sigma_{33}}{N} = \frac{6a \sigma_{+}}{N}
\end{equation}

Using Eqs. (4.17) and (4.18) we obtain:
\begin{eqnarray}
\Omega^{1} = \Omega^{2} &=& 0 \\
\dot{a} + \frac{1}{3} \theta a + \sigma_{33} a &=& 0 \nonumber \\
\Rightarrow \dot{a} + \frac{1}{3} \theta a - 2 \sigma_{+} a &=& 0
\end{eqnarray}

\begin{eqnarray}
\dot{N} - \frac{1}{3}\theta N - 2N \left(\sigma_{+} + \sqrt{3}\sigma_{-}\right) &=& 0 \\
\Omega^{3} &=& -6a \sigma_{+} 
\end{eqnarray}

We also compute the components of the three-dimensional Ricci tensor as follows:

\begin{eqnarray}
^{(3)}R_{ab} &=& -\epsilon^{cd}_{a} n_{bc} a_d - \epsilon ^{cd}_{b} n_{ac} a_{d} + 2 n_{ad}n^{d}_{b} - n n_{ab} - h_{ab}\left(2a^2 + n_{cd}n^{cd} - \frac{1}{2}n^2\right) \nonumber \\
&=& -\epsilon^{13}_{a} n_{b1} (a) - \epsilon^{13}_{b} n_{a1}(a) + 2n_{a1}n^{1}_{b} - nn_{ab} - \delta_{ab} \left(2a^2 + n_{11}n^{11} - \frac{1}{2}n^2\right) \nonumber \\
&=& -\epsilon^{13}_{a} n_{b1} (a) - \epsilon^{13}_{b} n_{a1}(a) + 2n_{a1}n^{1}_{b} - N n_{ab} - \delta_{ab}\left(2a^2 + \frac{1}{2}N^2 \right)
\end{eqnarray}

Now, if $a = 1, b=1$:

\begin{eqnarray}
^{(3)}R_{11} = \frac{1}{2}N^2 -  2a^2 
\end{eqnarray}

If $a = 1, b=2$:

\begin{eqnarray}
^{(3)}R_{12} = ^{(3)}R_{21} &=& -\epsilon^{13}_{2} n_{11} a \nonumber \\
&=& Na
\end{eqnarray}

If $a = 1, b=3$:
\begin{equation}
^{(3)}R_{13} = ^{(3)}R_{31} = 0
\end{equation}

If $a = 2, b=2$:

\begin{equation}
^{(3)}R_{22} =-\left(2a^2 + \frac{1}{2}N^2\right)
\end{equation}

If $a = 2, b=3$:

\begin{equation}
^{(3)}R_{23} =^{(3)}R_{32} = 0
\end{equation}

If $a = 3, b=3$:
\begin{equation}
^{(3)}R_{33} = -\left(2a^2 + \frac{1}{2}N^2\right)
\end{equation}

The Ricci scalar is computed from Eq. (151) to be:
\begin{equation}
R = ^{(3)}R_{a}^{a} = -6a^2 - \frac{1}{2}N^2
\end{equation}

Using Eq. (4.50), we derive the shear propagation equations as follows:
\begin{eqnarray}
\dot{\sigma_{ab}} + \theta \sigma_{ab} -\sigma^{d}_{a}\epsilon_{bcd}\Omega^{c} - \sigma^{d}_{b}\epsilon_{acd}\Omega^{c} + ^{(3)}R_{ab} - \frac{1}{3}\delta_{ab}^{(3)}R &=& 2 \eta \kappa \sigma_{ab} \nonumber \\
\Rightarrow \dot{\sigma_{ab}} + \theta \sigma_{ab} - \sigma^{d}_{a} \epsilon_{b3d}\Omega^{3} - \sigma^{d}_{b}\epsilon_{a3d}\Omega^{3} + ^{(3)}R_{ab} - \frac{1}{3}\delta_{ab}^{(3)}R &=& 2 \eta \kappa \sigma_{ab} \nonumber \\
\Rightarrow  \dot{\sigma_{ab}} + \theta \sigma_{ab} - \Omega^{3} \left(\sigma^{d}_{a}\epsilon_{b3d} + \sigma^{d}_{b}\epsilon_{a3d}\right)+ ^{(3)}R_{ab} - \frac{1}{3}\delta_{ab}^{(3)}R &=& 2 \eta \kappa \sigma_{ab} \nonumber \\
\end{eqnarray}

For the diagonal equations ($a = b$) we obtain:

\begin{equation}
\dot{\sigma_{+}} + \sqrt{3} \dot{\sigma_{-}} + \theta \left(\sigma_{+} + \sqrt{3} \sigma_{-}\right) - 72 a^2 \sigma_{+}^2 + \frac{2N^2}{3} = 2 \eta \kappa \left(\sigma_{+} + \sqrt{3} \sigma_{-}\right)
\end{equation}

\begin{equation}
\dot{\sigma_{+}} - \sqrt{3} \dot{\sigma_{-}} + \theta \left(\sigma_{+} - \sqrt{3} \sigma_{-}\right) + 72a^2 \sigma_{+}^2 - \frac{N^2}{3} = 2 \eta \kappa \left(\sigma_{+} - \sqrt{3} \sigma_{-}\right)
\end{equation}

\begin{equation}
 \dot{\sigma_{+}} + \sigma_{+} \theta + \frac{N^2}{6} = 2 \eta \kappa \sigma_{+}
\end{equation}

As for the off-diagonal equations, we obtain:

\begin{equation}
6 \dot{a} \sigma_{+} + 6 a \dot{\sigma_{+}} - \frac{6a\sigma_{+}\dot{N}}{N} + 6\theta  a \sigma_{+} + 12 \sqrt{3} a \sigma_{+} \sigma_{-} + N^2 a = 12 \eta \kappa a \sigma_{+}
\end{equation}

Recall, we had that:

\begin{equation*}
\dot{a} = -a \left( \frac{1}{3}\theta - 2 \sigma_{+}\right)
\end{equation*}

\begin{equation*}
\dot{N} = \frac{1}{3}\theta N + 2N \left(\sigma_{+} + \sqrt{3} \sigma_{-}\right)
\end{equation*}

Substituting these results into Eq. (5.22), we have upon simplification:

\begin{equation}
N^2 + 6 \dot{\sigma_{+}} + 2\sigma_{+} \theta = 12 \kappa \eta \sigma_{+}
\end{equation}

If we now subtract Eq. (5.19) from Eq. (5.20), we obtain:

\begin{equation}
N^4 - 144 a^2 \sigma_{+}^2 + 2 \sqrt{3} N^2\left[\dot{\sigma_{-}} + \sigma_{-} \left(\theta - 2 \kappa \eta\right)\right] = 0
\end{equation}

We can also add Eq. (5.21) and Eq. (5.23) to obtain:

\begin{equation}
\frac{-5N^2}{6} - 5 \dot{\sigma_{+}} - \sigma_{+}\theta = -10 \kappa \eta \sigma_{+}
\end{equation}

Equations (5.24) and (5.25) therefore constitute a  coupled system of first-order, ordinary differential equations for the shear variables $\sigma_{+}, \sigma_{-}$.

We will additionally write the generalized Friedmann equation in the following convenient form:
\begin{eqnarray}
\sigma^{2} &=& \frac{1}{3} \theta^2 + \frac{1}{2} \left(-6a^2 - \frac{1}{2}N^2\right) - \kappa \mu \nonumber \\
&=& \frac{1}{3} \theta^2 - 3 a^2 - \frac{1}{4}N^2 - \kappa \mu
\end{eqnarray}

Using this relationship, we will write the dynamical equation for the fluid in the following form:

\begin{equation}
\dot{\mu} = -12 a^2 \eta - 4 \kappa \mu \eta - \eta N^2 - \frac{4 \mu \theta}{3} + \frac{4 \eta \theta^2}{3} + \theta^2 \xi
\end{equation}

We will also write The Raychaudhuri equation in the following form:

\begin{eqnarray}
\dot{\theta} &=& -\frac{1}{2} \kappa \left(\mu + 3p\right) + \frac{3}{2} \kappa \xi \theta - 2 \sigma^2 -\frac{1}{3}\theta ^2 \nonumber \\
&=& \frac{1}{2} \left(12a^2 + 2 \kappa \mu + N^2 - 2 \theta^2 + 3 \kappa \theta \xi \right)
\end{eqnarray}

Therefore, our full cosmological system for the Bianchi IV viscous fluid universe is governed by a system of six, coupled, ordinary differential equations in six unknowns, with the two viscosity coefficients, $\eta$ and $\xi$:

\begin{eqnarray}
\dot{\theta} &=&\frac{1}{2} \left(12a^2 + 2 \kappa \mu + N^2 - 2 \theta^2 + 3 \kappa \theta \xi \right) \\
\dot{\mu} &=& -12 a^2 \eta - 4 \kappa \mu \eta - \eta N^2 - \frac{4 \mu \theta}{3} + \frac{4 \eta \theta^2}{3} + \theta^2 \xi \\
\dot{\sigma_{-}} &=& -\frac{N^2}{2\sqrt{3}} +\frac{ 24 \sqrt{3} a^2 \sigma_{+}^2}{N^2} - \sigma_{-}\theta + 2 \kappa \sigma_{-} \eta \\
\dot{\sigma_{+}} &=&  -\frac{N^2}{6} - \frac{\sigma_{+} \theta}{5} + 2 \kappa \eta \sigma_{+} \\
\dot{a} &=& -a \left(\frac{\theta}{3} - 2 \sigma_{+}\right) \\
\dot{N} &=& N \left[\frac{\theta}{3} + 2 \left(\sigma_{+} + \sqrt{3} \sigma_{-}\right)\right] 
\end{eqnarray}
Clearly, this system of equations has no exact solution, so numerical methods must be applied. However, even for a six-dimensional system, such as this one, numerical algorithms can be difficult to employ to obtain any relevant solutions. Therefore, we will write these equations in dimensionless form, by writing them using the expansion-normalized variables \cite{ellis}. 

\subsection{Dynamical Equations in Expansion-Normalized Variables}
The basic idea of the expansion-normalized variables is that the set of differential equations in Eqs. (5.29-5.34) is of the form:

\begin{equation}
\frac{d \mathbf{x}}{dt} = f(\mathbf{x})
\end{equation}

where $\mathbf{x} = (\theta, \mu, \sigma_{-}, \sigma_{+}, a, N) \in \mathbf{R}^{6}$. We will define $\theta \equiv 3H$, where $H$ is denoted the Hubble scalar and is defined as:

\begin{equation}
H = \frac{\dot{s}}{s}
\end{equation}

where $s$ is a cosmological length-scale function. It is also convenient to define the cosmological deceleration parameter $q$: 
\begin{equation}
q = -\frac{\ddot{s}s}{\dot{s}^2}
\end{equation}

Clearly, we have the relationship:

\begin{equation}
\dot{H} = -(1 + q)H^2
\end{equation}

It is also necessary to introduce a dimensionless time variable $\tau$ as:
\begin{equation}
s = s_{0}e^{\tau}
\end{equation}

From these equations, one can show that:
\begin{equation}
\frac{dt}{d\tau} = \frac{1}{H} \Rightarrow \frac{dH}{d\tau} \equiv H' = -(1+q)H
\end{equation}

Substituting $\theta = 3H$ into the Raychaudhuri equation, we obtain:

\begin{eqnarray}
3\dot{H} &=& -\frac{1}{2} \kappa (\mu + 3p) + \frac{6}{2} \kappa \xi H - 2 \sigma^{2} -3H^2 \nonumber \\
\Rightarrow \dot{H} &=& -\frac{1}{6} \kappa (\mu + 3p) + \kappa \xi H - \frac{2}{3}\sigma^{2} - H^2 \nonumber \\
\Rightarrow \dot{H} &=& -\frac{1}{3} \kappa \mu + \kappa \xi H - \frac{2}{3} \sigma^{2} - H^2
\end{eqnarray}

In addition, we define: a density parameter,

\begin{equation}
\Omega = \frac{\mu}{3H^2}
\end{equation}

a shear parameter, which compares the rate of shear with the rate of expansion,
\begin{equation}
\Sigma^2 = \frac{\sigma^2}{3H^2}
\end{equation}

and a curvature parameter:

\begin{equation}
K = -\frac{^{(3)}R}{6H^2}
\end{equation}

Note that, if the dynamical system has an evolution equation for $\Omega$, then this parameter is not explicitly needed in the system of differential equations. We have, however, just included it for completeness.

In addition, comparing Eqs. (5.41) and (5.38) we have the relationship:

\begin{eqnarray}
-(1+q)H^2 &=& -\frac{1}{3} \kappa \mu + \kappa \xi H - \frac{2}{3} \sigma^{2} - H^2 \nonumber \\
\Rightarrow (1+q)H^2 &=& \frac{1}{3}\kappa \mu - \kappa \xi H + \frac{2}{3} \sigma^2 + H^2 \nonumber \\
\Rightarrow (1+q) &=& \frac{1}{3H^2} \kappa \mu - \frac{\kappa \xi}{H} - \frac{2 \sigma^2}{3H^2} + 1 \nonumber \\
\Rightarrow q &=& \frac{1}{3H^2} \kappa \mu - \frac{\kappa \xi}{H} - \frac{2 \sigma^2}{3H^2}  \nonumber \\
\Rightarrow q &=& \kappa \Omega - \frac{\kappa \xi}{H} - 2\Sigma^2
\end{eqnarray}

We additionally assume the bulk viscosity obeys the \emph{equation of state}:
\begin{equation}
\frac{\xi}{H} = 3\xi_{0}
\end{equation}
Here, $\xi_{0}$ is a positive parameter. This equation of state essentially says that as the universe expands the bulk viscosity coefficient asymptotically vanishes describing the present-day, perfect fluid universe.

Therefore, we get:

\begin{equation}
q = \kappa \Omega - 3\kappa \xi_{0}- 2\Sigma^{2}
\end{equation}

Why this approach is so powerful in terms of analyzing the dynamics of our system is for the following reasons. Eq. (5.30) which is the evolution equation for the energy density, is now replaced (in the dynamical system) by $\Omega$. Based on physical arguments of physical matter having a nonnegative energy density, this quantity is taken to be positive. Therefore, the solution space of the system of differential equations has the constraint that it is effectively bounded by $\Omega \geq 0$. Keeping these concepts in mind, we will now write Eqs. (5.31) - (5.34) in dimensionless form. (Below, we define the additional dimensionless quantities: $n = N/H$, $A  = a/H$, $\Sigma_{\pm} = \sigma_{\pm} / H$):

Beginning with Eq. (5.31), we get:

\begin{eqnarray}
\dot{\Sigma_{-} H} &=& -\frac{n^2 H^2}{2\sqrt{3}} + 24\sqrt{3} \frac{A^2 H^2 \Sigma_{+}^2 H^2}{n^2H^2} - 3\Sigma_{-} H^2 + 2 \kappa \Sigma_{-}H \eta \nonumber \\
\Rightarrow \dot{\Sigma_{-}}H + \dot{H}\Sigma_{-} &=& -\frac{n^2 H^2}{2\sqrt{3}} + 24\sqrt{3} \frac{A^2 H^2 \Sigma_{+}^2 H^2}{n^2H^2} - 3\Sigma_{-} H^2 + 2 \kappa \Sigma_{-}H \eta \nonumber \\
\Rightarrow  \dot{\Sigma_{-}}H &=& -\frac{n^2 H^2}{2\sqrt{3}} + 24\sqrt{3} \frac{A^2 H^2 \Sigma_{+}^2 H^2}{n^2H^2} - 3\Sigma_{-} H^2 + 2 \kappa \Sigma_{-}H \eta -  \dot{H}\Sigma_{-} \nonumber \\
\Rightarrow \dot{\Sigma_{-}} &=& -\frac{n^2 H}{2\sqrt{3}} + 24\sqrt{3}\frac{A^2 H^2 \Sigma_{+}^2 H^2}{n^2H^3} - 3\Sigma_{-} H + 2 \kappa \Sigma_{-} \eta - \frac{\dot{H}}{H}\Sigma_{-} \nonumber \\
\Rightarrow \dot{\Sigma_{-}} &=&  -\frac{n^2 H}{2\sqrt{3}} + 24\sqrt{3} \frac{A^2 \Sigma_{+}^2 H}{n^2} - 3\Sigma_{-} H + 2 \kappa \Sigma_{-} \eta - \frac{\dot{H}}{H}\Sigma_{-} \nonumber \\
\end{eqnarray}

Using the chain rule, we obtain:
\begin{equation}
\frac{d \Sigma_{-}}{d \tau} = \frac{d \Sigma_{-}}{d t} \frac{dt}{d \tau} = \dot{\Sigma_{-}} \frac{1}{H}
\end{equation}

Therefore, we have that:

\begin{eqnarray}
\Sigma_{-} ' &=&  -\frac{n^2}{2\sqrt{3}} + 24\sqrt{3} \frac{A^2 \Sigma_{+}^2}{n^2} - 3\Sigma_{-} - 2 \kappa \Sigma_{-} \frac{\eta}{H} - \frac{\dot{H}}{H^2} \Sigma_{-} \nonumber \\
&=& -\frac{n^2}{2\sqrt{3}} + 24\sqrt{3} \frac{A^2 \Sigma_{+}^2}{n^2} - 3\Sigma_{-} - 2 \kappa \Sigma_{-} \frac{\eta}{H} + (1+q) \Sigma_{-} \nonumber \\
&=& -\frac{n^2}{2\sqrt{3}} + 24\sqrt{3} \frac{A^2 \Sigma_{+}^2}{n^2} - 3\Sigma_{-} - 6 \kappa \Sigma_{-}  \eta_{0}+ (1+q) \Sigma_{-} 
\end{eqnarray}
As is done by Saha and Rikhvitsky, \cite{saha}, we have set in the above equation, $\eta/H = 3 \eta_{0}$, where $\eta_{0}$ is a nonnegative parameter.

Continuing with Eq. (5.32), we have:

\begin{eqnarray}
\dot{\Sigma_{+} H} &=& -\frac{1}{6} (n^2 H^2) - \frac{3}{5} (\Sigma_{+} H^2) + 2 \kappa \eta \Sigma_{+}H \nonumber \\
\Rightarrow \dot{\Sigma_{+}}H &=&   -\frac{1}{6} (n^2 H^2) - \frac{3}{5} (\Sigma_{+} H^2) + 2 \kappa \eta \Sigma_{+}H - \dot{H} \Sigma_{+}   \nonumber \\
\Rightarrow \dot{\Sigma_{+}} &=& -\frac{1}{6} (n^2 H) - \frac{3}{5} \Sigma_{+} H + 2 \kappa \eta \Sigma_{+} - \frac{\dot{H}}{H} \Sigma_{+} \nonumber \\
\Rightarrow \Sigma_{+}' &=& -\frac{1}{6} n^2 - \frac{3}{5} \Sigma_{+} + 2 \kappa \frac{\eta}{H} \Sigma_{+} +(1+q)\Sigma_{+} \nonumber \\
\Rightarrow \Sigma_{+}' &=& -\frac{1}{6} n^2 - \frac{3}{5} \Sigma_{+} + 6 \kappa \eta_{0} \Sigma_{+} +(1+q)\Sigma_{+} \nonumber \\
\end{eqnarray}

Going on to Eq. (5.33), we have:

\begin{eqnarray}
\dot{A H} &=& -AH \left(H - 2\Sigma_{+}H\right) \nonumber \\
\Rightarrow \dot{A} H + \dot{H} A &=& -AH^2 + 2\Sigma_{+}AH^2 \nonumber \\
\Rightarrow \dot{A} &=& -AH + 2\Sigma_{+}AH - \frac{\dot{H}}{H}A \nonumber \\
\Rightarrow A' &=& -A + 2\Sigma_{+}A + (1+q)A
\end{eqnarray}

Finally going to Eq. (5.34), we have:
\begin{eqnarray}
\dot{nH} &=& nH \left[H + 2 \left(\Sigma_{+}H + \sqrt{3}\Sigma_{-}H\right) \right] \nonumber \\
\Rightarrow \dot{n}H &=& nH\left[H + 2 \left(\Sigma_{+}H + \sqrt{3}\Sigma_{-}H\right) \right] - \dot{H}n \nonumber \\
\Rightarrow \dot{n} &=& n \left[H + 2 \left(\Sigma_{+}H + \sqrt{3}\Sigma_{-}H\right) \right] - \frac{\dot{H}}{H} n \nonumber \\
\Rightarrow n' &=& \frac{n}H  \left[H + 2 \left(\Sigma_{+}H + \sqrt{3}\Sigma_{-}H\right) \right] - \frac{\dot{H}}{H^2} n \nonumber \\
\Rightarrow n' &=& n + 2n \left(\Sigma_{+} + \sqrt{3} \Sigma_{-}\right) + (1+q) n
\end{eqnarray}

Therefore, the non-dimensional dynamical system of equations takes the form:
\begin{eqnarray}
\Sigma_{-} ' &=& -\frac{n^2}{2\sqrt{3}} + 24\sqrt{3} \frac{A^2 \Sigma_{+}^2}{n^2} - 3\Sigma_{-} - 6 \kappa \Sigma_{-}  \eta_{0}+ (1+q) \Sigma_{-} \\
\Sigma_{+}' &=& -\frac{1}{6} n^2 - \frac{3}{5} \Sigma_{+} + 6 \kappa \eta_{0} \Sigma_{+} +(1+q)\Sigma_{+} \\
A' &=& -A + 2\Sigma_{+}A + (1+q)A \\
n' &=& n + 2n \left(\Sigma_{+} + \sqrt{3} \Sigma_{-}\right) + (1+q) n
\end{eqnarray}

One can see the advantage of this approach immediately. The solution space of the system has been reduced from $\mathbf{R}^6$ to $\mathbf{R}^5$. (The remaining evolution equation for the energy density is derived below). Raychaudhuri's equation has been decoupled from the system and is not necessary to solve to determine the dynamics of the cosmological system. It is important to note that $(\Sigma_{+}, \Sigma_{-})$ describe the anisotropy in the Hubble flow, and $(A,n)$ describe the spatial curvature of each spatial slice of the spatially homogeneous spacetime.



We now use Eq. (5.43) to calculate $\Sigma^{2}$ as follows:

\begin{eqnarray}
\Sigma^{2} &=& \frac{1}{6H^2} \left(\sigma_{ab} \sigma^{ab}\right) \nonumber \\
&=& \frac{1}{6H^2}  \left(\sigma_{11}^2 + \sigma_{22}^2 + \sigma_{33}^2\right) \nonumber \\
&=& \frac{1}{6H^2} \left(\sigma_{+}^2  + 3\sigma_{-}^2 + \sigma_{+}^2 + 3\sigma_{-}^2 + 4\sigma_{+}^2\right) \nonumber \\
&=& \frac{1}{H^2} \left(\sigma_{+}^2 + \sigma_{-}^2\right) \nonumber \\
&=& \Sigma_{+}^2 + \Sigma_{-}^2 
\end{eqnarray}

In addition, $K$ is computed from Eq.(5.44) as follows:

\begin{eqnarray}
K &=& -\frac{^{(3)}R}{6H^2} \nonumber \\
&=& \frac{1}{6H^2} \left(6a^2 + \frac{1}{2}N^2 \right) \nonumber \\
&=& A^2 + \frac{n^2}{12}
\end{eqnarray}

The $\Omega$ term in the above equations is governed by the non-dimensional version of Eq. (5.30), which we obtain as follows:

\begin{eqnarray}
\dot{\Omega}(3H^2) + 6\Omega \dot{H}H &=& -12H^2A^2\eta - 12 \kappa \eta H^2 \Omega - \eta n^2H^2 - 12H^3\Omega + 12\eta H^2 + 9H^2 \xi \nonumber \\
\Rightarrow \dot{\Omega}(3H^2) &=& -12H^2A^2\eta - 12 \kappa \eta H^2 \Omega - \eta n^2H^2 - 12H^3\Omega + 12\eta H^2 + 9H^2 \xi -  6\Omega \dot{H}H \nonumber \\
\Rightarrow \dot{\Omega} &=& -4A^2 \eta - 4\kappa \eta \Omega - \eta\frac{n^2}{3} - 4H\Omega + 4 \eta + 3\xi + 2\Omega(1+q)H \nonumber \\
\Rightarrow \Omega' &=& -12A^2 \eta_{0} - 12 \kappa \Omega \eta_{0} - \eta_{0} n^2 - 4 \Omega + 12 \eta_{0} + 9 \xi_{0} + 2\Omega(1+q)
\end{eqnarray}

Therefore, the cosmological dynamical system takes the form:

\begin{eqnarray}
\Sigma_{-} ' &=& -\frac{n^2}{2\sqrt{3}} + 24\sqrt{3} \frac{A^2 \Sigma_{+}^2}{n^2} - 3\Sigma_{-} - 6 \kappa \Sigma_{-}  \eta_{0}+ (1+q) \Sigma_{-} \nonumber \\
\Sigma_{+}' &=& -\frac{1}{6} n^2 - \frac{3}{5} \Sigma_{+} + 6 \kappa \eta_{0} \Sigma_{+} +(1+q)\Sigma_{+} \nonumber \\
A' &=& -A + 2\Sigma_{+}A + (1+q)A \nonumber \\
n' &=& n + 2n \left(\Sigma_{+} + \sqrt{3} \Sigma_{-}\right) + (1+q) n \nonumber \\
\Omega' &=& \eta_{0}\left(-12A^2 - 12 \kappa \Omega  -  n^2\right) - 4 \Omega + 12 \eta_{0} + 9 \xi_{0} + 2\Omega(1+q) \nonumber \\
\end{eqnarray}

where  $q = \kappa\Omega - 3\kappa\xi_{0}- 2 \Sigma^{2}$, $\Sigma^{2} =  \Sigma_{+}^2 + \Sigma_{-}^2$, $K =A^2 + \frac{n^2}{12}$, and $\Omega \geq 0$.

\section{Solutions To The Bianchi Type IV Dynamical Equations}
In this chapter, we present several numerical solutions to the system of equations (5.61). Clearly, this system, which is nonlinear and coupled has no exact solutions. Our goal however is to discover values for $\xi_{0}$ and $\eta_{0}$ such that the system asymptotically isotropizes, that is, $\Sigma_{\pm} \rightarrow 0 \mbox{ as } t \rightarrow 0$. Such physical characteristics resemble an anistropic early-universe that over time isotropized to the present-day universe. Note that in the diagrams below, asterisks indicate the initial conditions for the numerical experiments.
\newpage
\subsection{Plots of Numerical Experiments}
\begin{figure}[h]
\begin{center}
\caption{Phase plot of the anisotropy in the Hubble flow for $\xi_{0} = 0.1$, $\eta_{0} = 0.1$. }
\includegraphics{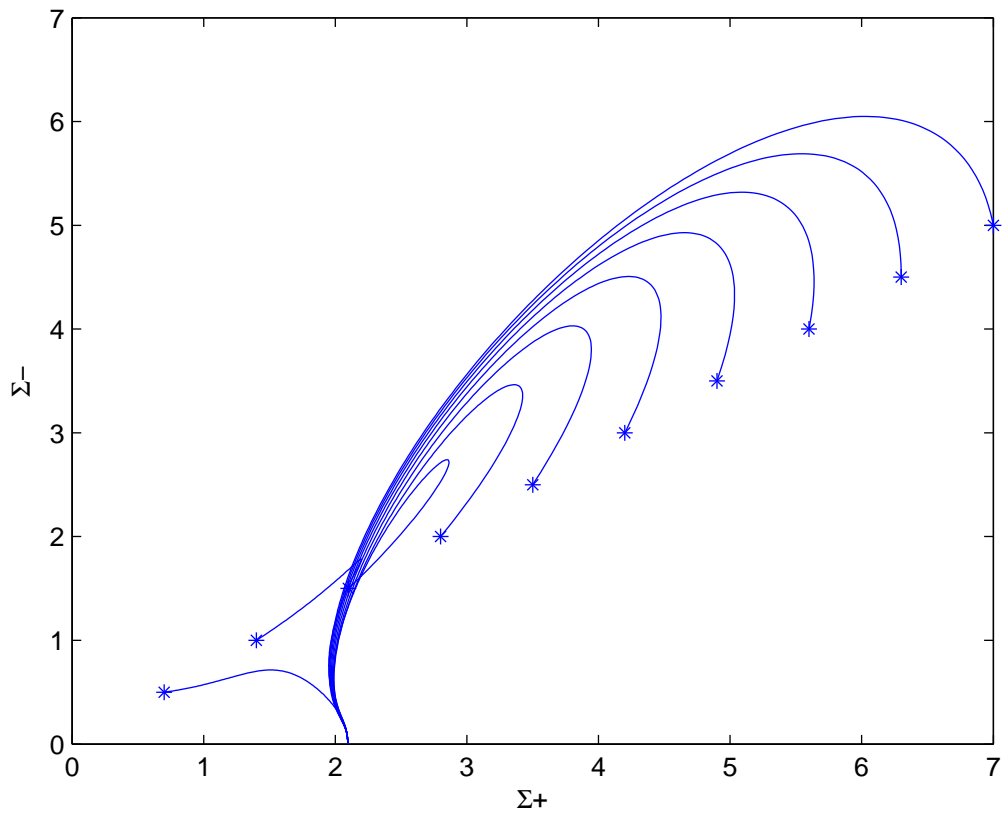}
\label{default}
\end{center}
\end{figure}

\newpage
\begin{figure}[h]
\begin{center}
\caption{Three-dimensional phase plot of the anisotropy in the Hubble flow and spatial curvature for $\xi_{0} = 0.1$, $\eta_{0} = 0.1$. }
\includegraphics{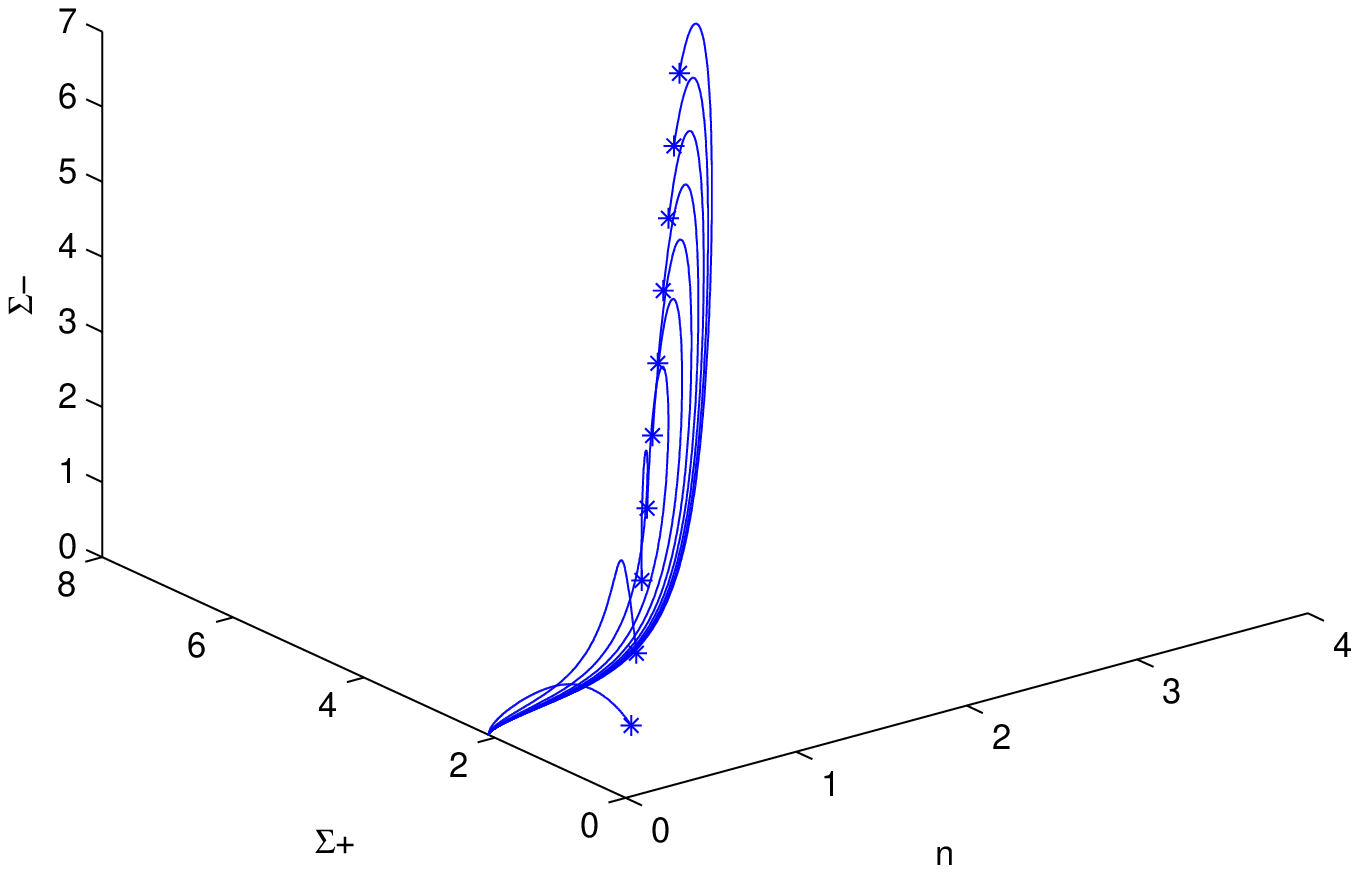}
\label{default}
\end{center}
\end{figure}

\newpage

\begin{figure}[h]
\begin{center}
\caption{Phase plot of the anisotropy and spatial curvature for $\xi_{0} = 0.1$, $\eta_{0} = 0.1$. }
\includegraphics{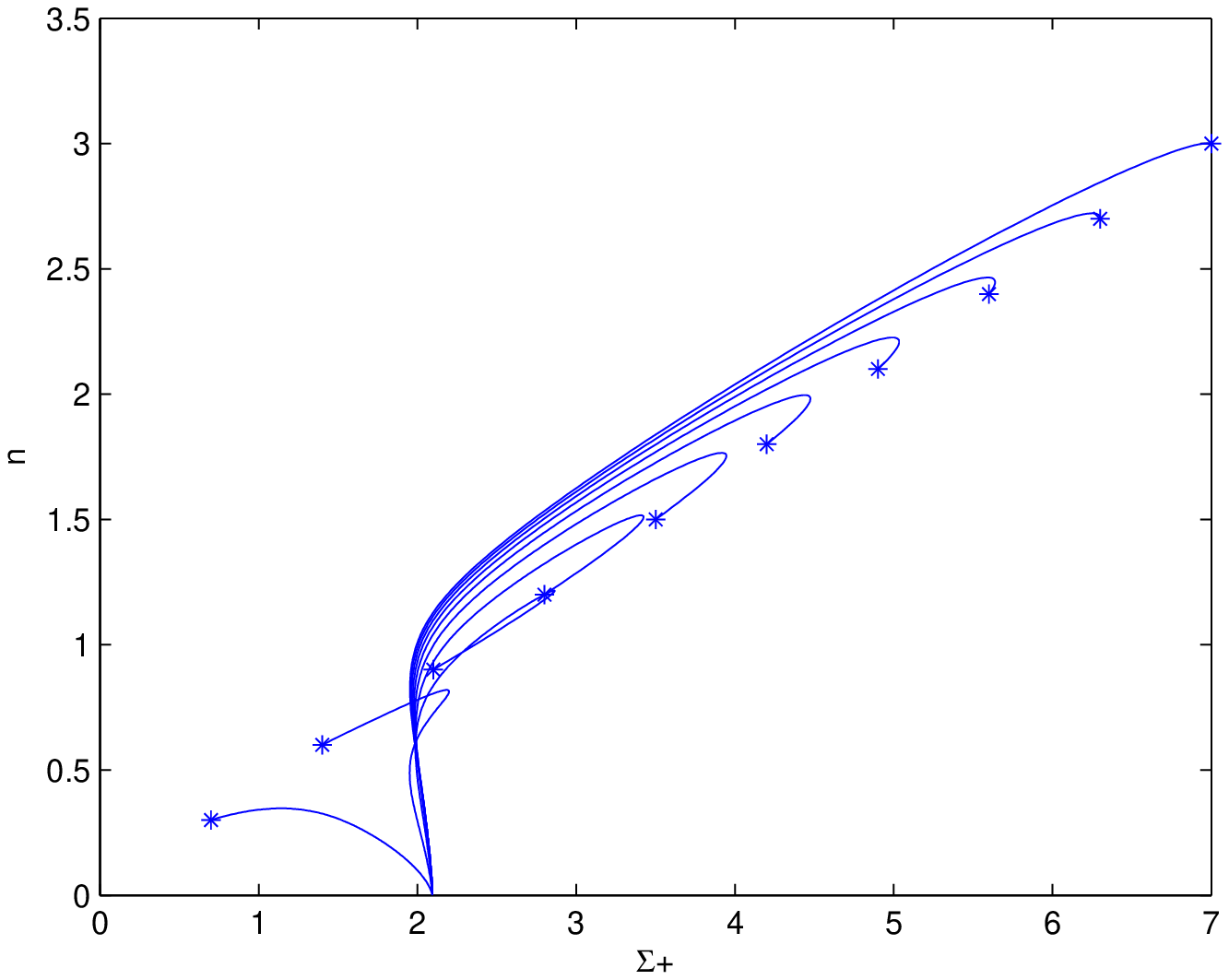}
\label{default}
\end{center}
\end{figure}

\newpage
\begin{figure}[h]
\begin{center}
\caption{Plots of the Energy Density and Anisotropy variables as functions of time, for $\xi_{0} = 0.1$, $\eta_{0} = 0.1$. One can see that $\Omega \rightarrow 0$, $\Sigma_{-} \rightarrow 0$, but $\Sigma_{+} \not \to 0$, so the model only partially isotropizes.}
\includegraphics{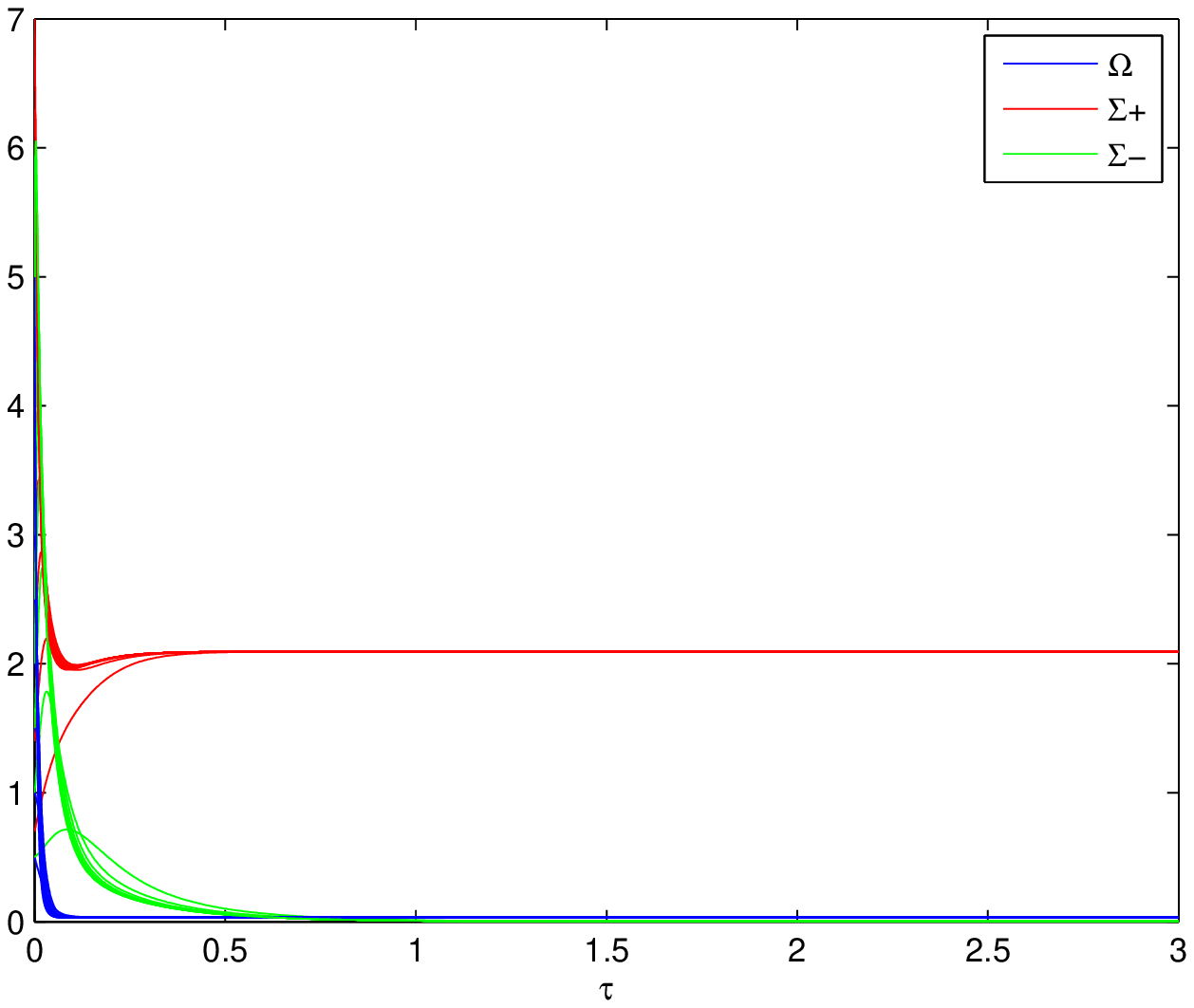}
\label{default}
\end{center}
\end{figure}


\newpage

\begin{figure}[h]
\begin{center}
\caption{Phase plot of the anisotropy in the Hubble flow for $\xi_{0} = 0.1$, $\eta_{0} = 1$. }
\includegraphics{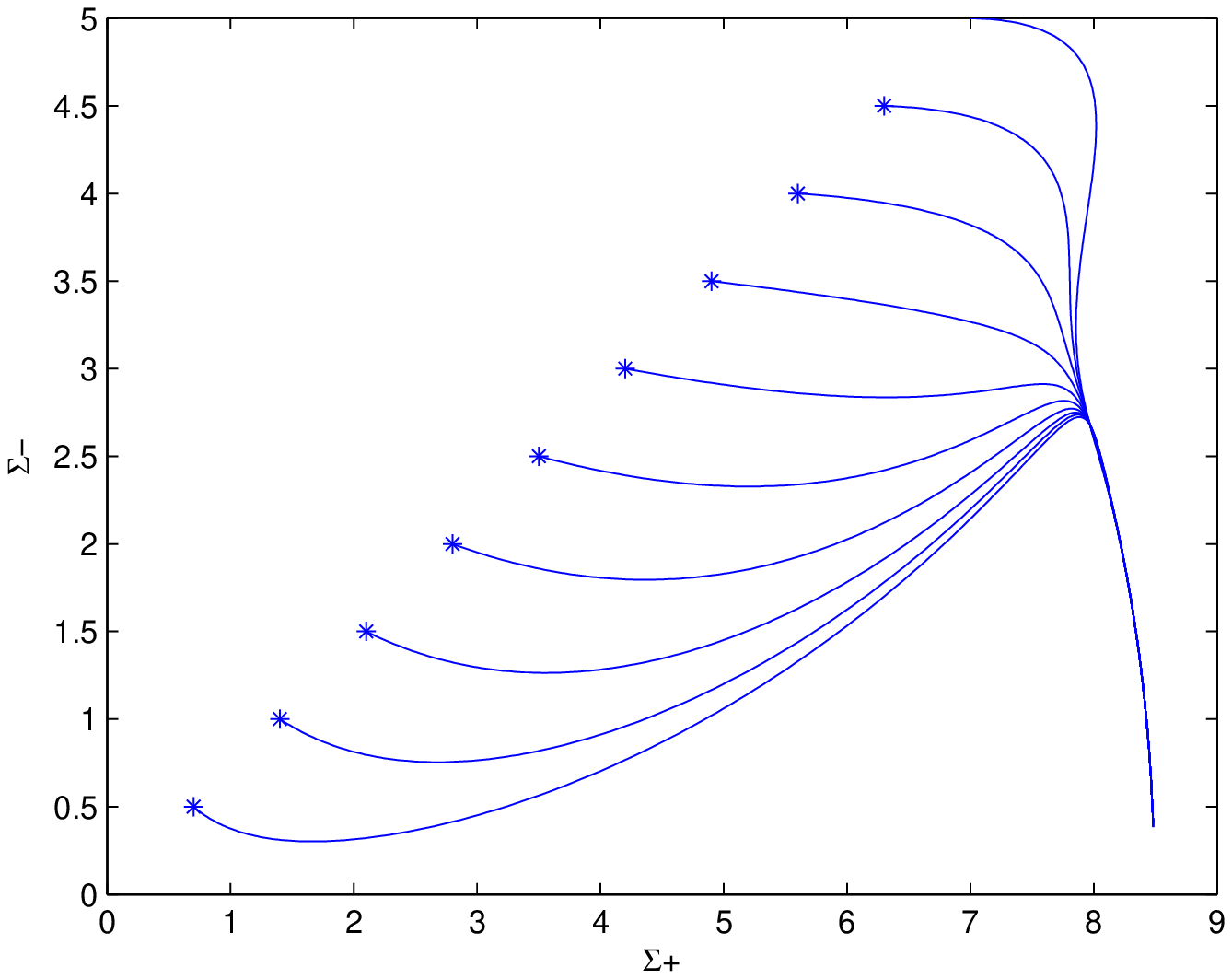}
\label{default}
\end{center}
\end{figure}

\newpage
\begin{figure}[h]
\begin{center}
\caption{Three-dimensional phase plot of the anisotropy in the Hubble flow and spatial curvature for $\xi_{0} = 0.1$, $\eta_{0} = 1$. }
\includegraphics{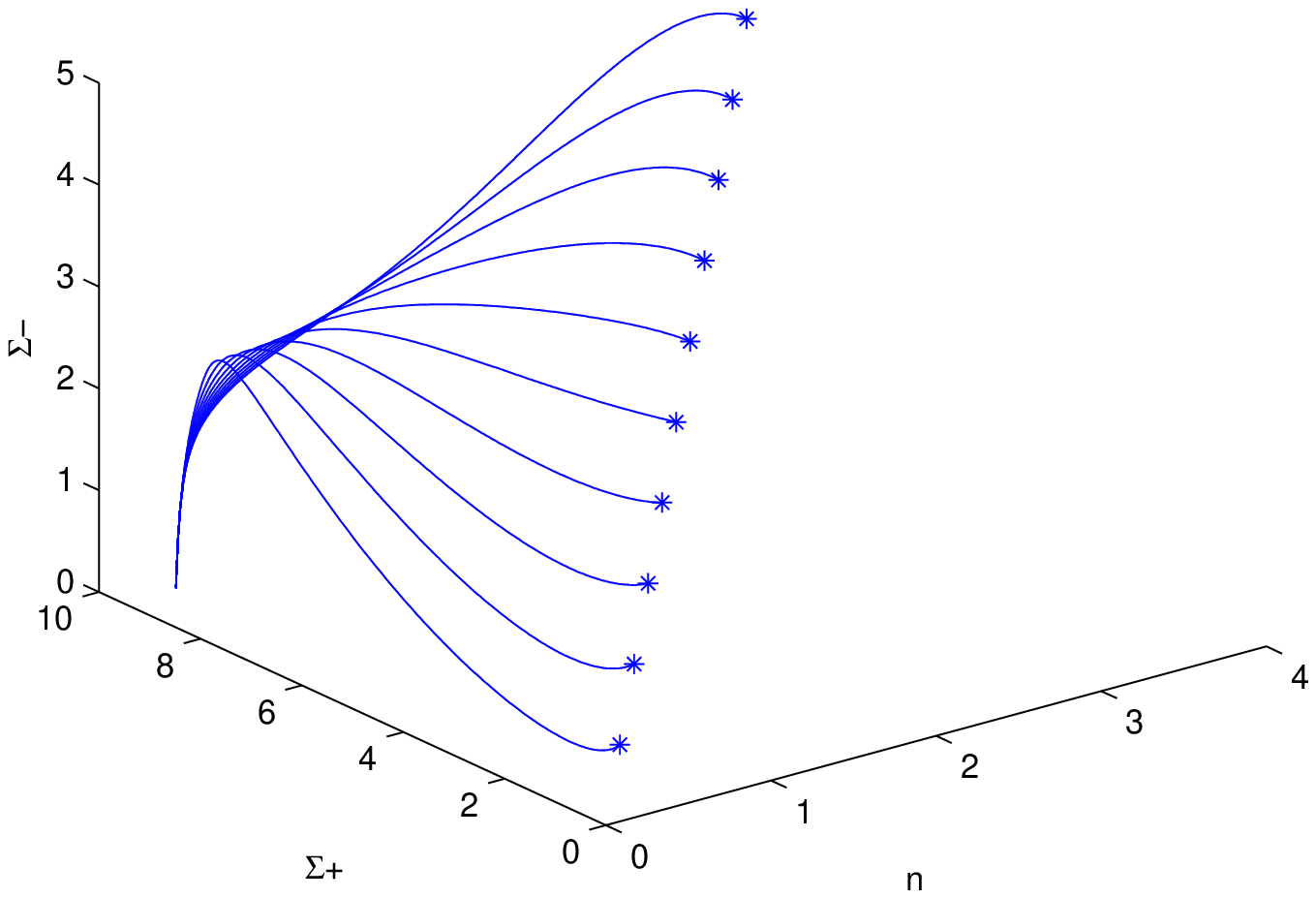}
\label{default}
\end{center}
\end{figure}

\newpage

\begin{figure}[h]
\begin{center}
\caption{Phase plot of the anisotropy and spatial curvature for $\xi_{0} = 0.1$, $\eta_{0} = 1$. }
\includegraphics{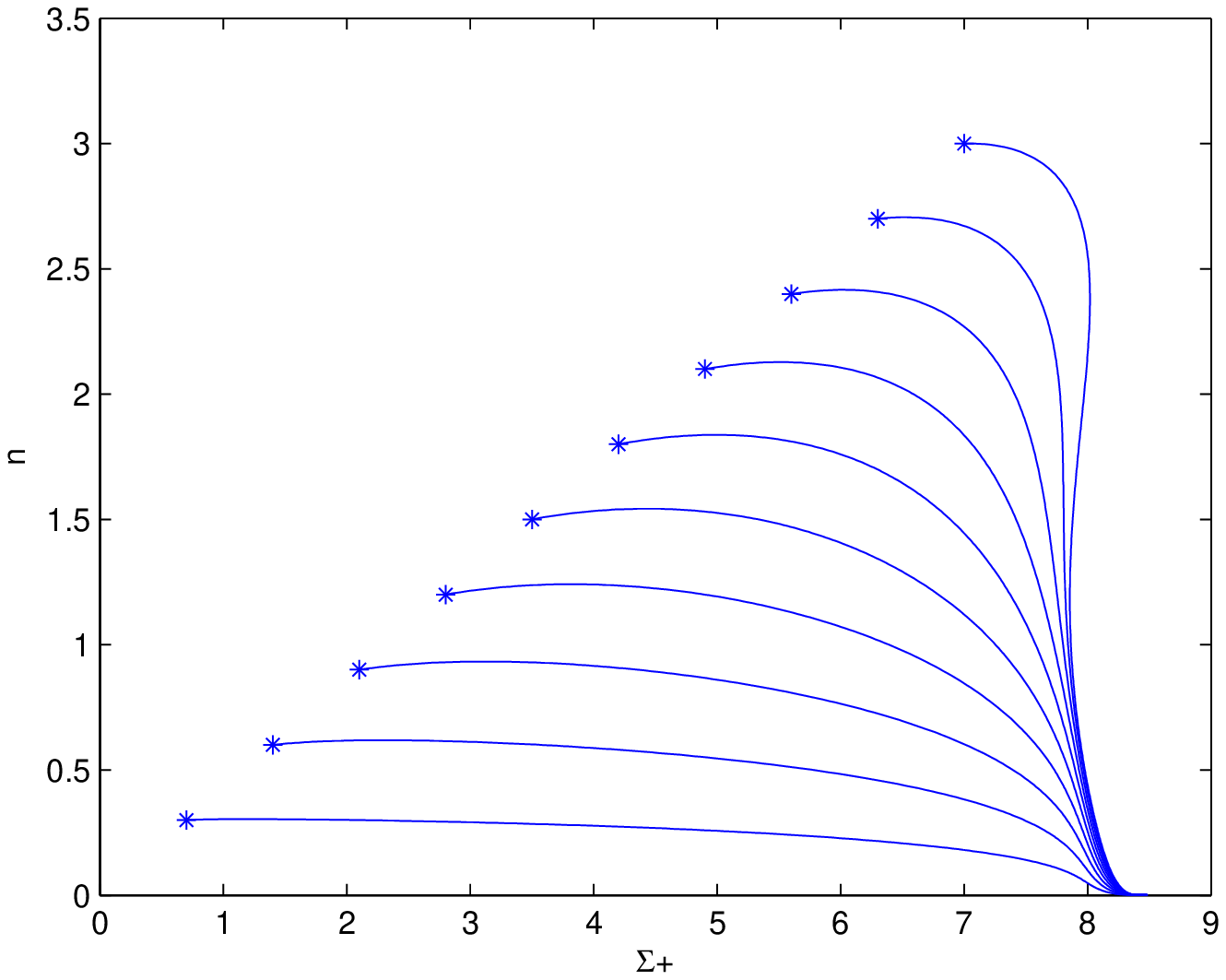}
\label{default}
\end{center}
\end{figure}

\newpage
\begin{figure}[h]
\begin{center}
\caption{Plots of the Energy Density and Anisotropy variables as functions of time, for $\xi_{0} = 0.1$, $\eta_{0} = 1$. One can see that $\Omega \rightarrow 0$, $\Sigma_{-} \rightarrow 0$, but $\Sigma_{+} \not \to 0$, so the model only partially isotropizes.}
\includegraphics{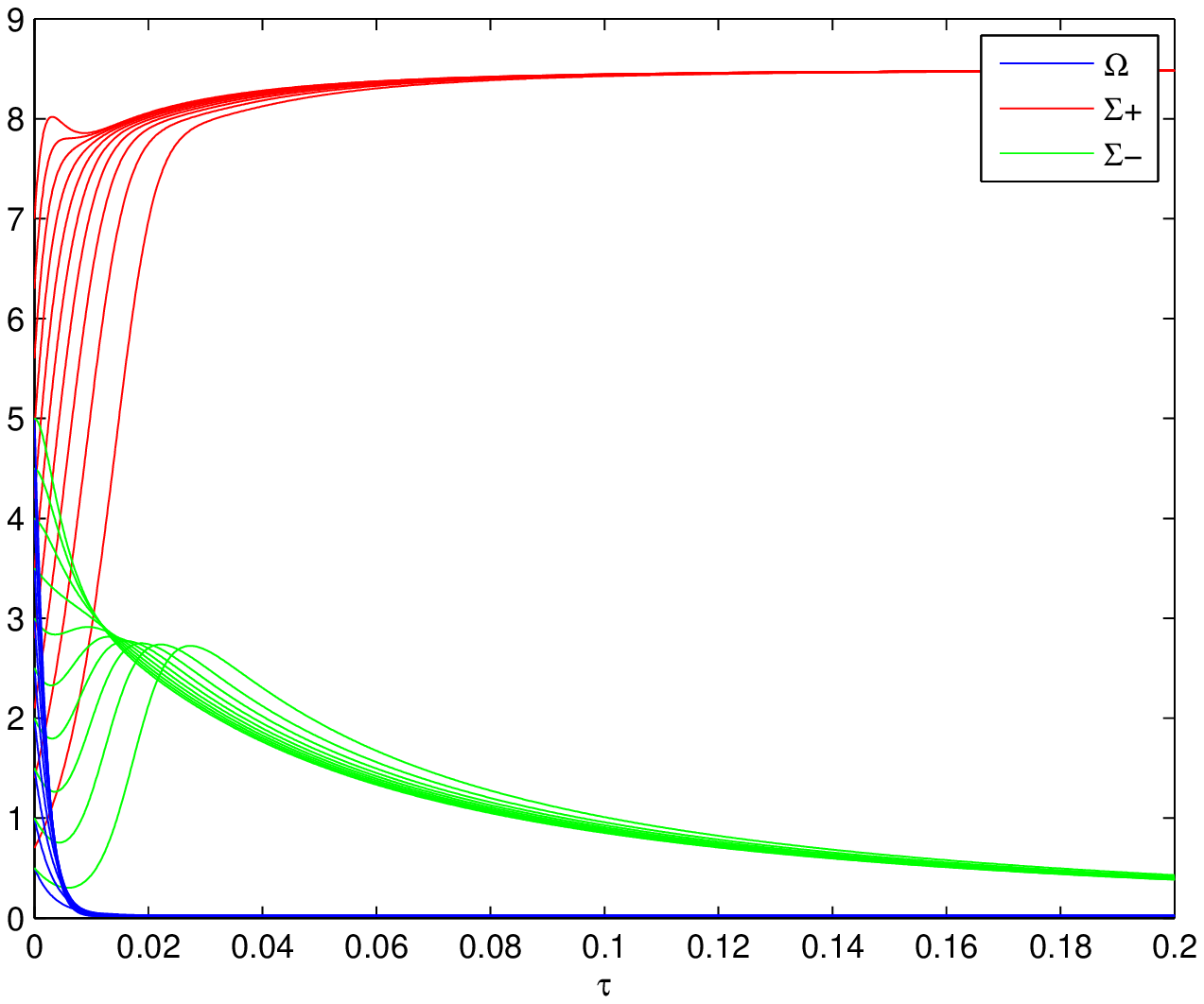}
\label{default}
\end{center}
\end{figure}

\newpage


\begin{figure}[h]
\begin{center}
\caption{Phase plot of the anisotropy in the Hubble flow for $\xi_{0} = 0.1$, $\eta_{0} = 10$. }
\includegraphics{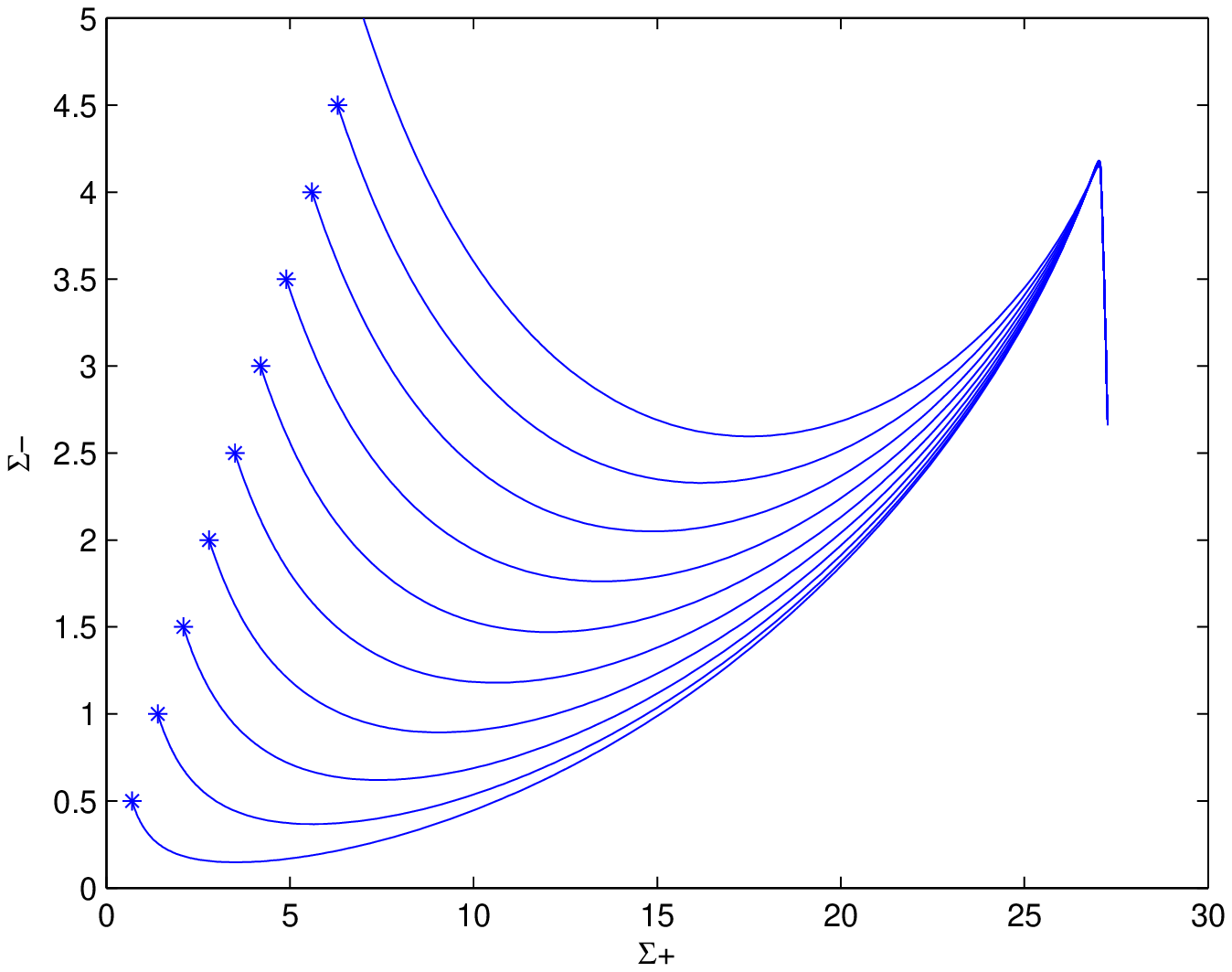}
\label{default}
\end{center}
\end{figure}

\newpage
\begin{figure}[h]
\begin{center}
\caption{Three-dimensional phase plot of the anisotropy in the Hubble flow and spatial curvature for $\xi_{0} = 0.1$, $\eta_{0} = 10$. }
\includegraphics{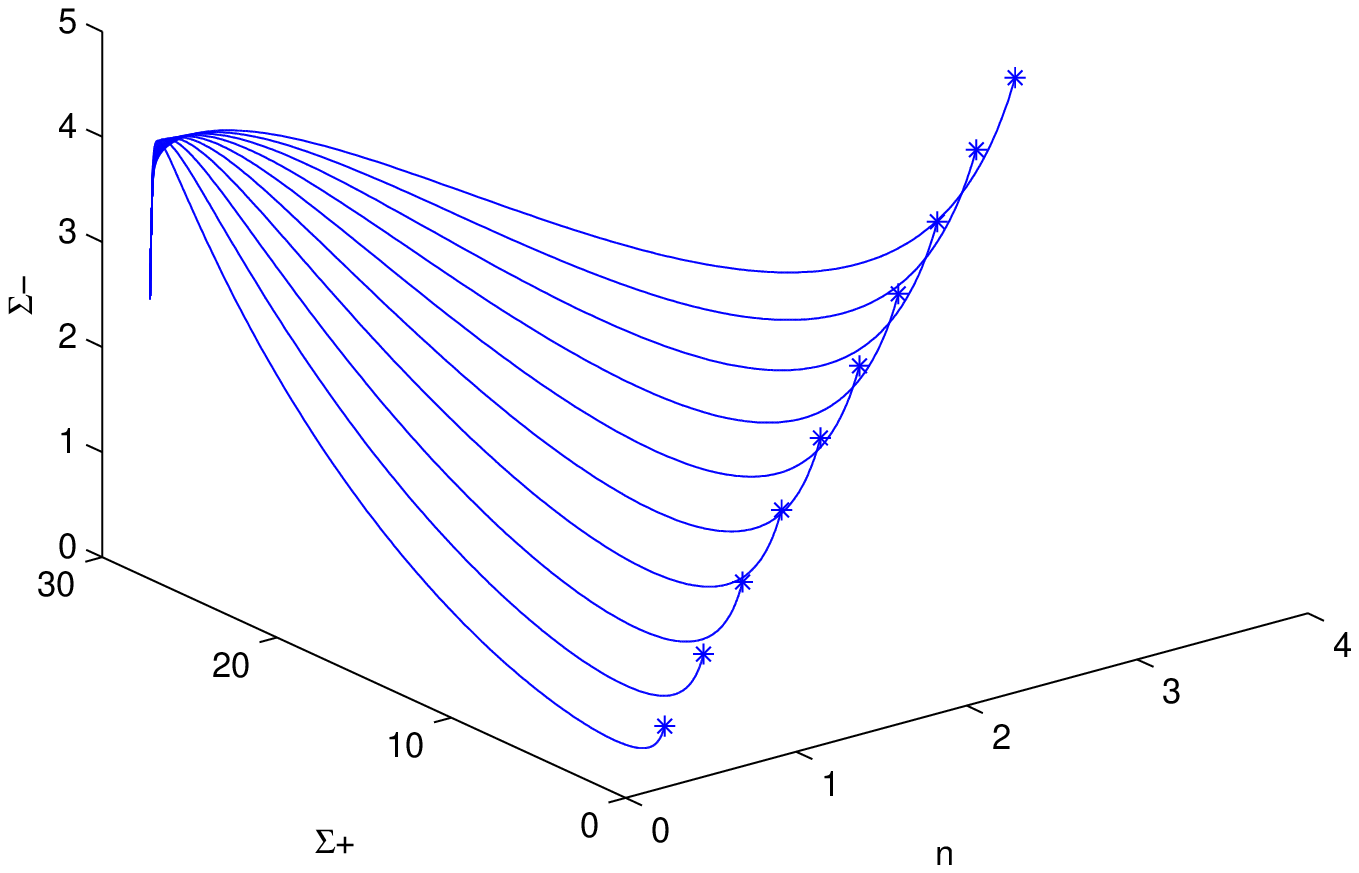}
\label{default}
\end{center}
\end{figure}

\newpage

\begin{figure}[h]
\begin{center}
\caption{Phase plot of the anisotropy and spatial curvature for $\xi_{0} = 0.1$, $\eta_{0} = 10$. }
\includegraphics{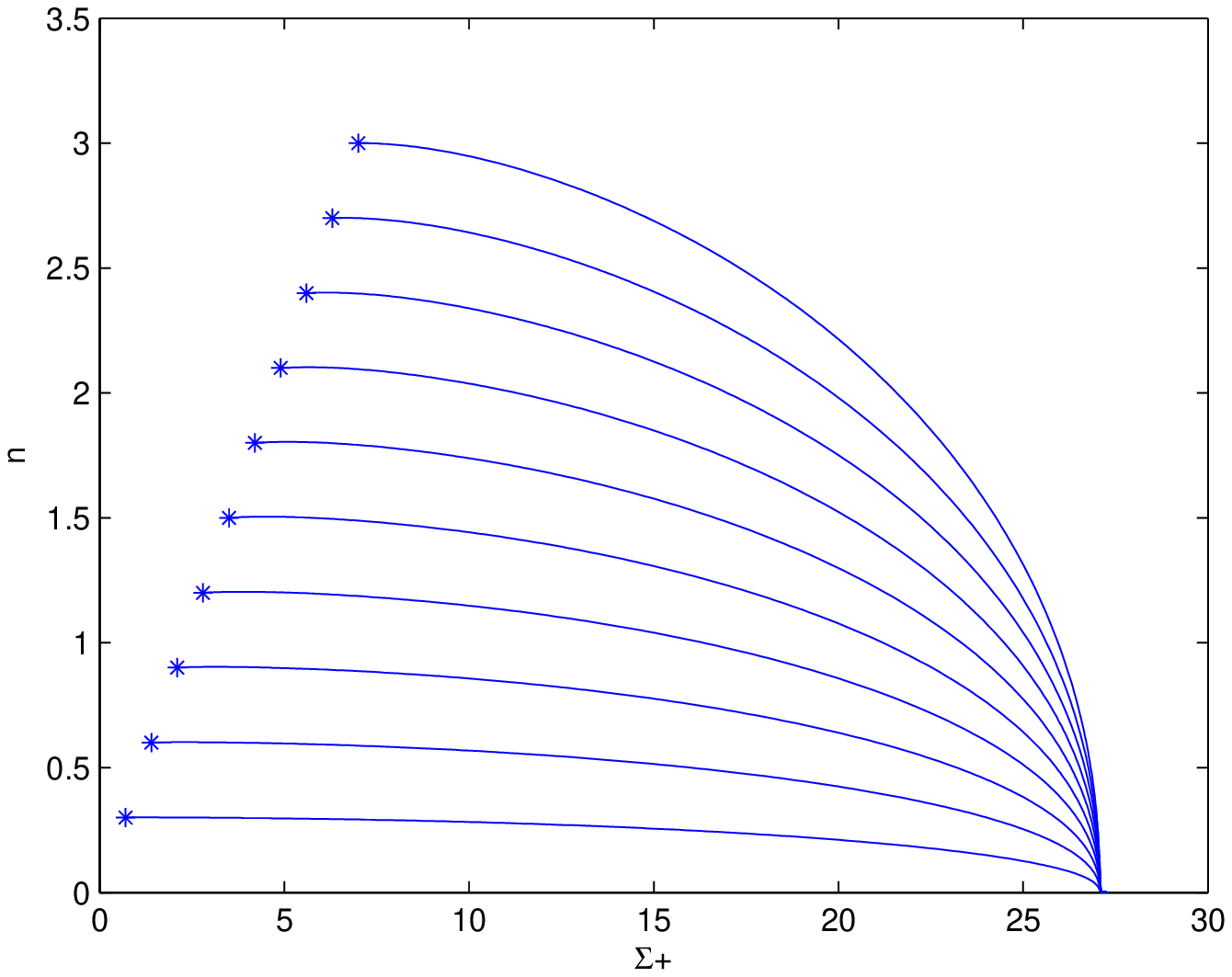}
\label{default}
\end{center}
\end{figure}

\newpage
\begin{figure}[h]
\begin{center}
\caption{Plots of the Energy Density and Anisotropy variables as functions of time, for $\xi_{0} = 0.1$, $\eta_{0} = 10$. One can see that $\Omega \rightarrow 0$, but $\Sigma_{-} \not \to 0$, and $\Sigma_{+} \not \to 0$, so the model does not isotropize at all in this case.}
\includegraphics{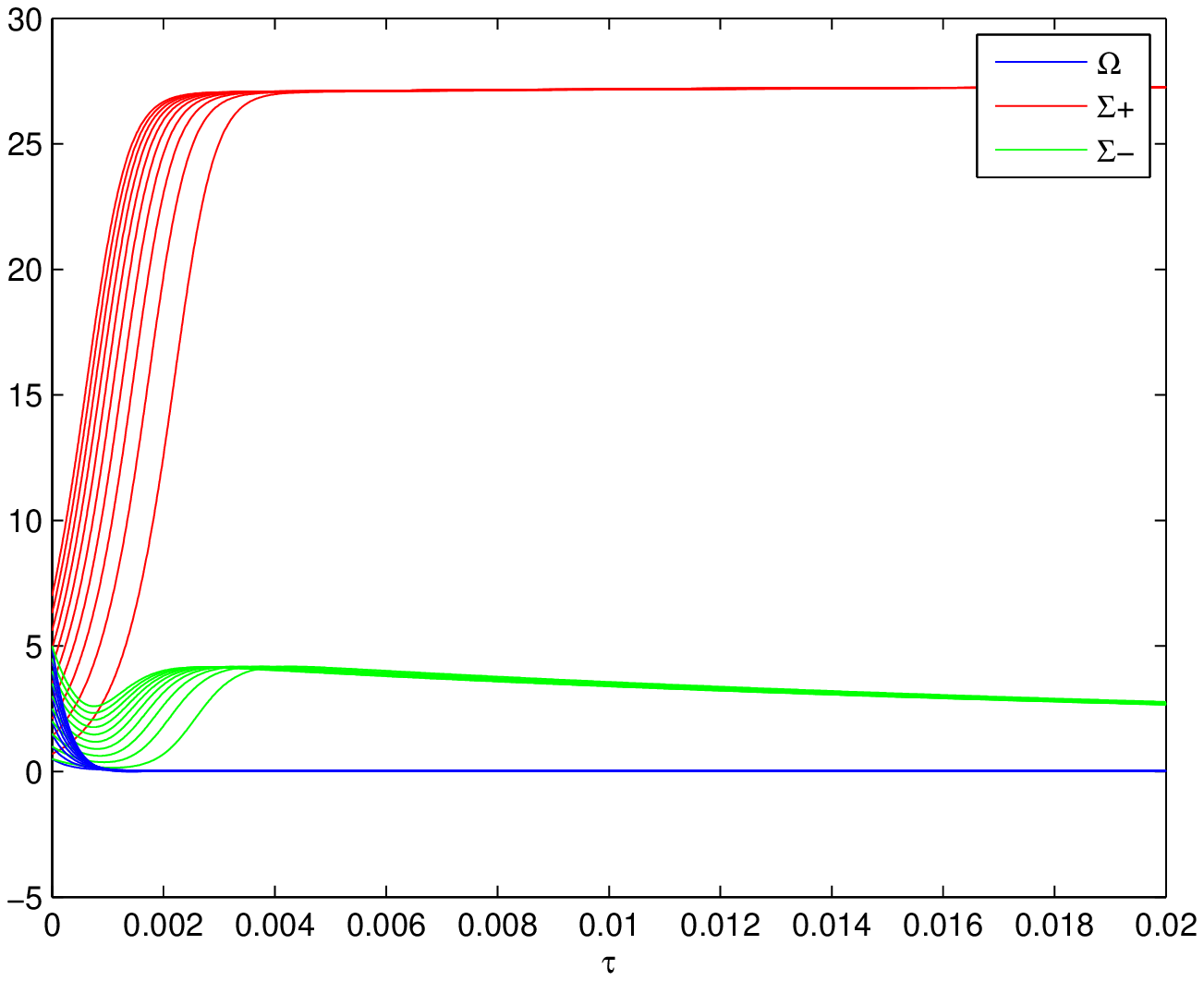}
\label{default}
\end{center}
\end{figure}

\newpage

\begin{figure}[h]
\begin{center}
\caption{Phase plot of the anisotropy in the Hubble flow for $\xi_{0} = 0.1$, $\eta_{0} = 100$. }
\includegraphics{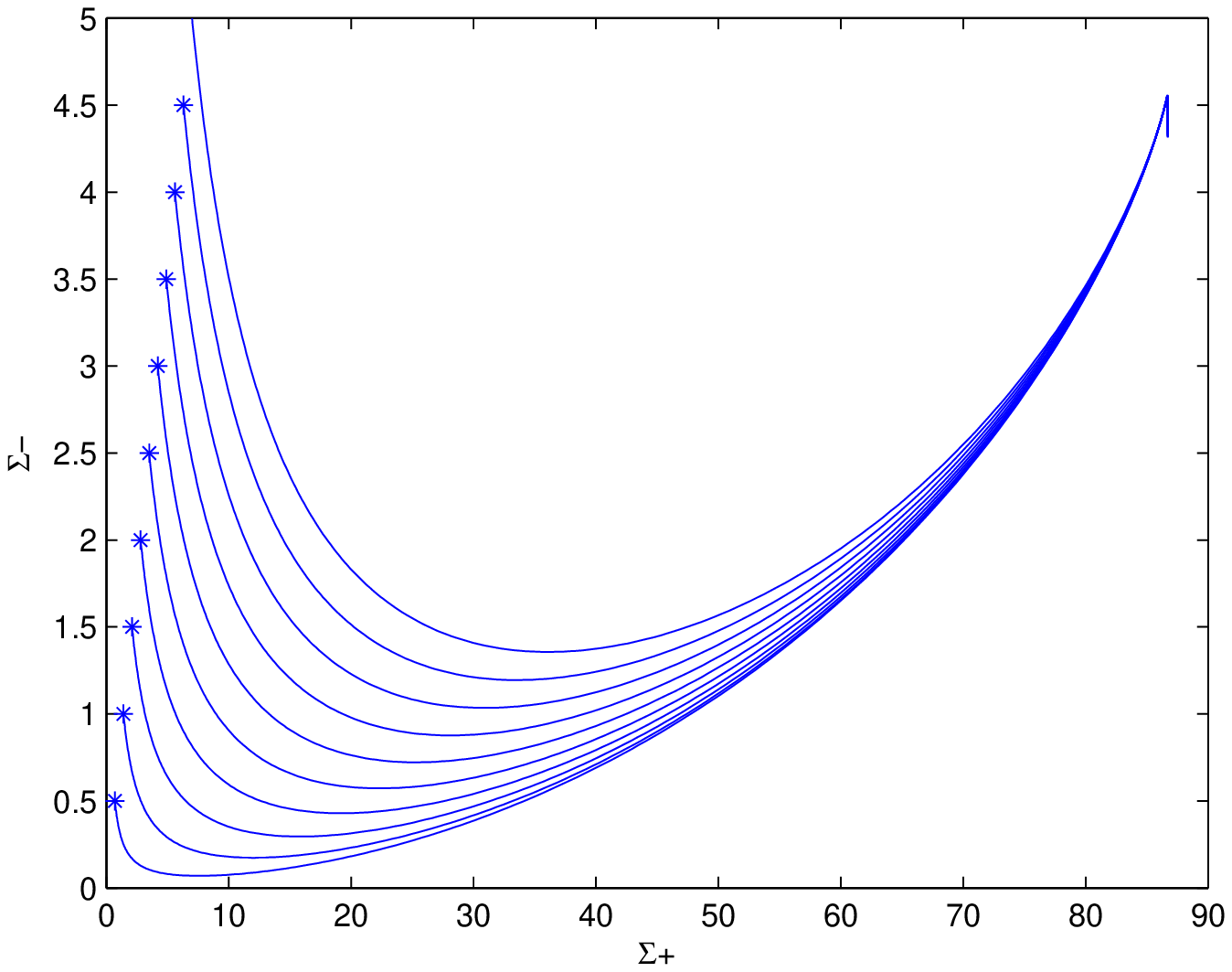}
\label{default}
\end{center}
\end{figure}

\newpage
\begin{figure}[h]
\begin{center}
\caption{Three-dimensional phase plot of the anisotropy in the Hubble flow and spatial curvature for $\xi_{0} = 0.1$, $\eta_{0} = 100$. }
\includegraphics{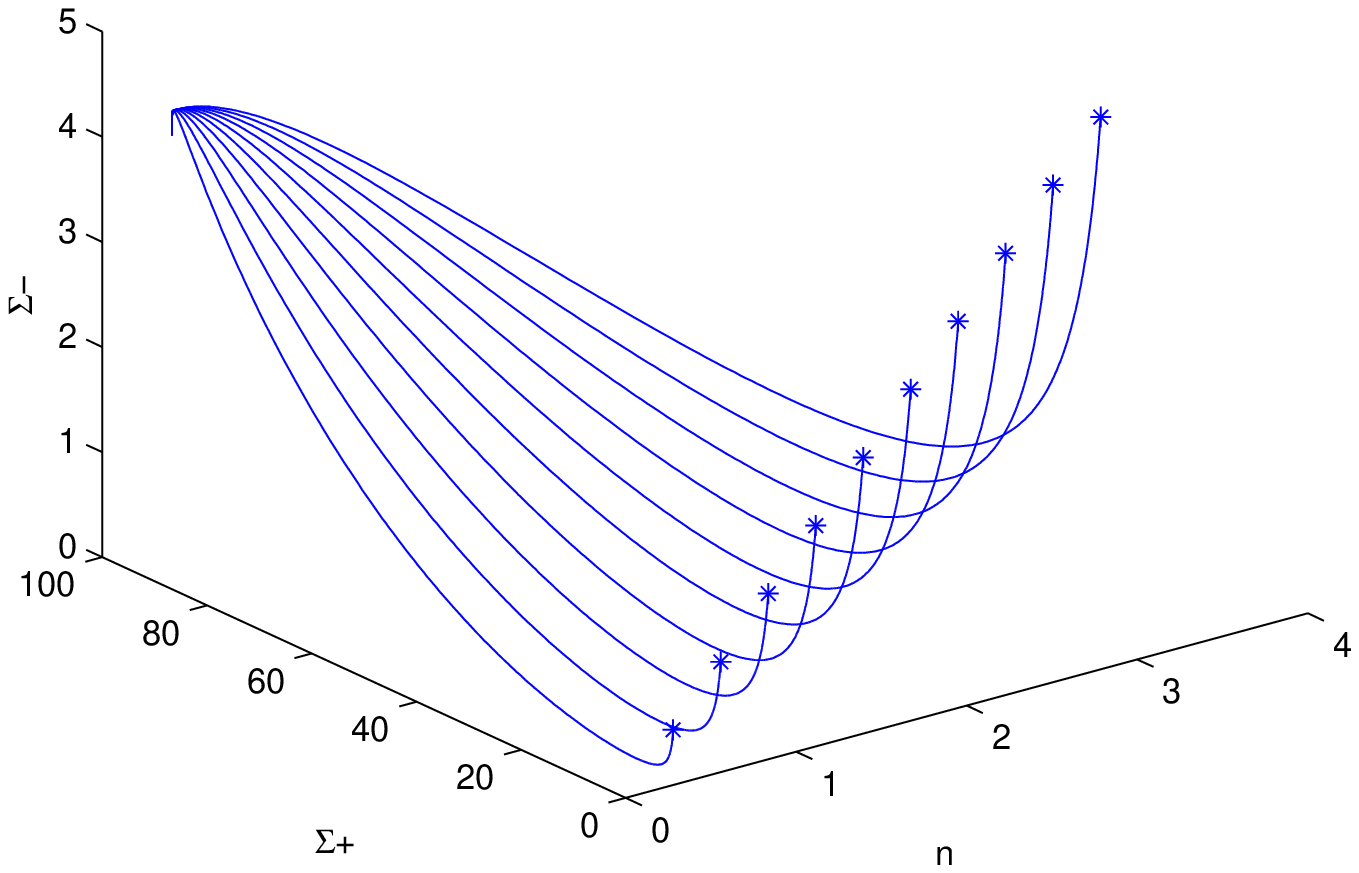}
\label{default}
\end{center}
\end{figure}

\newpage

\begin{figure}[h]
\begin{center}
\caption{Phase plot of the anisotropy and spatial curvature for $\xi_{0} = 0.1$, $\eta_{0} = 100$. }
\includegraphics{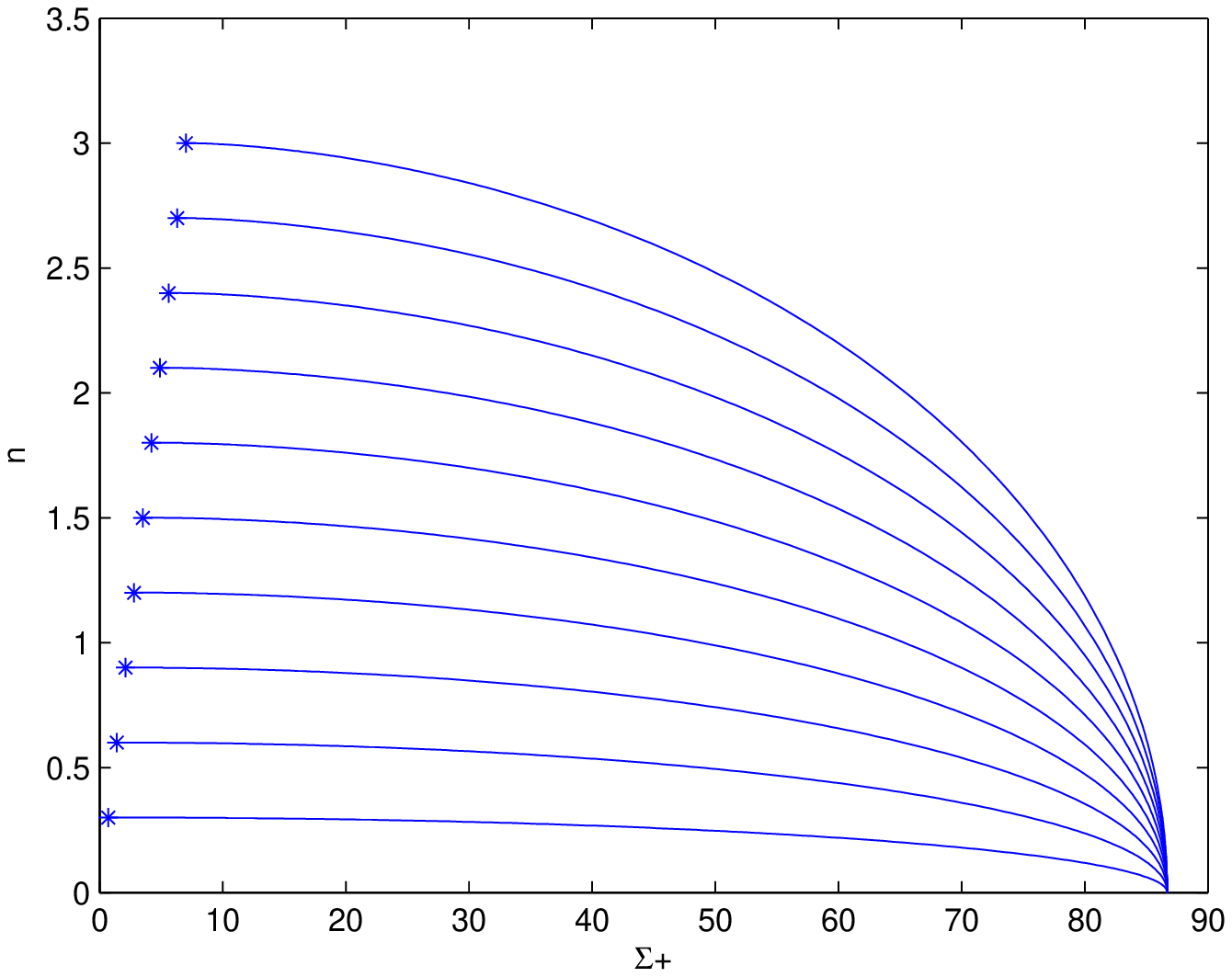}
\label{default}
\end{center}
\end{figure}

\newpage
\begin{figure}[h]
\begin{center}
\caption{Plots of the Energy Density and Anisotropy variables as functions of time, for $\xi_{0} = 0.1$, $\eta_{0} = 100$. One can see that $\Omega \rightarrow 0$, but $\Sigma_{-} \not \to 0$, and $\Sigma_{+} \not \to 0$, so the model does not isotropize at all in this case.}
\includegraphics{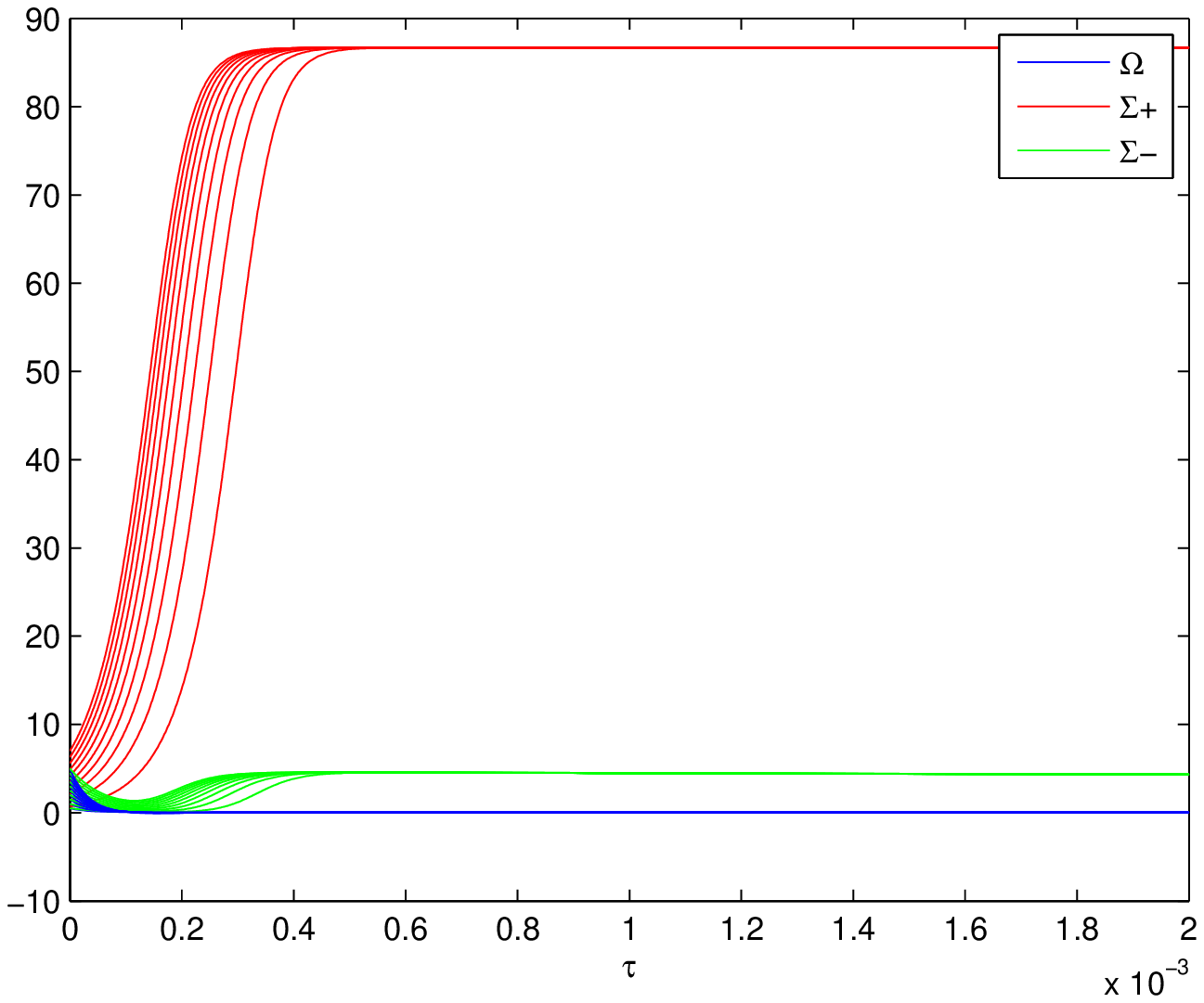}
\label{default}
\end{center}
\end{figure}

\newpage
\begin{figure}[h]
\begin{center}
\caption{Phase plot of the anisotropy in the Hubble flow for $\xi_{0} = 1$, $\eta_{0} = 0.1$. }
\includegraphics{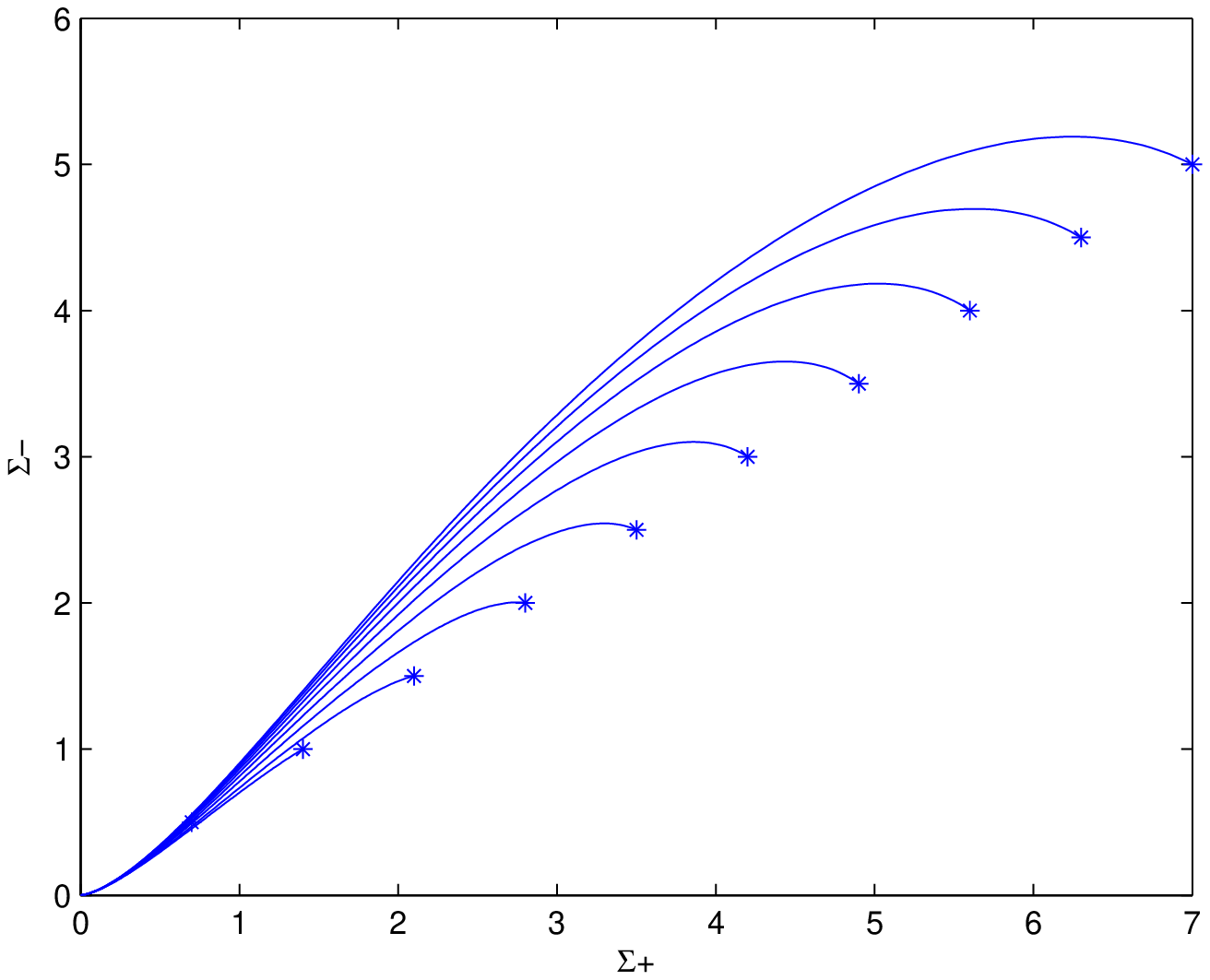}
\label{default}
\end{center}
\end{figure}

\newpage
\begin{figure}[h]
\begin{center}
\caption{Three-dimensional phase plot of the anisotropy in the Hubble flow and spatial curvature for $\xi_{0} = 1$, $\eta_{0} = 0.1$. }
\includegraphics{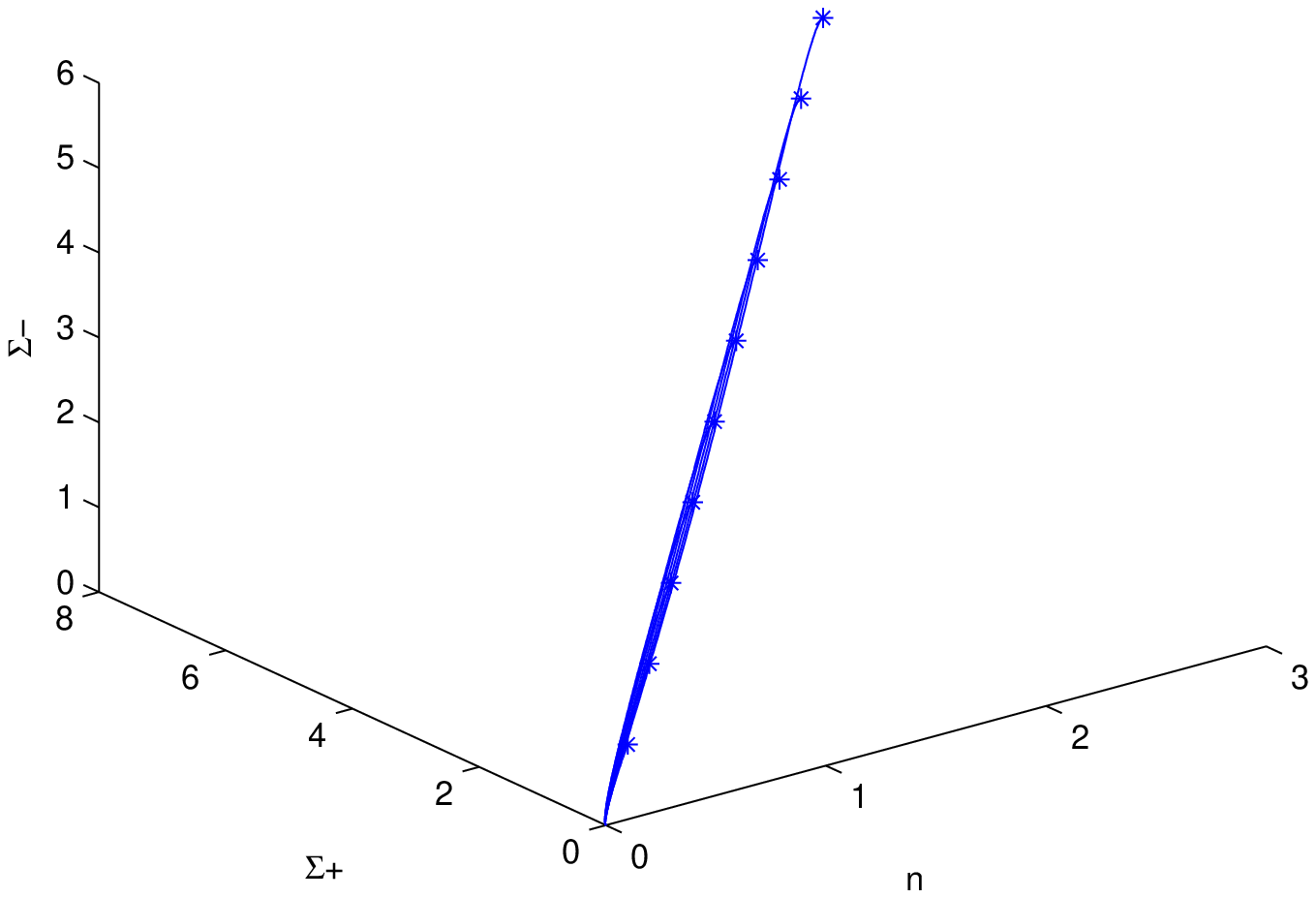}
\label{default}
\end{center}
\end{figure}

\newpage

\begin{figure}[h]
\begin{center}
\caption{Phase plot of the anisotropy and spatial curvature for $\xi_{0} = 1$, $\eta_{0} = 0.1$. }
\includegraphics{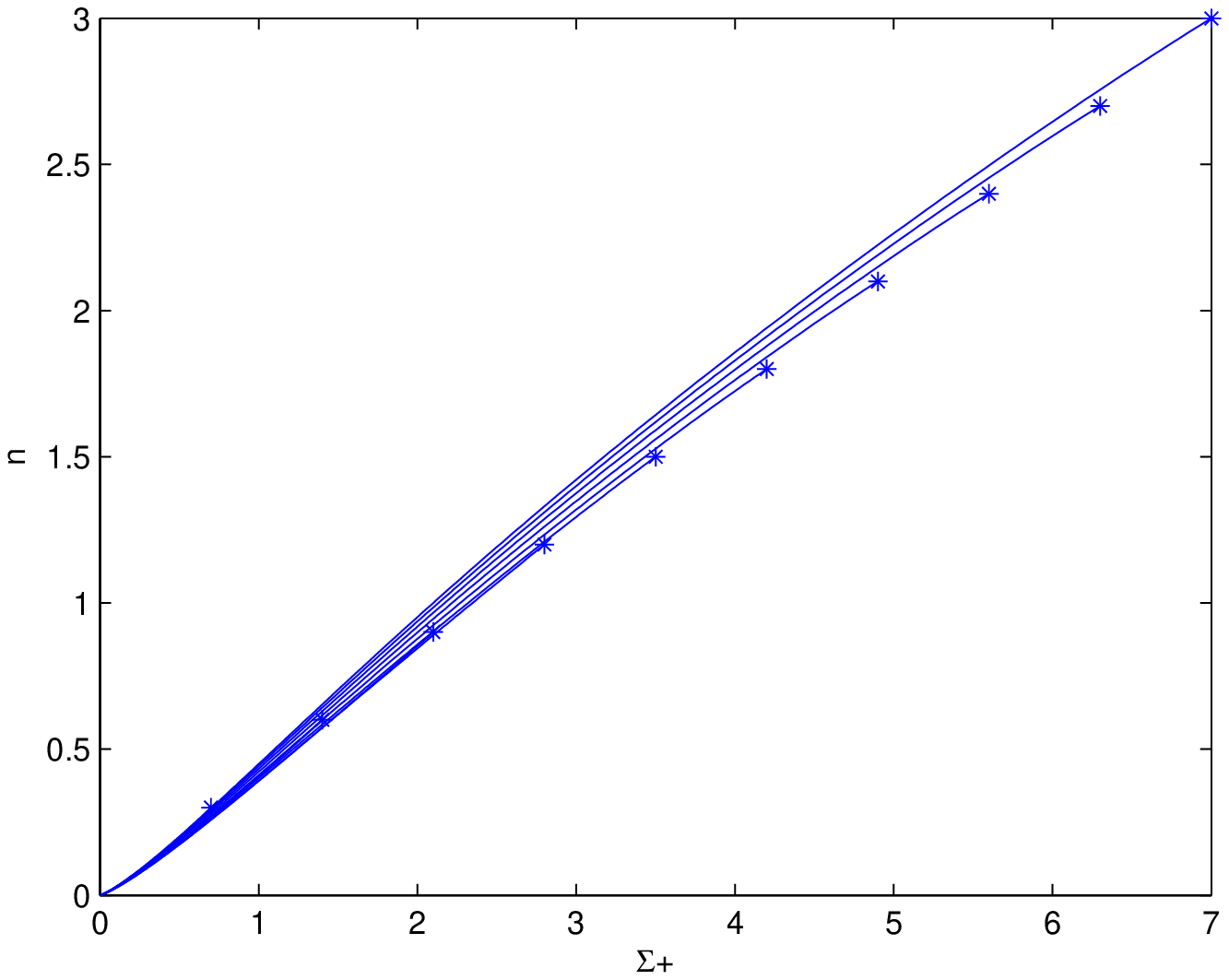}
\label{default}
\end{center}
\end{figure}

\newpage
\begin{figure}[h]
\begin{center}
\caption{Plots of the Energy Density and Anisotropy variables as functions of time, for $\xi_{0} = 1$, $\eta_{0} = 0.1$. One can see that $\Omega \rightarrow 0$,  $\Sigma_{-} \rightarrow 0$ and $\Sigma_{+} \rightarrow 0$, so the model does in fact isotropize.}
\includegraphics{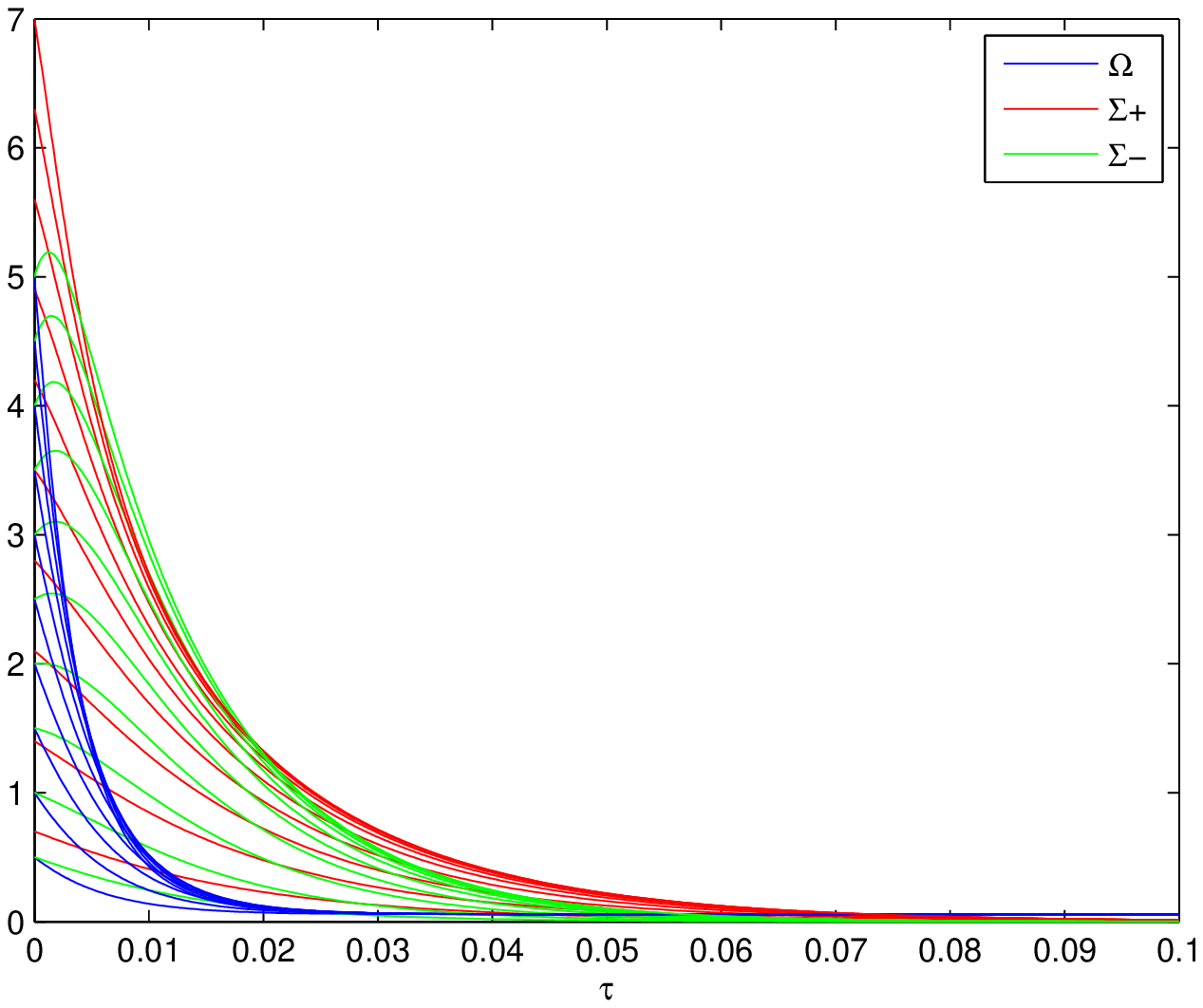}
\label{default}
\end{center}
\end{figure}

\newpage
\begin{figure}[h]
\begin{center}
\caption{Phase plot of the anisotropy in the Hubble flow for $\xi_{0} = 10$, $\eta_{0} = 0.1$. }
\includegraphics{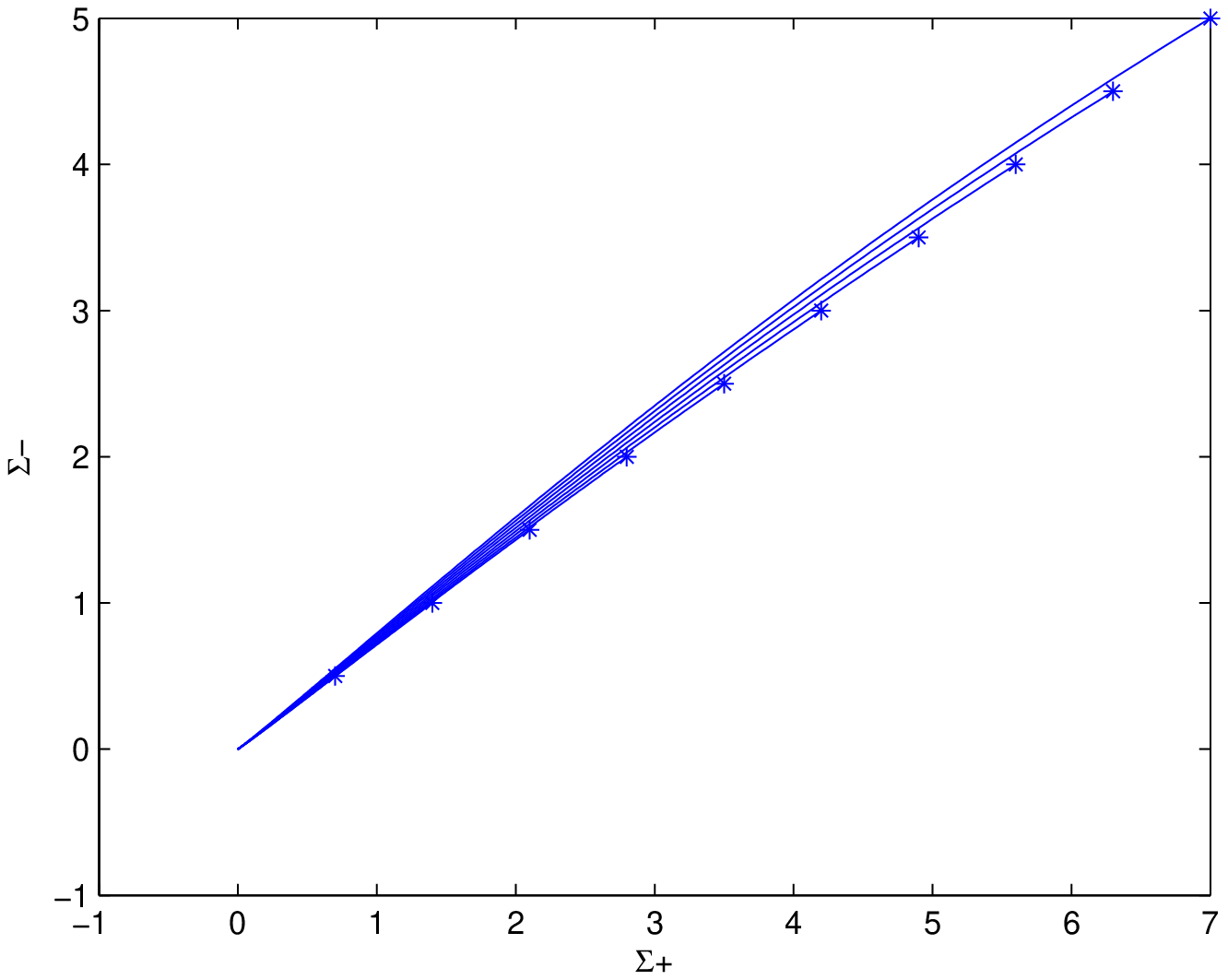}
\label{default}
\end{center}
\end{figure}

\newpage
\begin{figure}[h]
\begin{center}
\caption{Three-dimensional phase plot of the anisotropy in the Hubble flow and spatial curvature for $\xi_{0} = 10$, $\eta_{0} = 0.1$. }
\includegraphics{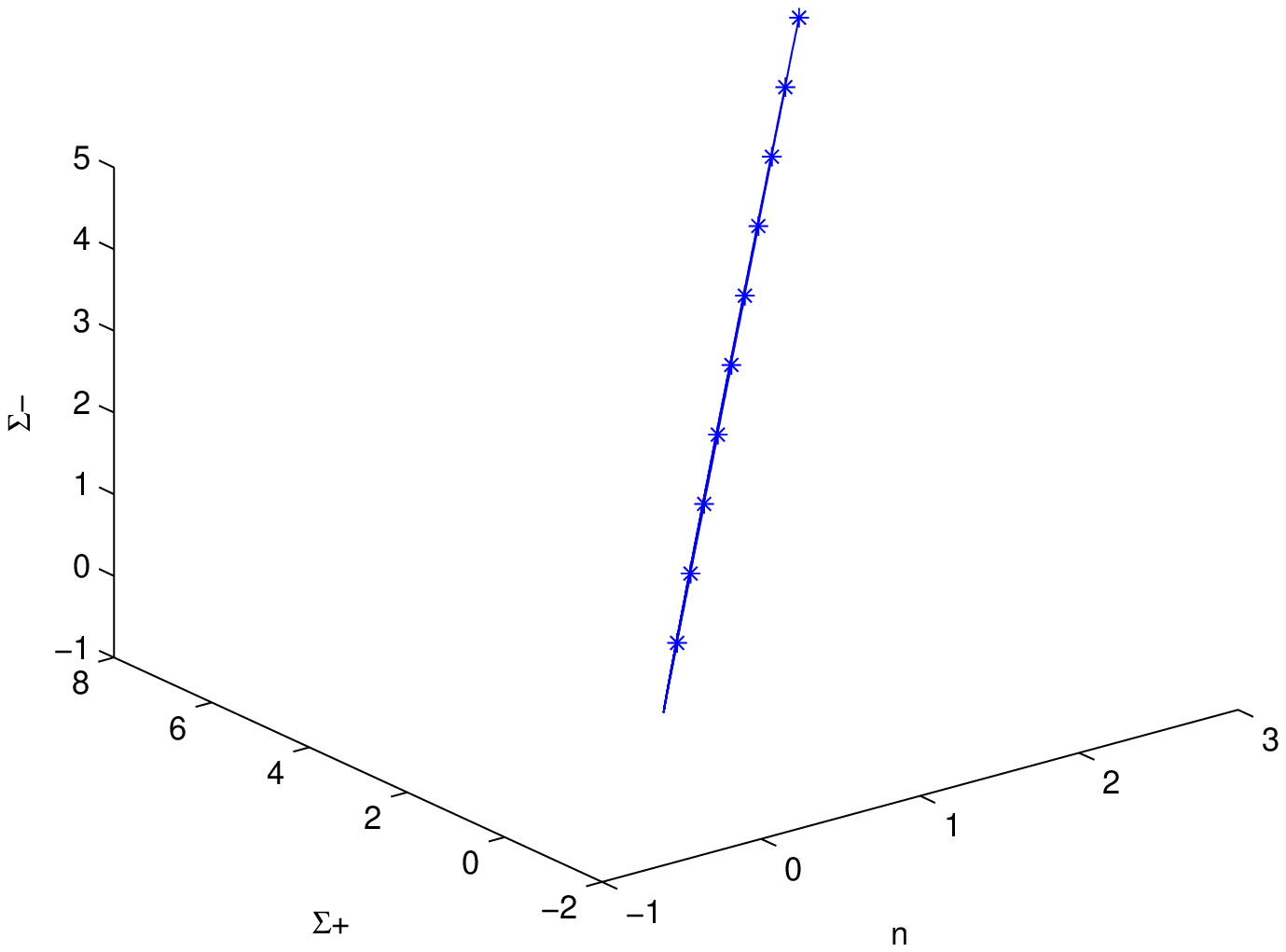}
\label{default}
\end{center}
\end{figure}

\newpage

\begin{figure}[h]
\begin{center}
\caption{Phase plot of the anisotropy and spatial curvature for $\xi_{0} = 10$, $\eta_{0} = 0.1$.}
\includegraphics{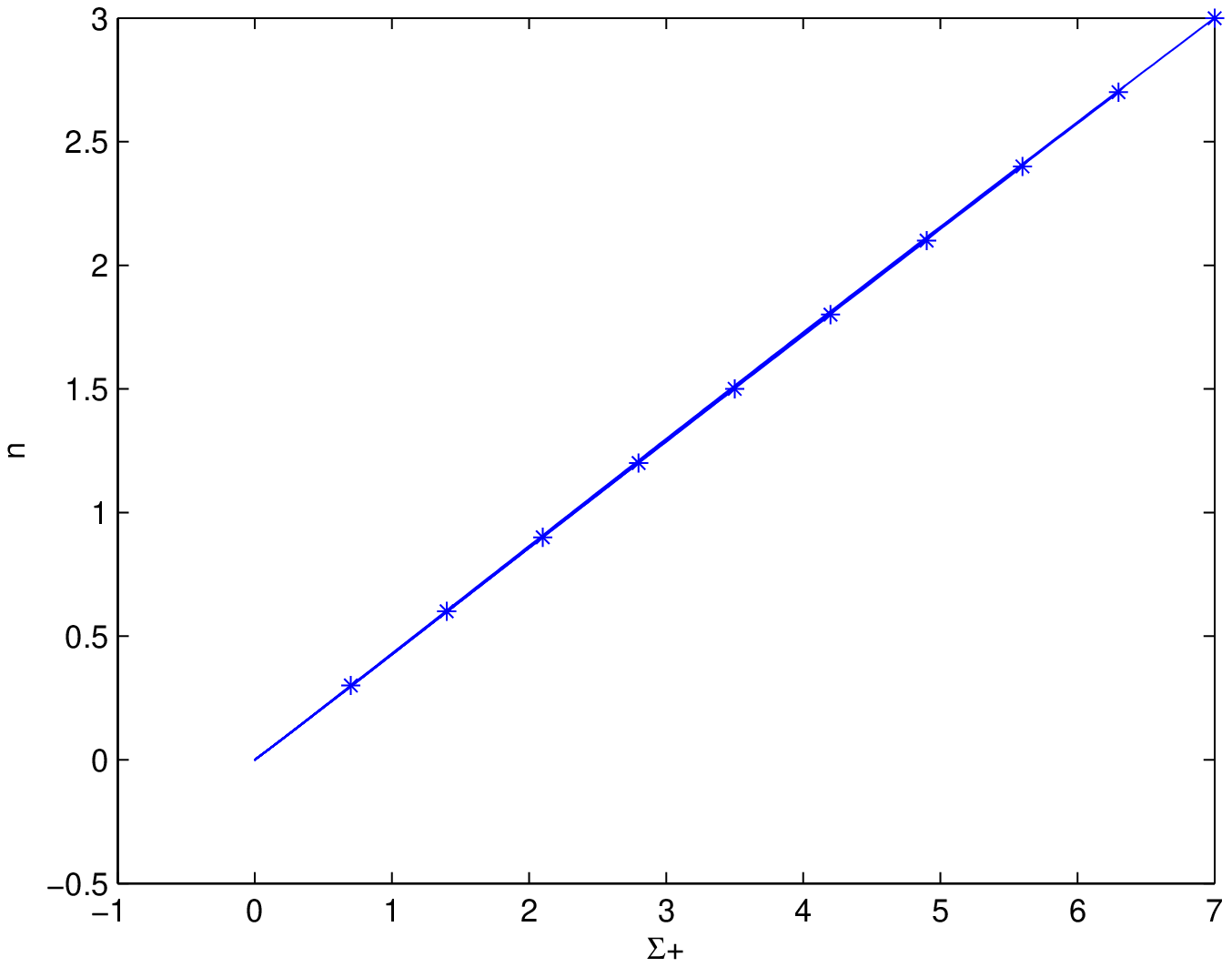}
\label{default}
\end{center}
\end{figure}

\newpage
\begin{figure}[h]
\begin{center}
\caption{Plots of the Energy Density and Anisotropy variables as functions of time, for $\xi_{0} = 10$, $\eta_{0} = 0.1$. One can see that $\Omega \rightarrow 0$,  $\Sigma_{-} \rightarrow 0$ and $\Sigma_{+} \rightarrow 0$, so the model does in fact isotropize.}
\includegraphics{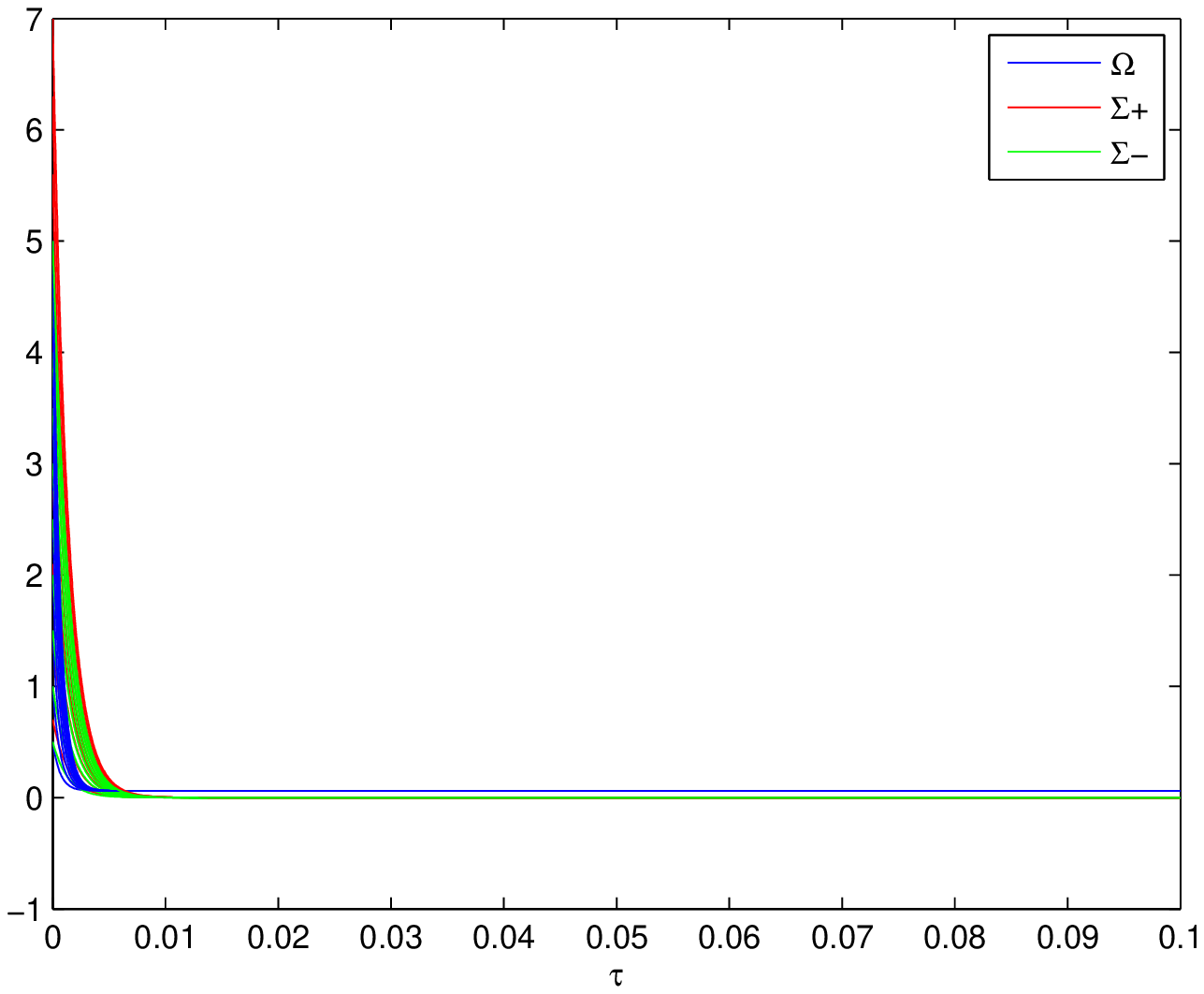}
\label{default}
\end{center}
\end{figure}

\newpage
\begin{figure}[h]
\begin{center}
\caption{Phase plot of the anisotropy in the Hubble flow for $\xi_{0} = 1$, $\eta_{0} = 0$.}
\includegraphics{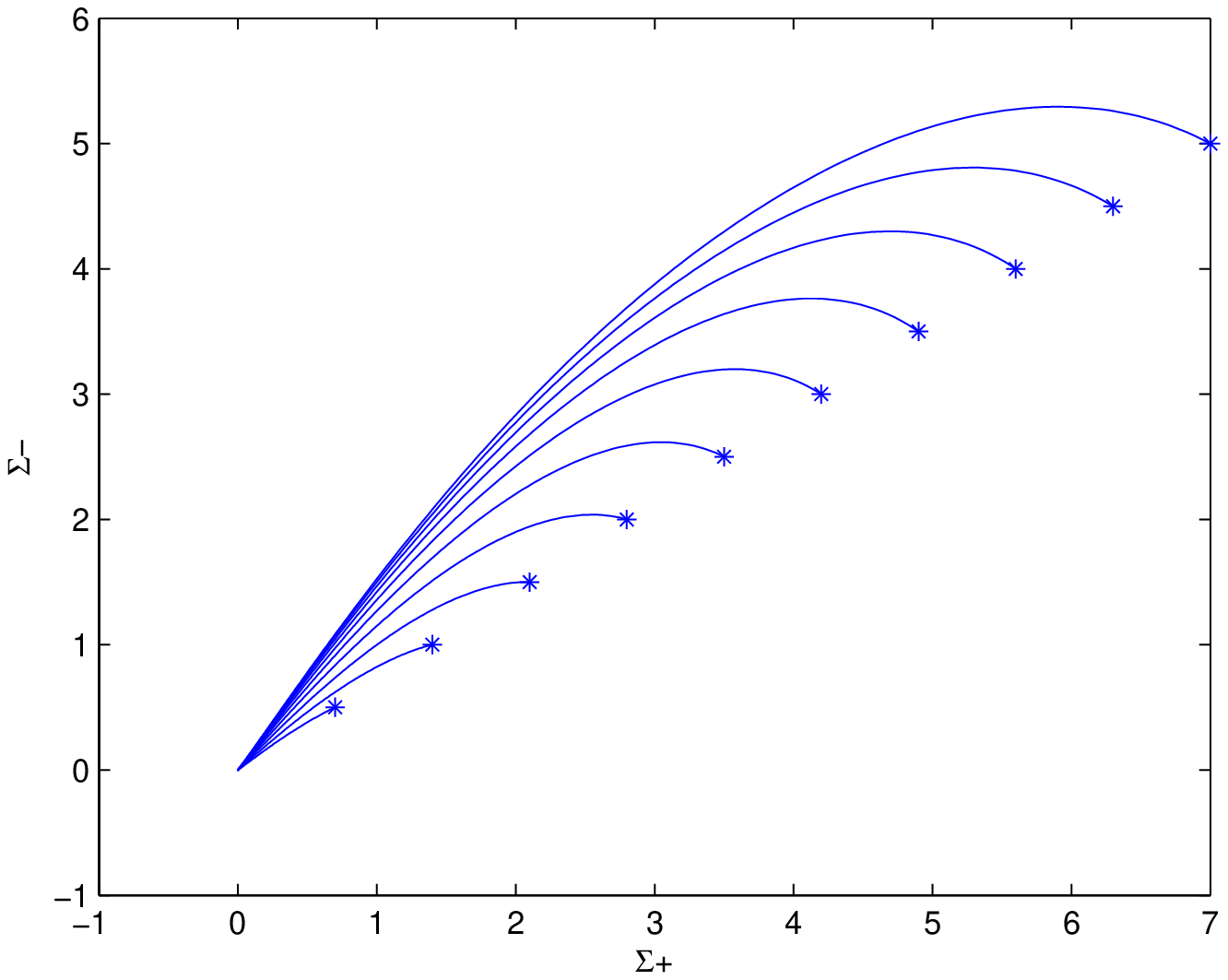}
\label{default}
\end{center}
\end{figure}

\newpage
\begin{figure}[h]
\begin{center}
\caption{Three-dimensional phase plot of the anisotropy in the Hubble flow and spatial curvature for $\xi_{0} = 1$, $\eta_{0} = 0$. }
\includegraphics{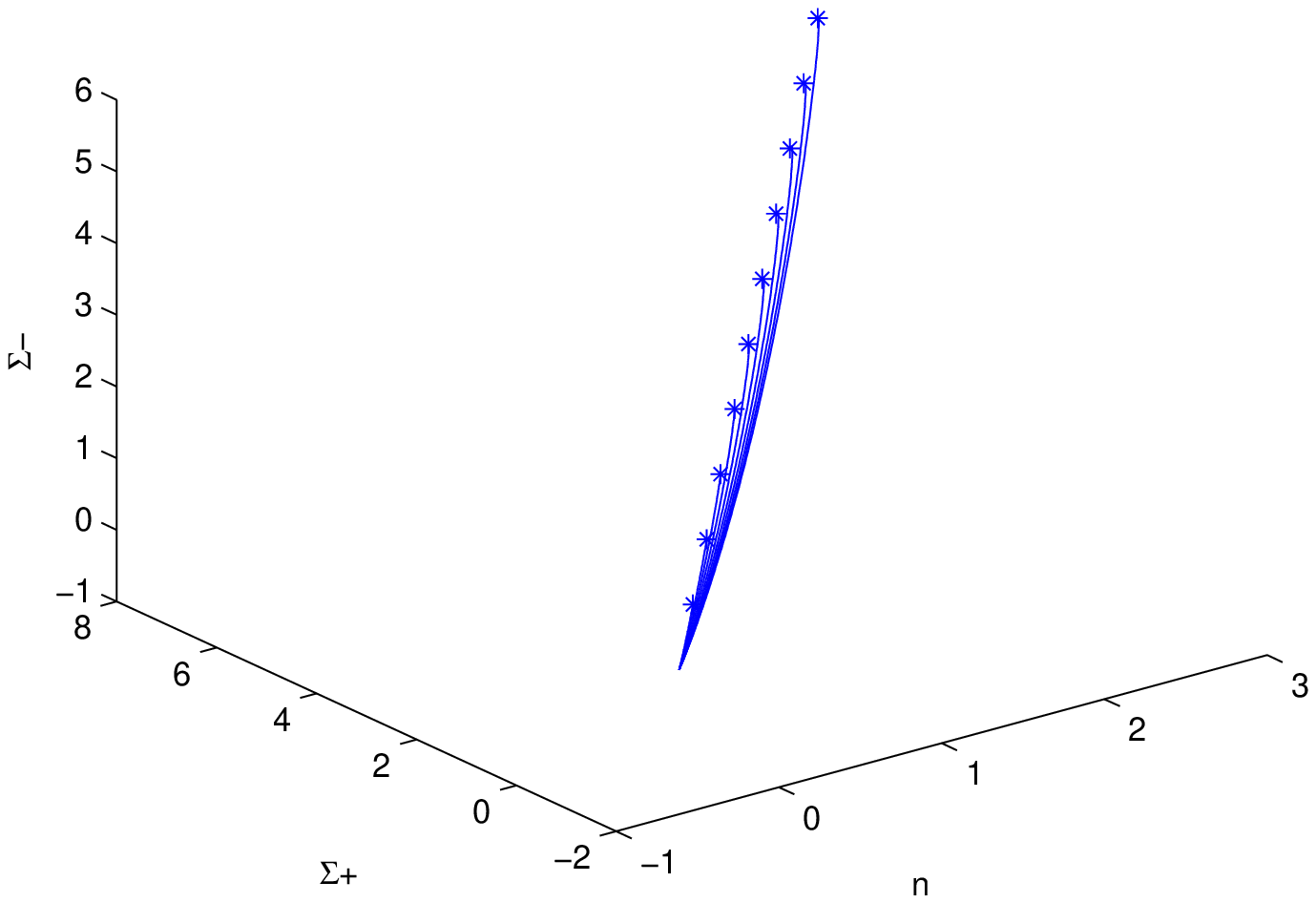}
\label{default}
\end{center}
\end{figure}

\newpage

\begin{figure}[h]
\begin{center}
\caption{Phase plot of the anisotropy and spatial curvature for $\xi_{0} = 1$, $\eta_{0} = 0$.}
\includegraphics{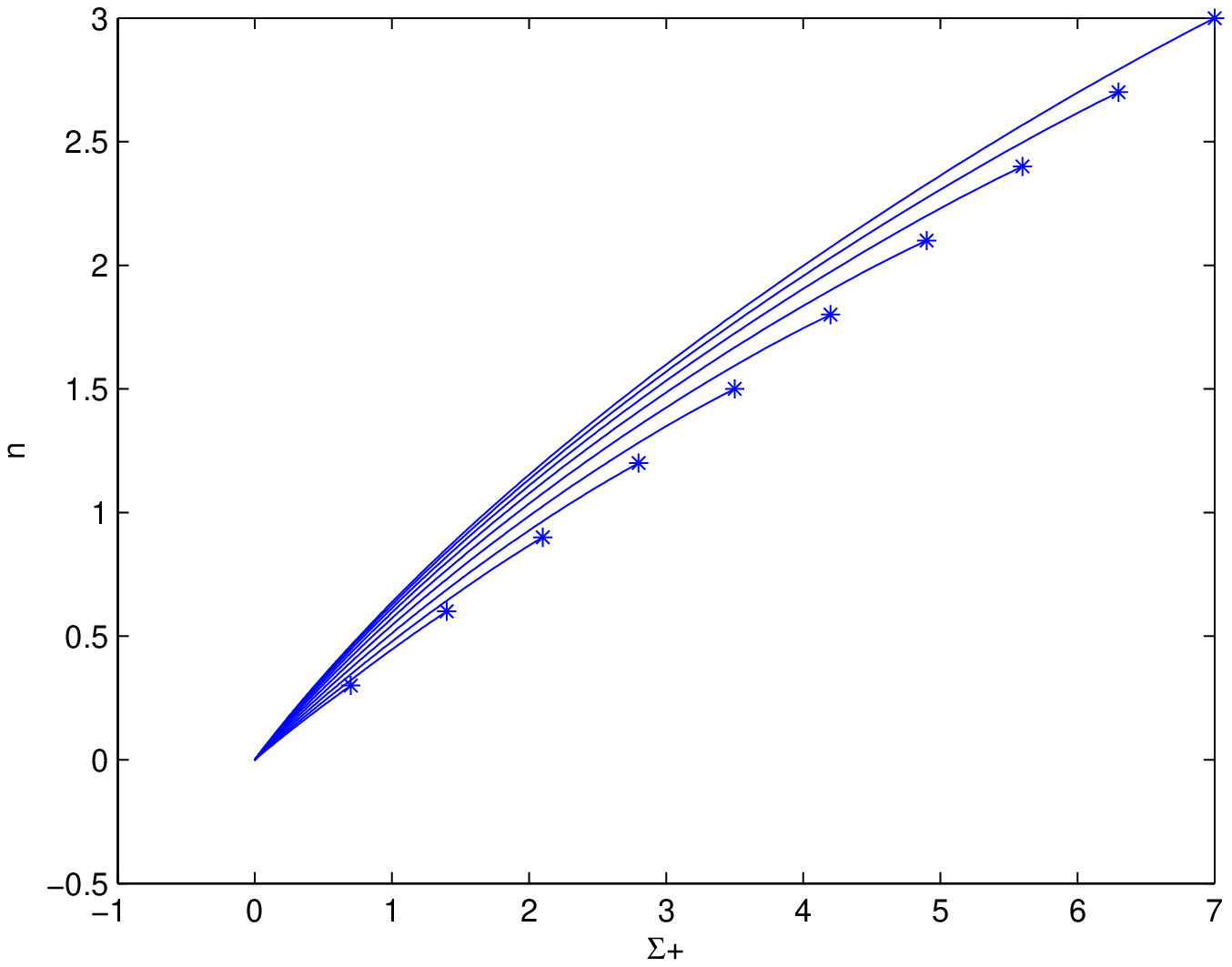}
\label{default}
\end{center}
\end{figure}

\newpage
\begin{figure}[h]
\begin{center}
\caption{Plots of the Energy Density and Anisotropy variables as functions of time, for $\xi_{0} = 1$, $\eta_{0} = 0$. One can see that $\Omega \rightarrow 0$,  $\Sigma_{-} \rightarrow 0$ and $\Sigma_{+} \rightarrow 0$, so the model does in fact isotropize.}
\includegraphics{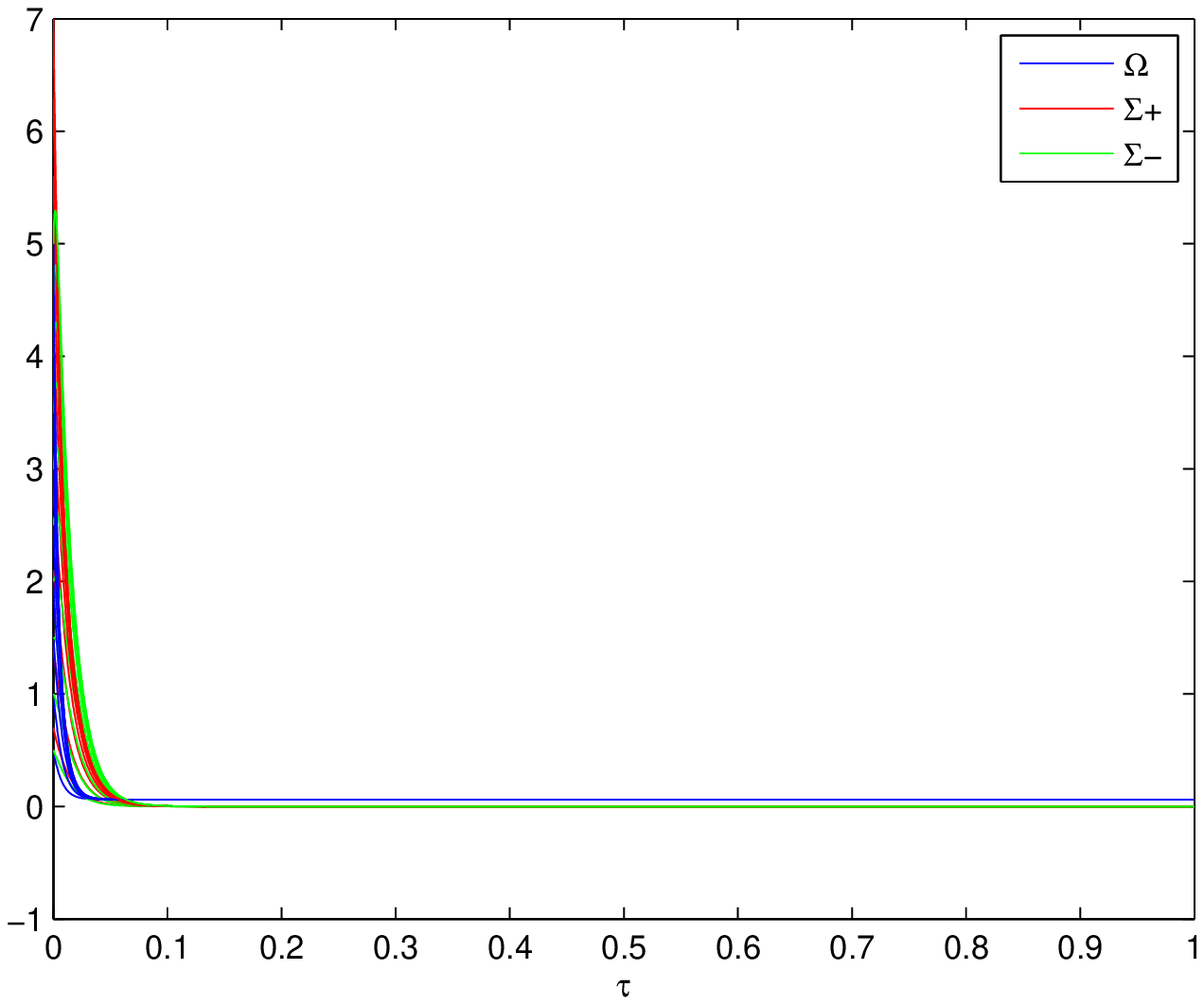}
\label{default}
\end{center}
\end{figure}

\newpage
\begin{figure}[h]
\begin{center}
\caption{Phase plot of the anisotropy in the Hubble flow for $\xi_{0} = 100$, $\eta_{0} = 0$.}
\includegraphics{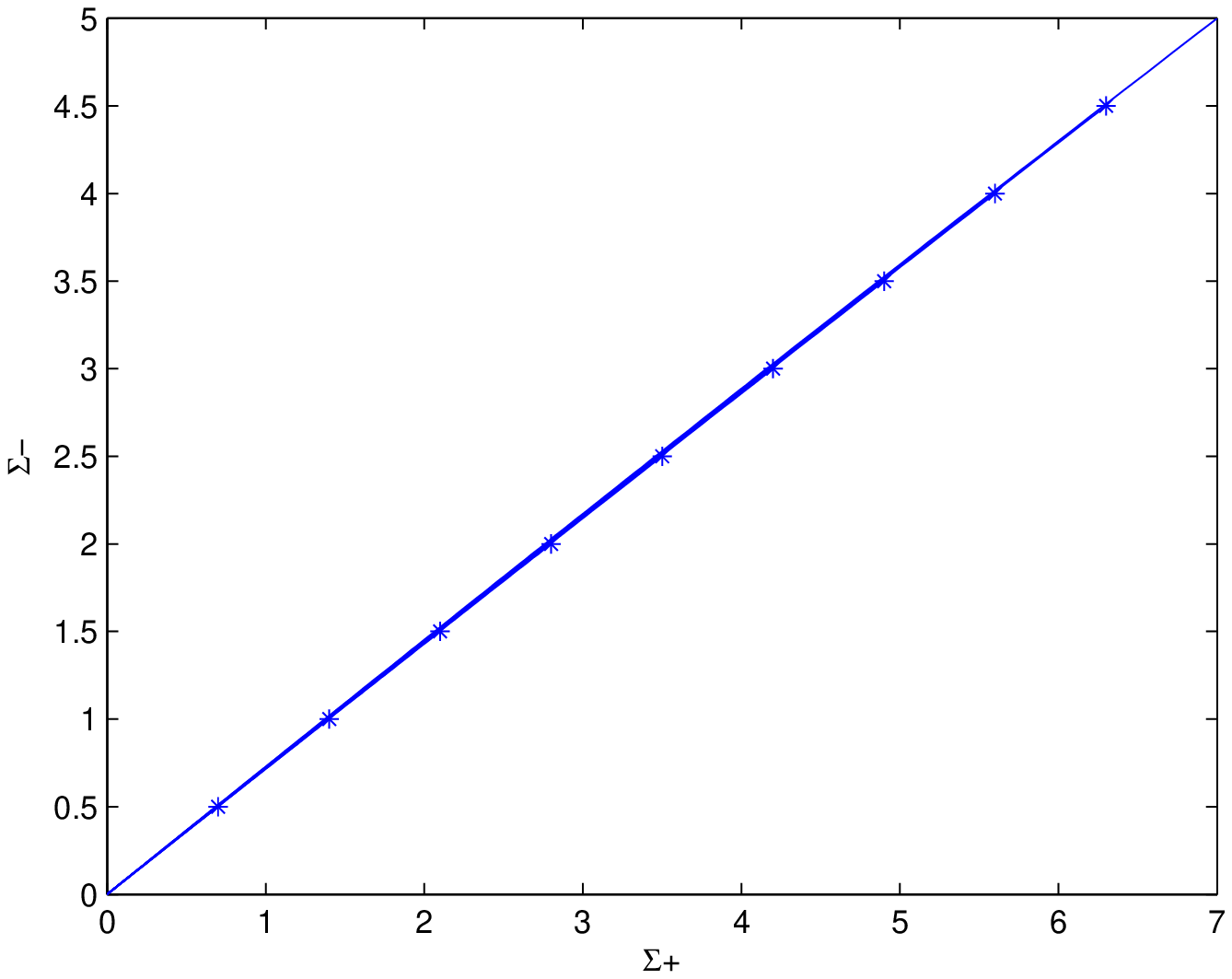}
\label{default}
\end{center}
\end{figure}

\newpage
\begin{figure}[h]
\begin{center}
\caption{Three-dimensional phase plot of the anisotropy in the Hubble flow and spatial curvature for $\xi_{0} = 100$, $\eta_{0} = 0$. }
\includegraphics{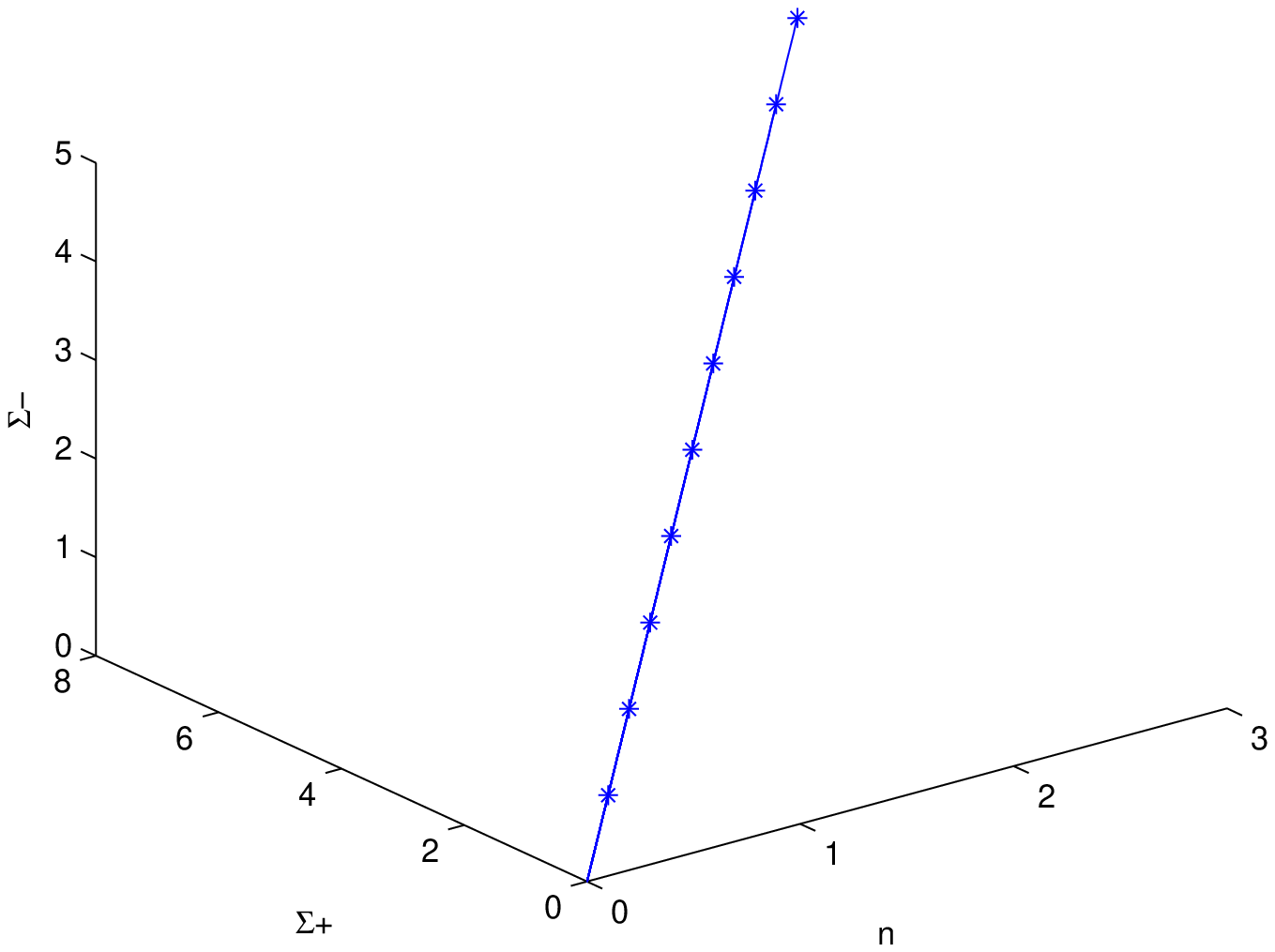}
\label{default}
\end{center}
\end{figure}

\newpage

\begin{figure}[h]
\begin{center}
\caption{Phase plot of the anisotropy and spatial curvature for $\xi_{0} = 100$, $\eta_{0} = 0$.}
\includegraphics{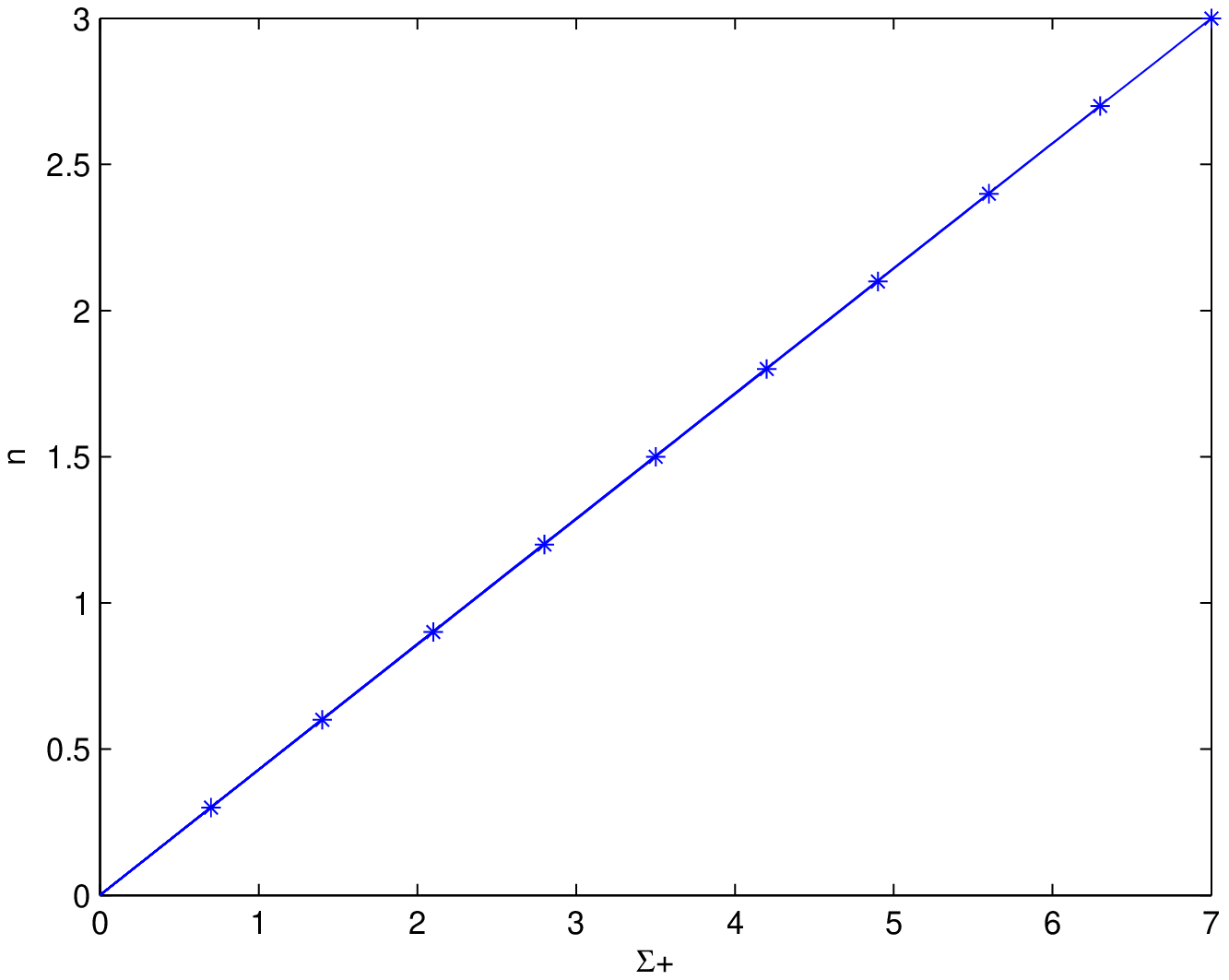}
\label{default}
\end{center}
\end{figure}

\newpage
\begin{figure}[h]
\begin{center}
\caption{Plots of the Energy Density and Anisotropy variables as functions of time, for $\xi_{0} = 100$, $\eta_{0} = 0$. One can see that $\Omega \rightarrow 0$,  $\Sigma_{-} \rightarrow 0$ and $\Sigma_{+} \rightarrow 0$, so the model does in fact isotropize.}
\includegraphics{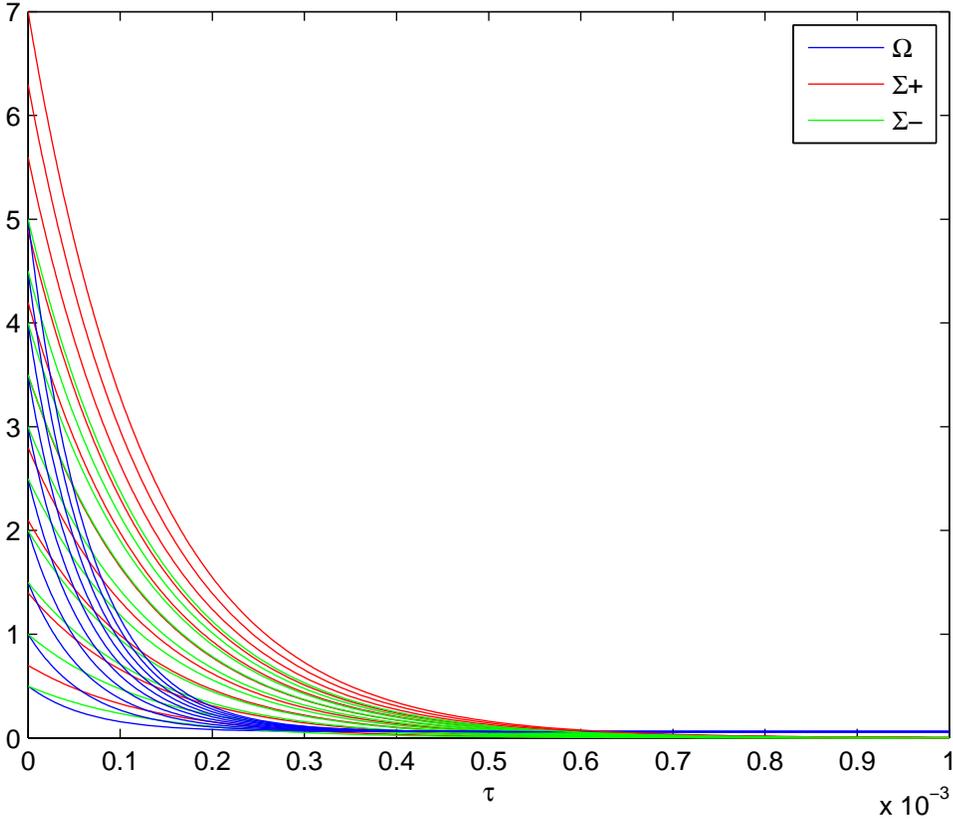}
\label{default}
\end{center}
\end{figure}

\newpage

\subsection{Discussion of Results}
First looking at Figs. (6.1-6.16), one can see that the cosmological model does not isotropize in all directions. In addition, from the phase plots of the dynamic variables, one can see that the dynamics of the model resemble that of \emph{plane-wave} solutions. Such behaviour was noted by Hewitt and Wainwright \cite{ellis} in their investigation of the dynamics of a non-tilted perfect-fluid Bianchi Type-IV model. In addition, the asymptotic limit of anisotropy of a Bianchi Type IV perfect fluid model was noted by Hervik, van den Hoogen and Coley \cite{hervikvan}. It is interesting to note that as the shear viscosity coefficient $\eta_{0}$ increased, the anisotropy in the cosmological model increased as well. However, looking at Figs. (6.17 - 6.24) where we introduced a bulk viscosity coefficient $\xi_{0}$ that was significantly larger than the shear viscosity coefficient caused the model to isotropize. What is additionally interesting is that as the bulk viscosity was increased, the relationship between the dynamical variables became increasingly linear, which hints at the notion that the dynamics of the transitionary period from the early-universe cosmology to the present-day cosmological models become simpler. Perhaps the most interesting point of all is that while the shear viscosity coefficient was set to zero in Figs. (6.25 - 6.32), with the bulk viscosity coefficient becoming larger and larger, the model still isotropized. This provides evidence that the presence of a significant bulk viscous pressure rather than a significant shear viscosity is necessary for such a cosmological model to isotropize. Similar results for other Bianchi models have been reported by van den Hoogen and Coley \cite{vdh}, Singh and Kale \cite{singhkale}, and Pradhan, Rai and Singh \cite{pradhan}. It is also important to note that no cosmological model that is spatially homogeneous and anisotropic and necessarily has a three-dimensional isometry group can evolve to a model with a six-dimensional isometry group. When one speaks of a model isotropizing, $\Sigma_{\pm} \rightarrow 0$ only asymptotically. However, any Bianchi model can become arbitrary close to that of the FLRW models as a consequence of asymptotic isotropization. It is also a strong implication that an isotropic universe necessarily implies the existence of an isotropic fluid. Therefore, one must conclude that the nature of our present-day spatially homogeneous and isotropic universe as described by the FLRW models is actually quite special indeed!

\section{On The Penrose-Hawking Singularity Theorem}
Up to this point, we have demonstrated that a Bianchi Type-IV model isotropizes with the significant presence of bulk viscous pressure. One can interpret our work so far as evolving an early-universe cosmological model towards the future. Now, we wish to analyze the reverse question by considering the physical past of a viscous model. In particular, we are interested to know whether such a universe could have emerged from a physical singularity point. 

We begin by considering The Raychaudhuri equation as given in Eq. (4.42):

\begin{equation*}
\dot{\theta} = -\frac{1}{2}\kappa\left(\mu + 3p\right) + \frac{3}{2}\kappa \xi \theta - 2 \sigma^2 - \frac{1}{3}\theta^2
\end{equation*}

If we assume the strong energy condition (Eq. (3.31)):

\begin{eqnarray}
R^{ab}u_{a}u_{b} &\geq& 0 \nonumber \\
\Rightarrow \frac{1}{2}\kappa \left(\mu + 3p\right) - \frac{3}{2}\kappa \xi \theta &\geq& 0 \mbox{ (Eq. (4.39))} \nonumber \\
\Rightarrow -\frac{1}{2}\kappa \left(\mu + 3p\right) + \frac{3}{2}\kappa \xi \theta &\leq& 0 
\end{eqnarray}

Here, we assume the equation of state:

\begin{equation}
\xi = \frac{C}{\theta}, \mbox{ $C \in \mathbb{R^{+}}$}
\end{equation}

Hence, Eq. (7.1) becomes:
\begin{equation}
 -\frac{1}{2}\kappa \left(\mu + 3p\right) + \frac{3}{2}\kappa C \leq 0 
\end{equation}

Since $\sigma^2 \geq 0$, we have from The Raychaudhuri equation:

\begin{equation}
\dot{\theta}  + \frac{1}{2}\kappa\left(\mu + 3p\right) -  \frac{3}{2}\kappa C+ \frac{1}{3}\theta^2 = -2\sigma^2 \leq 0
\end{equation}

Using Eq. (7.3), this becomes:

\begin{eqnarray}
\dot{\theta}  + \frac{1}{2}\kappa\left(\mu + 3p\right) -  \frac{3}{2}\kappa C+ \frac{1}{3}\theta^2 &\leq& 0 \nonumber \\
\Rightarrow \dot{\theta} + \frac{1}{3}\theta^2 \leq -\frac{1}{2}\kappa\left(\mu + 3p\right) + \frac{3}{2}\kappa C &\leq& 0 \nonumber \\
\Rightarrow \dot{\theta} + \frac{1}{3}\theta^2 &\leq& 0 \nonumber \\
\Rightarrow \dot{\theta} &\leq& -\frac{1}{3}\theta^2
\end{eqnarray}

We first proceed assuming a strict inequality in Eq. (7.5). That is:

\begin{eqnarray}
 \dot{\theta} &\leq& -\frac{1}{3}\theta^2 \nonumber \\
 \Rightarrow \frac{d}{dt} \left(\frac{1}{\theta}\right) &\geq& \frac{1}{3} \nonumber \\
 \Rightarrow \frac{1}{\theta(t)} &\leq& \frac{1}{\theta_{i}} + \frac{t}{3}
\end{eqnarray}

We have assumed the initial condition: $\theta_{i} = \theta(0)$. In addition, we also assume that $t \leq 0$. If we have an expanding universe, then, necessarily, the congruence of geodesics is expanding at $t=0$, which implies that $\theta_{i} > 0$. This in turn implies that:

\begin{equation}
\forall \mbox{ } t \in \mathbb{R}, \exists \mbox{ } t_{sing} \mbox{ such that: } \frac{1}{\theta(t_{sing})} = 0
\end{equation}

That is, the function $1/\theta(t)$ must have passed through zero at some time, $t_{sing}$. It then follows that:
\begin{equation}
 \frac{1}{\theta(t_{sing})} = 0 \Rightarrow \theta(t_{sing}) = \infty
\end{equation}

This tells us that there was a physical singularity at the time $t_{sing}$, where the expansion scalar was infinite. 

In the case of a strict equality in Eq. (7.6), we can be slightly more precise by employing the intermediate value theorem. We begin by defining the function:

\begin{equation}
F(t) \equiv \frac{1}{\theta(t)} = \frac{1}{\theta_{i}} + \frac{t}{3}
\end{equation}

We wish to find an interval $[a,b]$ such that $F(t) = 0$ on this interval. Suppose that $F: [a,b] \rightarrow \mathbb{R}$ is continuous and that $u \in \mathbb{R}$ such that: $F(a) < u < F(b)$. Then, by the intermediate value theorem, $\exists \mbox{ } c \in [a,b] \mbox{ such that } F(c) = u$. Clearly, we wish to find a $c$ such that $F(c) = 0$, that is, $u = 0$. This means that $u = 0$ satisfies:

\begin{equation}
F(a) = \frac{1}{\theta_{i}} + \frac{a}{3} < 0 <  F(b) = \frac{1}{\theta_{i}} + \frac{b}{3}
\end{equation}

This implies that:

\begin{equation}
\frac{1}{\theta_{i}} + \frac{a}{3} < 0 < \frac{1}{\theta_{i}} + \frac{b}{3}
\end{equation}

Since in an expanding universe, $\theta_{i} > 0$, we need to choose $[a,b] \in \mathbb{R}$ such that Eq. (7.11) is satisfied. Any choice for which $a < 0$, and $b > 0$ will satisfy this relationship. In this interval:

\begin{equation}
F(c) = \frac{1}{\theta_{i}} + \frac{c}{3} = 0 \Rightarrow |c| = \frac{3}{\theta_{i}}
\end{equation}

Therefore, $F(t) = \frac{1}{\theta(t)}$ does indeed equal zero at $c\in[a,b]$ such that $a<0, \mbox{} b>0$ and $c = \frac{3}{\theta_{i}}$. Which means that at some time $t = c$, the expansion scalar was infinite, which indicates that there was a singularity at $t = c$. Note that, in the above analysis, $t=0$ is some arbitrary time, and \emph{not} the time of the singularity. It can be thought of as the initial time, and $\theta_{i}$ is the value of the expansion scalar at this time. The arguments given above revolve around the idea that the equation for $\theta(t)$ is integrated backwards upon which noting that there exists a time $t_{sing} < 0$ which corresponds to the singularity. It then so happens that the singularity occurs at $t = t_{s} < 0$.
$\qedhere$
\\
A more formal version of the arguments given above could be given as follows. We essentially follow the methods in \cite{salas}.

\newtheorem{lem}{Lemma}
\begin{lem}
Let $F(t) \equiv \frac{1}{\theta(t)} = \frac{1}{\theta_{i}} + \frac{t}{3}$ be continuous on [a,b]. If $F(a) < 0 < F(b) \vee F(b) < 0 < F(a) \Rightarrow \exists \mbox{ } d, \mbox{ such that } a < d < b$ for which $F(d) = 0$.
\end{lem}

\begin{proof}
Suppose that $F(a) < 0 < F(b)$. Since $F(a) < 0$, by the continuity of $F(t)$, $\exists \mbox{ } h \in \mathbb{R} \mbox{ such that } F < 0 \mbox{ on } [a,h)$. Define:

\begin{equation}
d = \sup \{h: F(t) < 0 \mbox{ on } [a,h)\}
\end{equation}

It is clear that $d \leq b$. One cannot have $F(d) > 0$, since this would mean that $F > 0$ on some interval to the left of $d$, which is a contradiction based on our definition of $d$ above. Additionally, we cannot have that $d = b$, so, $d < b$. We also cannot have $F(d) < 0$, since this would imply that there would be an interval $[a,t)$, with $t > d$ which would in turn imply that $F(t) < 0$. This, once again, contradicts the definition of $d$. Therefore, it must be that $F(d) = 0$. 
\end{proof}

All of the aforementioned arguments point to a general theorem on singularities first given by Penrose and Hawking \cite{penrose} \cite{hervik} \cite{ellis3}. More or less, we can say that if matter as described by the energy-momentum tensor satisfies the strong energy condition, and there exists a $\alpha \in \mathbb{R}^{+}$ such that $\theta(t) > \alpha$, everywhere in the past of a specific spatial slice, then, there exists a singularity point where all geodesics end. As noted by Hervik and Gr{\o}n, spacetimes can have singularities even if the strong energy condition is violated. If, however, we assume that the strong energy condition must be satisfied, then, for a viscous universe, Eq. (7.3) implies conditions on $C$ for the existence of a past singularity:

\begin{eqnarray}
-\frac{1}{2}\kappa \left(\mu + 3p\right) &\leq& -\frac{3}{2}\kappa C \nonumber \\
\Rightarrow \frac{1}{2}\left(\mu + 3p\right) &\geq& \frac{3}{2}C \nonumber \\
\Rightarrow \frac{1}{3} \left(\mu + 3p\right) &\geq& C
\end{eqnarray}

Since this relationship was derived independent of any specific Bianchi type, one can conclude that for a past singularity to exist in an universe described by an energy-momentum tensor of the form Eq. (3.16), with an equation of state of the form Eq.(7.2),  Eq. (7.14) must strictly hold. A diagram of an expanding congruence of geodesics from an initial singularity point is shown in Fig. (7.1) below.

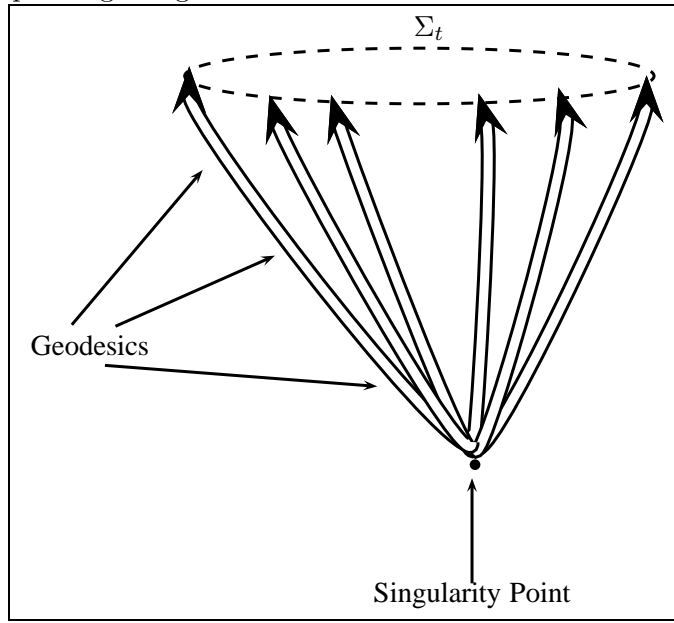
\begin{figure}[h]
\begin{center}
\caption{Expanding Congruence of Geodesics From an Initial Singularity Point}
\fbox{
\scalebox{1} 
{
\begin{pspicture}(0,-3.978047)(8.227344,3.978047)
\psellipse[linewidth=0.04,linestyle=dashed,dash=0.16cm 0.16cm,dimen=outer](5.0873437,3.1496484)(3.14,0.39)
\usefont{T1}{ptm}{m}{n}
\rput(5.2187986,3.7846484){$\Sigma_{t}$}
\psbezier[linewidth=0.04,doubleline=true,doublesep=0.12,arrowsize=0.05291667cm 2.0,arrowlength=1.4,arrowinset=0.4]{->}(5.82493,-1.7203516)(5.82493,-2.5203516)(8.107991,2.3596485)(8.107991,3.1596484)
\psbezier[linewidth=0.04,doubleline=true,doublesep=0.12,arrowsize=0.05291667cm 2.0,arrowlength=1.4,arrowinset=0.4]{->}(5.807344,-1.5803516)(5.804293,-2.3403516)(6.0903435,2.1396484)(5.8738465,2.9196484)
\psbezier[linewidth=0.04,doubleline=true,doublesep=0.12,arrowsize=0.05291667cm 2.0,arrowlength=1.4,arrowinset=0.4]{->}(5.827344,-1.6803516)(5.8652487,-2.5003517)(2.0273438,2.4796484)(2.0273438,3.2796485)
\usefont{T1}{ptm}{m}{n}
\rput(0.71473634,-0.43535155){Geodesics}
\psline[linewidth=0.04cm,arrowsize=0.05291667cm 2.0,arrowlength=1.4,arrowinset=0.4]{->}(0.46734375,-0.16035156)(2.2073438,1.8996484)
\psline[linewidth=0.04cm,arrowsize=0.05291667cm 2.0,arrowlength=1.4,arrowinset=0.4]{->}(1.0473437,-0.18035156)(3.1673439,0.75964844)
\psline[linewidth=0.04cm,arrowsize=0.05291667cm 2.0,arrowlength=1.4,arrowinset=0.4]{->}(0.90734375,-0.70035154)(4.5273438,-0.98035157)
\psdots[dotsize=0.142](5.827344,-2.0203516)
\psline[linewidth=0.04cm,arrowsize=0.05291667cm 2.0,arrowlength=1.4,arrowinset=0.4]{->}(5.787344,-3.6003516)(5.787344,-2.2003515)
\usefont{T1}{ptm}{m}{n}
\rput(5.7863674,-3.7553515){Singularity Point}
\psbezier[linewidth=0.04,doubleline=true,doublesep=0.12,arrowsize=0.05291667cm 2.0,arrowlength=1.4,arrowinset=0.4]{->}(5.807344,-1.7603515)(5.844293,-2.4603515)(4.1303434,2.1196485)(3.9138465,2.8996484)
\psbezier[linewidth=0.04,doubleline=true,doublesep=0.12,arrowsize=0.05291667cm 2.0,arrowlength=1.4,arrowinset=0.4]{->}(5.827344,-1.7603515)(5.824293,-2.4803514)(7.1503434,2.2196484)(6.9338465,2.9996483)
\psbezier[linewidth=0.04,doubleline=true,doublesep=0.12,arrowsize=0.05291667cm 2.0,arrowlength=1.4,arrowinset=0.4]{->}(5.787344,-1.7203516)(5.787344,-2.3403516)(3.3103437,2.1196485)(3.0938463,2.8996484)
\end{pspicture} 
}

}
\label{geodesics}
\end{center}
\end{figure}

\newpage
\section{Conclusions}
In this thesis, we were interested in formulating a viscous model of the universe based on The Bianchi Type IV algebra. We first began by considering a congruence of fluid lines in spacetime, upon which, analyzing their propagation behaviour, we rigorously derived the famous Raychaudhuri equation, but, in the context of viscous fluids. We then went through in great detail the topological and algebraic structure of a Bianchi Type IV algebra, by which we derived the corresponding structure and constraint equations. From this, we looked at the Einstein field equations in the context of orthonormal frames, and derived the resulting dynamical equations: The \emph{Raychaudhuri Equation}, \emph{generalized Friedmann equation}, \emph{shear propagation equations}, and a set of non-trivial constraint equations. We showed that for cases in which the bulk viscous pressure is significantly larger than the shear viscosity, this cosmological model isotropizes asymptotically to the present-day universe. We finally concluded by discussing The Penrose-Hawking singularity theorem, and showed that the viscous universe under consideration emerged from a past singularity point.

%
\newpage

\bibliographystyle{plain} 
\bibliography{sources} 

\end{document}